\documentclass[12pt]{article}
\usepackage{amsmath,amssymb}
\usepackage{colortbl} 
\usepackage{tabularx}
\usepackage{quantikz}
\providecommand{\setwiretype}[1]{}
\usepackage{setspace}
\usepackage{tikz-cd}
\usepackage{verbatim}
\usepackage{subcaption}
\usetikzlibrary{external}
\usetikzlibrary{arrows.meta,decorations.pathreplacing}
\usepackage{algorithm}
\usepackage{algpseudocode}
\usepackage{placeins}
\algrenewcommand{\algorithmicrequire}{\textbf{Require:}}
\algrenewcommand{\algorithmicensure}{\textbf{Ensure:}}
\algrenewcommand{\algorithmicindent}{1em}
\newcommand{\routine}[1]{\textsc{#1}}
\newcommand{\Sample}{\routine{Sample}}
\usepackage{braket}
\usepackage{enumerate}
\usepackage{hyperref}
\hypersetup{hidelinks}
\usepackage{mathtools}
\numberwithin{equation}{section}
\usepackage{graphicx}
\usepackage{booktabs}
\usepackage{xcolor}

\usepackage{amsthm}

\usepackage[margin=1in]{geometry}
\newcommand{\be}{\begin{eqnarray} \begin{aligned}}
\newcommand{\ee}{\end{aligned} \end{eqnarray} }
\newcommand{\benn}{\begin{eqnarray*} \begin{aligned}}
\newcommand{\eenn}{\end{aligned} \end{eqnarray*}}

\newcommand*{\bbN}{\mathbb{N}}

\newcommand*{\bbC}{\mathbb{C}}

\newcommand*{\cC}{\mathcal{C}}
\newcommand*{\sU}{\mathsf{U}}
\newcommand*{\tsU}{\tilde{\sU}}
\newcommand*{\gateLocations}[1]{\mathsf{Loc}(#1)}
\newcommand*{\numCuts}{N_{\mathrm{cut}}}

\newcommand*{\cF}{\mathcal{F}}
\newcommand*{\cG}{\mathcal{G}}
\newcommand*{\cH}{\mathcal{H}}

\newcommand*{\cN}{\mathcal{N}}

\newcommand*{\cR}{\mathcal{R}}

\newcommand*{\Pauli}{\mathsf{Pauli}}

\newcommand*{\cU}{\mathcal{U}}
\newcommand*{\cT}{\mathcal{T}}

\newcommand{\Good}{\mathsf{Good}}
\newcommand{\Bad}{\mathsf{Bad}}

\newcommand{\bits}{\{0,1\}}

\providecommand{\abs}[1]{\left|#1\right|}
\providecommand{\given}{\,|\,}

\newcommand{\bc}{\begin{center}}
\newcommand{\ec}{\end{center}}

\newtheorem{theorem}{Theorem}[section]
\newtheorem{lemma}[theorem]{Lemma}

\newtheorem{proposition}[theorem]{Proposition}
\newtheorem{corollary}{Corollary}[section]

\DeclareMathOperator{\tr}{tr}
\DeclareMathOperator{\poly}{poly}

\def\01{\{0,1\}}

\newcommand*{\unitarychannel}[1]{\mathcal{U}_{#1}}

\begin{document}

\title{Parallel classical simulation of noisy shallow  circuits:\\
 no quantum advantage in $1D$}

\author{Robert K\"onig\thanks{
Department of Mathematics, School of Computation, Information and Technology, Technical University of Munich, \&  Munich Center for Quantum Science and Technology, Munich, Germany} \and Marco Tomamichel\thanks{Department of Electrical and Computer Engineering and Centre for Quantum Technologies, National University of Singapore, Singapore}}
\date{September 30, 2026}
\maketitle

\begin{abstract}
We consider quantum circuits consisting of $d$~layers of nearest-neighbor two-qubit gates acting on $n$~qubits arranged on a line, where every qubit is independently depolarized with a constant probability before each layer. We describe a randomized parallel algorithm which samples from the output distribution of any such circuit to within total variation error~$\delta$, with parallel runtime $2^{O(d)}\log\log(n/\delta)$ and $n\,2^{O(d)}$ elementary real-arithmetic operations. Without noise, the same approach gives an exact sampler with parallel runtime $O(\log n)$. For constant depth, we further show that the input/output behavior of the noisy quantum circuit is reproduced up to a constant error by a randomized $\mathsf{AC}^0$-circuit, that is, a Boolean circuit of polynomial size and constant depth with unbounded fan-in AND and OR gates and NOT gates. Consequently, every relation problem solved by a noisy constant-depth quantum circuit in one dimension is also solved, with essentially the same success probability, by a randomized $\mathsf{AC}^0$-circuit. This rules out, for noisy circuits in one dimension, the unconditional quantum advantage established for noisy shallow circuits in two and three dimensions, where polynomial-size classical circuits over the same gate set require depth $\Omega(\log n/\log\log n)$. Our algorithm exploits the fact that depolarizing noise cuts a one-dimensional circuit into independent pieces of logarithmic width, each of which can be sampled exactly in parallel.
\end{abstract}

\section{Introduction}
\label{sec:introduction}

Shallow quantum circuits -- circuits whose depth does not grow with the number of qubits -- are among the simplest models of quantum computation, and among the most instructive. On the one hand, they are believed to be hard to simulate classically. Terhal and DiVincenzo~\cite{TD04} showed that no polynomial-time classical algorithm samples exactly from the output distribution of a certain depth-$3$ circuit unless the polynomial hierarchy collapses, and under additional average-case hardness and anticoncentration conjectures this remains true for certain constant-depth circuits even if a constant error in total variation is allowed~\cite{BMS16,BHSRE18}. On the other hand, shallow circuits are the setting in which quantum advantage can be established unconditionally. Ref.~\cite{BGK18} exhibited a relation problem which is solved with certainty by a constant-depth circuit of nearest-neighbor gates on a two-dimensional grid, whereas any classical circuit of bounded fan-in gates which solves it with constant probability must have depth $\Omega(\log n)$. This separation was extended to classical circuits which additionally have unbounded fan-in AND and OR gates, for which the depth lower bound becomes $\Omega(\log n/\log\log n)$~\cite{WKST19}, and to quantum circuits subject to local stochastic noise~\cite{BGKT20,GJS21,CCR26,CKP26}. Equating circuit depth with parallel running time, these results show that constant-time parallel quantum computation can outperform logarithmic-time parallel classical computation.

Whether such a separation is relevant for experiments depends on three features of the quantum circuit: it should tolerate noise of constant strength, it should be geometrically local in the architecture at hand, and the gap between quantum and classical depth should be large as a function of the number~$n$ of qubits, which is limited on near-term devices. The noise-robust, geometrically local constructions known to date all use a two- or three-dimensional layout~\cite{BGKT20,CCR26,CKP26}. It is natural to ask whether a similar advantage is possible in one dimension, the simplest geometry and a natural one for several hardware platforms. Without noise, the answer is yes: there are constant-depth circuits of nearest-neighbor gates on a line ($1D$-local shallow circuits) which solve relation problems requiring classical depth $\Omega(\log n)$ with bounded fan-in gates~\cite{BGKT20} and $\Omega(\log n/\log\log n)$ with unbounded fan-in gates~\cite{CCR26}. Can such a quantum advantage also be achieved by $1D$-local quantum circuits which are noisy?

In this work we show that the answer is negative. Our main result is a randomized parallel classical algorithm which samples from the output distribution of any depth-$d$ circuit of nearest-neighbor gates on $n$~qubits on a line, subject to independent depolarizing noise of constant strength before every layer, to within total variation error~$\delta$, with parallel runtime $2^{O(d)}\log\log(n/\delta)$ and $n\,2^{O(d)}$ arithmetic operations. For a constant depth, the algorithm can be implemented by randomized Boolean circuits of size $O(n^2(\log n)^4)$ and depth $O(\log\log n)$ composed of unbounded fan-in AND and OR gates and NOT gates. If a larger polynomial size is allowed, sampling each local marginal in a single step using a precomputed table reduces the depth to a constant, which gives a randomized $\mathsf{AC}^0$-circuit. These circuits reproduce the input/output behavior of the noisy quantum circuit up to a constant error on every input. Any relation problem solved by a noisy constant-depth quantum circuit in one dimension is therefore solved, with an arbitrarily small constant loss in success probability, by a randomized $\mathsf{AC}^0$-circuit, and no superconstant lower bound on the depth of polynomial-size classical circuits with unbounded fan-in can hold in this setting.

We emphasize that the relevant complexity measure here is parallel runtime, or classical circuit depth. Polynomial sequential time does not exclude a depth advantage: the problems underlying the classical lower bounds of~\cite{BGK18,WKST19,BGKT20,CCR26,CKP26} can all be solved in classical polynomial time, and noisy one-dimensional circuits of logarithmic depth are known to be simulable in polynomial time by a sequential algorithm with runtime at least linear in~$n$~\cite{YDCM25}. To our knowledge, our result provides the first classical simulation algorithm for noisy quantum circuits with a bound on parallel runtime or classical circuit depth.

\subsection{Parallel classical simulation of noisy one-dimensional circuits}
\label{sec:simulation-overview}

To describe our results, consider $n$~qubits arranged on a line and a brickwork circuit~$\sU$ of depth~$d$, that is, $d$~layers of two-qubit unitary gates acting on nearest neighbors, where consecutive layers alternate between the two possible pairings of neighboring qubits (see Fig.~\ref{fig:bilayerqubits}). Before every layer, each qubit is independently replaced by the maximally mixed state with probability $p\in(0,1)$. The noisy circuit is applied to the input state~$\ket{0^n}$ and all qubits are measured in the computational basis; we write $P^{(p)}$ for the resulting distribution on $n$-bit strings. Classical (weak) simulation of the noisy circuit is the task of sampling from a distribution~$\widehat P$ which is close to~$P^{(p)}$ in total variation distance, $d_{\mathrm{TV}}(P,Q):=\frac12\sum_{z}|P(z)-Q(z)|$. Section~\ref{sec:quantum-noise-model} specifies the model in more detail.

\begin{figure}[!htbp]
\centering
\begin{tikzpicture}
  \node[anchor=north west,inner sep=0] (overview) at (0,0)
    {\includegraphics[height=0.30\textheight]{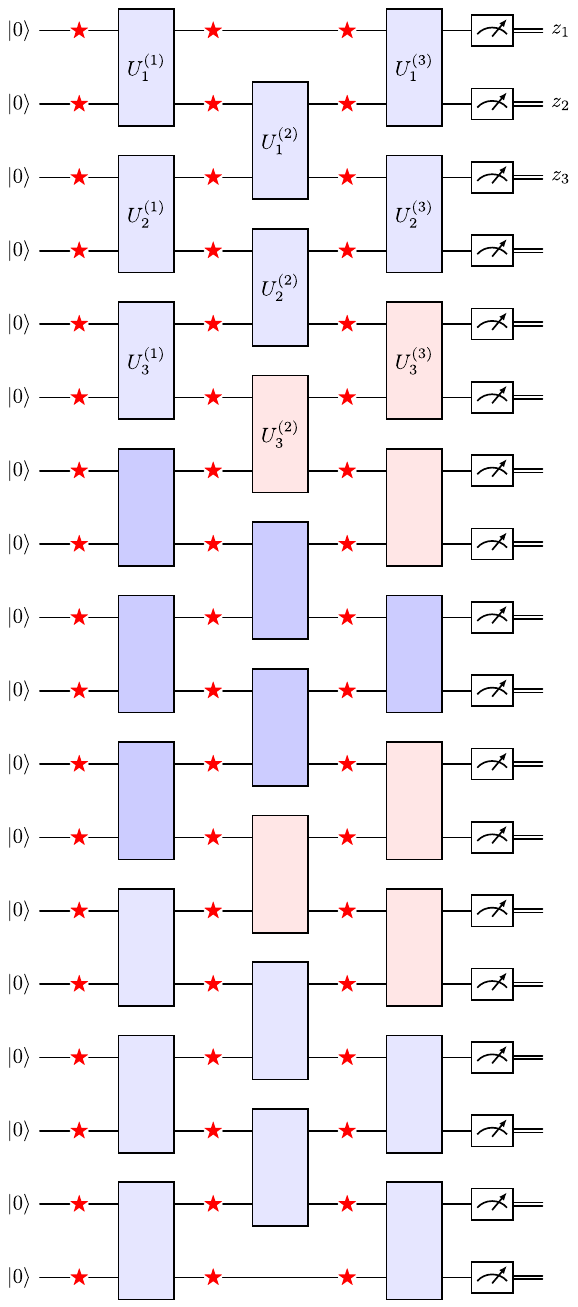}};
  \node[anchor=north west,inner sep=0] (detail)
    at ([xshift=1.4cm]overview.north east)
    {\includegraphics[width=0.47\textwidth,
      trim=0 385.036bp 0 0,clip]{bilayerqubits.pdf}};

  \coordinate (cropSW) at
    ($(overview.south west)!0.6140575!(overview.north west)$);
  \coordinate (cropSE) at
    ($(overview.south east)!0.6140575!(overview.north east)$);
  \draw[gray,line width=0.5pt]
    (cropSW) -- (overview.north west) -- (overview.north east) -- (cropSE);
  \draw[gray,densely dashed,line width=0.5pt] (cropSW) -- (cropSE);
  \draw[gray!65,line width=0.4pt]
    (overview.north east) -- (detail.north west)
    (cropSE) -- (detail.south west);
  \draw[gray,line width=0.5pt]
    (detail.south west) -- (detail.north west) --
    (detail.north east) -- (detail.south east);
  \draw[gray,densely dashed,line width=0.5pt]
    (detail.south west) -- (detail.south east);
  \node[anchor=south,font=\small,inner sep=0,yshift=4pt]
    at (overview.north) {Full circuit};
  \node[anchor=south,font=\small,inner sep=0,yshift=4pt]
    at (detail.north) {First labeled gates};
\end{tikzpicture}
\caption{A three-layer brickwork circuit on 18 qubits, initialized in
$\ket{0^{18}}$ and measured in the computational basis. The inset enlarges
the first labeled gates. Each red star represents the channel~$\cN_p$: independently at every location,
the qubit is completely depolarized with probability~$p$.
 Gate colors indicate the
groups used in the bilayer reduction of Fig.~\ref{fig:coarsegrained}.}
\label{fig:bilayerqubits}
\end{figure}

We measure the cost of a classical simulation in a parallel real-arithmetic model of computation (see Section~\ref{sec:computational-model}). Registers hold exact real numbers. An elementary operation is a real arithmetic operation, a comparison, a Boolean operation on a constant number of bits, a selection between two values according to a bit, or the generation of a random bit with distribution $\mathsf{Ber}(r)$ for a register value $r\in[0,1]$. Independent operations may be applied in parallel. The size of an algorithm is the number of elementary operations it uses, and its runtime is the length of the longest chain of dependent operations.

Our main result is the following. It assumes that the quantum circuit~$\sU$ is specified by the matrix elements of its gates in the computational basis.
\begin{theorem}[Parallel classical simulation]\label{thm:main}
Let $n,d$ be even integers with $2\leq d\leq n$, and let
$p,\delta\in(0,1)$. Let $\sU$ be any depth-$d$ brickwork circuit on
$n$ qubits, and let $P^{(p)}$ be the output distribution of its noisy implementation with independent single-qubit depolarizing noise of strength~$p$ before every gate layer, applied to the input state~$\ket{0^n}$ and measured in the computational basis.
There is a
randomized parallel algorithm \textnormal{\Sample{}} in the
real-arithmetic model which, given the matrix elements of the gates of~$\sU$ and the noise
parameter~$p$, produces a sample from a distribution $\widehat P$ satisfying
\begin{align}
  d_{\mathrm{TV}}(\widehat P,P^{(p)})\leq\delta.
  \label{eq:main-total-variation}
\end{align}
Its worst-case  runtime and number of elementary real operations
are, respectively,
\begin{align}
\begin{matrix}
\mathsf{runtime}(\textnormal{\Sample{}})&=&2^{O(d)}\bigl(d\log\tfrac{1}{p}+\log\log\tfrac{n}{\delta}\bigr)&\\
\mathsf{size}(\textnormal{\Sample{}})&=&n\,2^{O(d)}&\ .
\end{matrix}
\label{eq:main-complexity}
\end{align}
\end{theorem}
\noindent The proof is given in Section~\ref{sec:simulation-proof}; we sketch the main idea below.

In particular, for constant noise strength $p$ and depth~$d$ and constant or inverse-polynomial error~$\delta$, the algorithm has runtime $O(\log\log n)$ and uses $O(n)$ elementary operations; this is the regime of interest for quantum advantage with shallow circuits. 
The runtime remains $o(\log n)$ for $d = o(\log\log n)$, and both bounds are polynomial for $d=O(\log n)$. 

Theorem~\ref{thm:main} extends to circuits with a classical input. For $x\in\{0,1\}^k$ with $k\leq n$, let $P^{(p)}(\cdot|x)$ be the output distribution of the noisy circuit applied to the input state $\ket{x}\otimes\ket{0^{n-k}}$. A simple modification of the proof of Theorem~\ref{thm:main} gives a parallel algorithm which takes $x$ as an additional input and samples from a distribution within total variation distance~$\delta$ of $P^{(p)}(\cdot|x)$, with the same runtime and size as in Eq.~\eqref{eq:main-complexity}. The Boolean circuits of Section~\ref{sec:boolean-overview} extend to classical inputs in the same way, and this form is used in Section~\ref{sec:no-practical-advantage} to bound the circuit-depth advantage of noisy shallow circuits in~$1D$. Similar generalizations apply to classically controlled circuits assuming that each gate is controlled by at most a constant number of input bits.

Our construction also applies without noise, where it gives an exact sampler for the output distribution of any depth-$d$ brickwork circuit with parallel runtime $O(\poly(2^d)\log n)$ using $O(n\,\poly(2^d))$ elementary operations (Corollary~\ref{cor:noiseless}). To our knowledge, this result is new as well. For constant depth, noise thus improves the parallel runtime from~$O(\log n)$ to~$O(\log\log n)$. The logarithmic runtime in the noiseless case is optimal~\cite{BGKT20,CCR26}. We do not know whether the doubly logarithmic runtime in the noisy case can be improved in the real-arithmetic model; with Boolean gates of unbounded fan-in, a larger polynomial size allows constant depth (Corollary~\ref{cor:compiled-constant-depth}).

\paragraph{Idea of the proof.} The noise preceding a two-qubit gate is fully depolarizing on both input qubits with probability~$p^2$, in which case the noisy gate can be replaced by a random product of single-qubit Pauli gates; with the remaining probability~$1-p^2$, the gate is retained, with a random Pauli pair applied before it (Lemma~\ref{lem:circuitcutting}). We call the first branch a cut. Sampling one of the two branches, together with its Pauli gates, at every gate of the circuit produces a random ensemble of noiseless unitary circuits whose average output distribution is~$P^{(p)}$. If cuts occur at all $d/2$ gates crossing a fixed bond in a depth-$d$ circuit, the circuit separates across that bond; this happens with probability~$p^d$, independently for different bonds (Section~\ref{sec:circuit-cutting}). With probability at least $1-\delta$, the sampled circuit therefore decomposes into independent pieces of width $O(p^{-d}\log(n/\delta))$, and conditioned on the noise realization, its output distribution factorizes over these pieces.

Since the boundaries of the pieces are random, we sample the marginals on two overlapping families of fixed intervals of this width, using independent randomness for each interval, and a parallel selection procedure copies each piece from a single interval sample containing it (Section~\ref{sec:marginalsamples}). Whenever all pieces satisfy the width bound, the assembled sample has exactly the output distribution conditioned on the noise realization, so only realizations with an oversized piece contribute to the error~$\delta$. The marginal on a fixed interval is sampled exactly by reducing the circuit on the interval's backward light cone to a bilayer circuit and joining samples on neighboring subintervals as in~\cite{BravyiGossetLiu22} along a balanced binary tree (Section~\ref{sec:circuitsampling}). Sampling and selection take parallel runtime logarithmic in the interval width, up to a factor $2^{O(d)}$ for local computations, which gives Eq.~\eqref{eq:main-complexity}. Because the intervals are fixed in advance, this bound holds for every noise realization.

\subsection{Boolean circuit implementation}
\label{sec:boolean-overview}
For constant depth $d=O(1)$, the algorithm of Theorem~\ref{thm:main} can be implemented by Boolean circuits with access to independent uniform random bits, at a controlled overhead and loss in sampling accuracy. Here we assume that the real and imaginary parts of the gate entries and the noise probability~$p$ are given with $O(\log(n/\delta))$ bits of precision (see Appendix~\ref{app:full-implementation}). The size of a Boolean circuit is its number of gates, and its depth is the maximum number of gates on a directed path from an input to an output. Gates have unrestricted fan-out; we distinguish between circuits with bounded fan-in (that is, fan-in $2$) and unbounded fan-in gates.

\begin{corollary}[Weak simulation by bounded-fan-in classical circuits]
\label{cor:overview-bounded-boolean}
Let $d\geq 2$ be an even constant, let the noise strength $p\in(0,1)$ be constant, and let the error $\delta\in(0,1)$ be constant or inverse-polynomial in~$n$. For every noisy depth-$d$ brickwork circuit on $n$~qubits as above, there is a randomized Boolean circuit~$\cC_2$ with
fan-in-two AND and OR gates and unary NOT gates whose output distribution
$\widehat P$ satisfies~$d_{\mathrm{TV}}(\widehat P,P^{(p)})\leq\delta$.
Its depth and size are, respectively,
\begin{align}
\begin{matrix}
\mathsf{depth}(\cC_2)&=&O((\log\log n)^2)\,,\\
\mathsf{size}(\cC_2)&=&O(n(\log n)^4)\,.
\end{matrix}
\end{align}
\end{corollary}
With unbounded fan-in gates, the depth improves to doubly logarithmic.
\begin{corollary}[Weak simulation by unbounded-fan-in classical circuits]
\label{cor:overview-unbounded-boolean}
Under the same assumptions as in Corollary~\ref{cor:overview-bounded-boolean}, there is a randomized Boolean circuit~$\cC_\infty$ with
unbounded fan-in AND and OR gates and unary NOT gates whose output distribution
$\widehat P$ satisfies $d_{\mathrm{TV}}(\widehat P,P^{(p)})\leq\delta$.
Its depth and size are, respectively,
\begin{align}
\begin{matrix}
\mathsf{depth}(\cC_\infty)&=&O(\log\log n)\,,\\
\mathsf{size}(\cC_\infty)&=&O(n^2(\log n)^4)\,.
\end{matrix}
\end{align}
\end{corollary}

Corollaries~\ref{cor:overview-bounded-boolean} and~\ref{cor:overview-unbounded-boolean} are proved in Appendix~\ref{app:boolean-arithmetic}; the implementations are described in Section~\ref{sec:circuit-realizations}, and the finite-precision analysis is given in Appendix~\ref{app:boolean-implementation}. Both circuits are uniform, i.e., they can be constructed in polynomial time, and both extend to circuits with a classical input $x\in\{0,1\}^k$ as discussed after Theorem~\ref{thm:main}. Beyond constant depth, a brickwork circuit of depth $d=O(\log\log n)$ with constant noise strength and constant or inverse-polynomial error has a randomized classical simulation with bounded fan-in gates, depth $O((\log\log n)^2)$, and size $n(\log n)^{O(1)}$ (Theorem~\ref{thm:full-bounded-boolean}).

The depth in Corollary~\ref{cor:overview-unbounded-boolean} comes from the dependent levels of the joining trees used for local sampling (see Section~\ref{sec:coarsening-subsec}). If a larger polynomial size is allowed, these levels can be avoided. For fixed quantum gates, the marginal distribution on each sampling interval depends only on $O(\log n)$ bits of classical input and noise, so a sample from it can be drawn in a single step using a precomputed table of polynomial size. This gives a randomized $\mathsf{AC}^0$-circuit, that is, a circuit of polynomial size and constant depth composed of unbounded fan-in AND and OR gates and NOT gates, with independent uniform random bits as additional inputs.
\begin{corollary}[Weak simulation by randomized $\mathsf{AC}^0$-circuits]
\label{cor:compiled-constant-depth}
Under the same assumptions as in Corollary~\ref{cor:overview-bounded-boolean}, there is a randomized Boolean circuit~$\cC_{\mathrm{lt}}$ with
unbounded fan-in AND and OR gates and unary NOT gates whose output distribution
$\widehat P$ satisfies $d_{\mathrm{TV}}(\widehat P,P^{(p)})\leq\delta$.
Its depth and size are, respectively,
\begin{align}
\begin{matrix}
\mathsf{depth}(\cC_{\mathrm{lt}})&=&O(1) \, ,\\
\mathsf{size}(\cC_{\mathrm{lt}})&=& \mathrm{poly}(n) \,.
\end{matrix}
\label{eq:compiled-complexity}
\end{align}
\end{corollary}
Corollary~\ref{cor:compiled-constant-depth} is proved in Appendix~\ref{app:lookup-tables}. The depth bound is an absolute constant: the circuit~$\cC_{\mathrm{lt}}$ has depth at most~$10$ for all $d$, $p$, and~$\delta$. Its size, however, is $2^{O(d^2)}(n/\delta)^{O(dp^{-d})}$, a polynomial whose degree is large for small~$p$, so Corollary~\ref{cor:compiled-constant-depth} complements rather than supersedes the explicit size bounds of Corollaries~\ref{cor:overview-bounded-boolean} and~\ref{cor:overview-unbounded-boolean}. Like these, the circuit~$\cC_{\mathrm{lt}}$ can be constructed in polynomial time, and it extends to classical inputs $x\in\{0,1\}^k$: a single circuit takes~$x$ as an input and satisfies the error bound for every~$x$. It also extends to circuits with classically controlled gates, each depending on $O(1)$ input bits.

\subsection{Comparison with prior work}
\label{sec:related-simulation}

Theorem~\ref{thm:main} holds for every brickwork circuit, at every depth including constant depth, and it bounds parallel rather than sequential complexity. Prior work on the classical simulation of noisy circuits differs in at least one of these respects. It falls into three groups: general limitations of noisy circuits, algorithms for random circuits, and algorithms for arbitrary circuits from restricted classes.

Local depolarizing noise drives the state of a circuit without fresh ancillas towards the maximally mixed state, and noisy circuits retain no nontrivial correlations beyond logarithmic depth~\cite{ABIN96}. Quantitatively, after $d$~layers of depolarizing noise of strength~$p$, the relative entropy between the output state and the maximally mixed state is at most $(1-p)^{2d}n$~\cite{MFW16,FG21}, so that for constant~$p$ the output distribution is $\epsilon$-close to uniform once $d=\Omega(\log(n/\epsilon))$, and weak simulation becomes trivial. The same contraction bounds the sample complexity of error mitigation~\cite{TEMG22,TTG23,QFK+24} and limits noisy variational optimization and the size of gradients~\cite{FG21,DMRF23,WFC+21}. These results hold for arbitrary circuits but concern specific tasks and depths beyond a mixing threshold. At constant depth, the noisy output distribution can be highly structured -- for the noisy identity circuit it has total variation distance $1-2^{-\Omega(n)}$ from uniform -- and our algorithm relies on the local effect of individual depolarizing events rather than on global mixing.

A second line of work gives algorithms for circuits with Haar-random gates which succeed with high probability over the choice of the circuit; the key ingredient is anticoncentration of the noiseless output distribution. Depolarizing noise suppresses the Pauli-path contributions to the output distribution exponentially in their weight~\cite{GD18}, which gives a polynomial-time algorithm sampling within inverse-polynomial total variation distance from noisy random circuits of depth $\Omega(\log n)$, the depth needed for anticoncentration~\cite{AGLLV23}; anticoncentration alone suffices for quasi-polynomial-time sampling~\cite{SYGY25}. For extensions and limitations of these methods, see~\cite{GGCT25,FGGKS24,MAG+26}. Tensor-network methods target the same regime, for noisy random $1D$ circuits~\cite{NJF20} and for random shallow $2D$ circuits, noiseless~\cite{NLD+22} or noisy~\cite{CI23}, where sampling maps to a one-dimensional dynamics with measurements which is efficiently simulable in an area-law phase. Bene Watts et al.~\cite{BWGLS25} showed that noiseless shallow $2D$ circuits with short-range measurement-induced entanglement can be sampled by classical circuits of depth $O(d)$ whose gates act on $O(\log^2 n)$ bits; the circuit of Corollary~\ref{cor:compiled-constant-depth} has the same form if each table lookup is regarded as a single gate, with constant depth and gates acting on $O(\log n)$ bits. None of these assumptions is available in our setting: the circuits in quantum advantage proposals are specific rather than typical, arbitrary shallow circuits need not anticoncentrate (in one dimension even random circuits require depth $\Omega(\log n)$ to do so), and short-range measurement-induced entanglement fails already for the $1D$ cluster state, whose bulk measured in the $X$-basis leaves its end qubits maximally entangled.

The result closest to ours is the exact sampler of Yan, Du, Chen, and Ma~\cite[Lemma~3]{YDCM25} for $1D$ circuits of depth~$d$ on $n$~qubits, with or without single-qubit noise. It draws the output bits one after the other from their conditional distributions in time $O(nd\,2^{2d})$, which is polynomial for $d=O(\log n)$; combined with mixing, this gives polynomial-time simulation at all depths under strictly contractive unital noise. The sampler is sequential, and no bound on parallel runtime or classical circuit depth follows from it. Two recent results exploit noise in a spirit similar to ours, for restricted gate sets on arbitrary interaction graphs: arbitrary instantaneous quantum polynomial-time (IQP) circuits under dephasing or depolarizing noise can be sampled in polynomial time beyond a constant critical depth, and this depth is tight~\cite{RWL25}; the same holds for Clifford circuits with product-state inputs and for IQP circuits with CNOT gates~\cite{NRHG26}. There, errors are propagated to the input through the commutation structure of the gates, where they act as depolarization of input qubits, and beyond the critical depth the remaining qubits form components of logarithmic size. This mechanism is not available for general gates, and our algorithm uses the one-dimensional geometry instead. The role of noise also differs: for IQP circuits, noise is what makes sampling from constant-depth circuits tractable, whereas in one dimension constant-depth circuits are tractable without noise, and noise reduces the parallel runtime from $O(\log n)$ to $O(\log\log n)$. Finally, several worst-case results concern tasks weaker than sampling, such as estimating expectation values with error averaged over inputs~\cite{SYGY25} or producing samples which approximately preserve objectives depending on low-order correlations~\cite{MFF26}. For relation problems, an output satisfying the relation is required for every input, which needs samples that are close in total variation, as provided by Theorem~\ref{thm:main}.

\subsection{Application: no asymptotic depth advantage in one dimension}
\label{sec:no-practical-advantage}
We now return to the question raised at the beginning of the introduction. Table~\ref{tab:quantum-advantage-comparison} compares the known unconditional separations between shallow quantum circuits and classical circuits by the three criteria named there: robustness to circuit-level noise of constant strength, geometric locality of the quantum circuit, and the size of the classical depth lower bound as a function of~$n$. The noise-robust, geometrically local constructions use a $2D$ or $3D$ layout, with constant quantum depth and classical depth $\Omega(\log n/\log\log n)$.

\begin{table}[htbp]
\centering
\footnotesize
\setlength{\tabcolsep}{4pt}
\renewcommand{\arraystretch}{1.35}
\begin{tabularx}{\linewidth}{@{}>{\raggedright\arraybackslash}p{0.28\linewidth}>{\centering\arraybackslash}p{0.15\linewidth}|>{\centering\arraybackslash}p{0.14\linewidth}>{\centering\arraybackslash}p{0.21\linewidth}>{\centering\arraybackslash}X@{}}
\toprule
\multicolumn{2}{c|}{\textbf{Quantum}} & \multicolumn{2}{c}{\textbf{Classical}} & \\
\midrule
Noise robustness & Geometric locality & Fan-in & Depth lower bound & Reference \\
\midrule
No & 2D & Bounded & $\Omega(\log n)$ & \cite{BGK18} \\
No & 2D & Unbounded & $\Omega\big(\frac{\log n}{\log\log n}\big)$ & \cite{WKST19} \\
No & 1D & Bounded & $\Omega(\log n)$ & \cite{BGKT20} \\
Local stochastic, $p<p_0$ & 3D & Bounded & $\Omega\big(\frac{\log n}{\log\log n}\big)$ & \cite{BGKT20} \\
Local stochastic, $p<p_0$ & None & Unbounded & $\Omega\big(\frac{\log n}{\log\log n}\big)$ & \cite{GJS21} \\
No & 1D & Unbounded & $\Omega\big(\frac{\log n}{\log\log n}\big)$ & \cite{CCR26} \\
Local stochastic, $p<p_0$ & 3D & Unbounded & $\Omega\big(\frac{\log n}{\log\log n}\big)$ & \cite{CCR26} \\
Local stochastic, $p<p_0$ & 2D & Unbounded & $\Omega\big(\frac{\log n}{\log\log n}\big)$ & \cite{CKP26} \\\bottomrule
\end{tabularx}
\caption{Unconditional quantum advantages for relation problems, where $n$ parametrizes the problem size. All quantum circuits have polynomial size and depth $O(1)$.
Noise-robust circuits tolerate local stochastic noise below a threshold strength~$p_0$ which depends on the construction;
``No'' means that the circuit does not solve the problem when subjected to noise. Geometric locality refers to the quantum circuit being composed of nearest-neighbor gates on a regular lattice in~$D$ dimensions; ``None'' means that the circuit may use geometrically non-local gates when the qubits are embedded in~$\mathbb{R}^3$. The depth lower bounds apply to noise-free, polynomial-size classical circuits of unconstrained geometry which solve the problem with a constant probability, with either bounded fan-in gates or unbounded fan-in AND and OR gates and unary NOT gates.}
\label{tab:quantum-advantage-comparison}
\medskip
\end{table}

To state the consequence of Corollary~\ref{cor:compiled-constant-depth} for one dimension precisely, consider a relation problem defined by a subset $R_n\subseteq \{0,1\}^{\mathsf{poly}(n)}\times \{0,1\}^{\mathsf{poly}(n)}$
that lists the accepted input-output pairs and a subset $S_n$ of all inputs for which at least one output is accepted. A randomized classical circuit or a (possibly noisy) quantum circuit~$\sU$ solves the problem defined by $(R_n,S_n)$ with probability~$p$ if its output $\sU(x)$ satisfies
\begin{align}
\Pr\left[(x,\sU(x))\in R_n\right]\geq p\qquad\textrm{ for every input }\qquad x\in S_n\ .
\end{align}
For a quantum circuit, $\sU(x)$ is a sample from the distribution obtained by
preparing the state~$\ket{x}\otimes\ket{0^{\mathsf{poly}(n)}}$, applying unitary one- and two-qubit gates, possibly interspersed with noise, and measuring in the computational basis (see Section~\ref{sec:quantum-noise-model} for precise definitions).
\par\noindent\begin{minipage}{\linewidth}
\begin{theorem}[No asymptotic depth advantage in $1D$]\label{thm:no-practical-advantage}
Suppose $(R_n,S_n)$ is a relation such that the following holds for two constants $0<p_c<p_q\leq 1$.
\begin{enumerate}[(i)]
\item
There is a constant-depth polynomial-size $1D$ brickwork quantum circuit which solves the relation problem defined by~$(R_n,S_n)$ with probability at least~$p_q$ even under circuit-level independent single-qubit depolarizing noise of constant strength.
\item
Any randomized polynomial-size classical circuit composed of unbounded fan-in AND, OR and NOT gates, which solves the relation problem defined by~$(R_n,S_n)$ with probability at least~$p_c$, has depth at least~$d_c$.
\end{enumerate}
Then $d_c=O(1)$.
\end{theorem}
\end{minipage}\par\medskip
\begin{proof}
Apply Corollary~\ref{cor:compiled-constant-depth} in its form for classical inputs (Appendix~\ref{app:compiled-simulation}), with a constant error $\delta<p_q-p_c$. The resulting classical circuit succeeds with probability at least $p_q-\delta>p_c$ on every input $x\in S_n$. Since the quantum circuit has polynomially many qubits, the classical circuit has polynomial size and constant depth.
\end{proof}
\noindent The proof gives $d_c\leq10$. With bounded fan-in gates, Corollary~\ref{cor:overview-bounded-boolean} gives depth $O((\log\log n)^2)$ instead. In particular, neither the $\Omega(\log n/\log\log n)$ lower bounds for unbounded fan-in established for noiseless circuits in one dimension~\cite{CCR26} and for noisy circuits in two and three dimensions~\cite{CCR26,CKP26}, nor the $\Omega(\log n)$ lower bound for bounded fan-in established for noiseless circuits in one dimension~\cite{BGKT20}, can hold for noisy shallow circuits in one dimension. Theorem~\ref{thm:no-practical-advantage} concerns the asymptotic depth at polynomial size; for the circuits of Corollary~\ref{cor:compiled-constant-depth}, the degree of this polynomial in the number of qubits is $O(dp^{-d})$. It does not compare quantum and classical running times at finite system sizes.

Theorem~\ref{thm:no-practical-advantage} is stated for relation problems. Since our simulation replaces the entire input/output behavior of a noisy $1D$ brickwork circuit by a classical circuit, in a black-box manner, we expect analogous bounds on depth advantage to hold more generally: for noisy shallow $1D$ circuits used to establish quantum advantages in interactive settings such as those considered by Grier, Ju, and Schaeffer~\cite{GJS21}, and for input-independent sampling problems where shallow circuits provide an advantage (cf.~\cite{BWP26,GKMOW26}).

\paragraph{Outline.} The remainder of the paper is structured as follows. Section~\ref{sec:problem-statement} specifies the noise model and the model of computation and delimits the scope of Theorem~\ref{thm:main}. Section~\ref{sec:circuit-cutting} shows how depolarizing noise decomposes a brickwork circuit into independent components and bounds their widths. Section~\ref{sec:marginalsamples} shows how to combine samples on fixed intervals, and Sections~\ref{sec:messagepassingalgorithmcontainment} and~\ref{sec:circuitsampling} construct a parallel exact sampler for these intervals. Section~\ref{sec:full-routine} assembles these steps into the algorithm \Sample{} and proves the runtime and size bounds of Theorem~\ref{thm:main}. Section~\ref{sec:circuit-realizations} describes the Boolean circuit implementations, with the finite-precision analysis deferred to Appendix~\ref{app:boolean-implementation}. Finally,  Appendix~\ref{app:lookup-tables} proves Corollary~\ref{cor:compiled-constant-depth}.

\section{Noisy brickwork quantum circuits and parallel classical computation in the arithmetic model}
\label{sec:problem-statement}
In this section we define the noisy $1D$-local quantum circuits considered and the real-arithmetic model of parallel classical computation in which Theorem~\ref{thm:main} is stated.

\subsection{Quantum circuit and noise model}
\label{sec:quantum-noise-model}
  For concreteness, let $n$ and $d$ be even, with $d\leq n$. We consider a depth-$d$ brickwork circuit~$\sU$ on $n$~qubits $Q_1,\ldots,Q_n$ arranged on a line. Such circuits are geometrically local in one dimension, as shown in Fig.~\ref{fig:bilayerqubits}.

The circuit has the form
\begin{align}
\sU=L_dL_{d-1}\cdots L_1\ , \label{eq:idealcircuit}
\end{align}
where each gate layer~$L_t$, $t\in [d]$, is a product of nearest-neighbor two-qubit gates. 
Specifically, we assume that for $t\in [d]$, we have 
\begin{align}
L_t &=
\begin{cases}
\bigotimes_{j=1}^{n/2} (U_j^{(t)})_{Q_{2j-1}Q_{2j}}\qquad &\textrm{ for } t \textrm{ odd}\\
\bigotimes_{j=1}^{n/2-1} (U_j^{(t)})_{Q_{2j}Q_{2j+1}}\qquad &\textrm{ for } t \textrm{ even}\ ,
\end{cases}
\label{eq:physical-gate-layers}
\end{align}
where $\{U_j^{(t)}\}_{j,t}$ are two-qubit unitaries supported on the qubits
\begin{align}
\mathsf{supp}(U^{(t)}_j)=
\begin{cases}
(Q_{2j-1},Q_{2j})\qquad &\textrm{ for }t\textrm{ odd}\\
(Q_{2j},Q_{2j+1})\qquad &\textrm{ for }t\textrm{ even}\ .
\end{cases}
\end{align}
We consider a noisy implementation of $\sU$ where each qubit is independently depolarised before each gate layer.  For $p\in [0,1]$,~let 
\begin{align}
  \cN_p(\rho)=(1-p)\rho+p\,\tr(\rho)\frac{I}{2}
  \label{eq:depol}
\end{align}
be the single-qubit depolarizing channel. Denoting the quantum channel associated with a unitary~$U$ by $\unitarychannel{U}(\cdot)=U\cdot U^\dagger$, this is the quantum channel
\begin{align}
\cU^{(p)}&=(\unitarychannel{L_d}\circ \cN_p^{\otimes n})\circ (\unitarychannel{L_{d-1}}\circ \cN_p^{\otimes n})\circ\cdots \circ (\unitarychannel{L_1}\circ \cN_p^{\otimes n})\ \label{eq:noisyUimpelmentdefinition}.
\end{align}
For $z=(z_1,\ldots,z_n)\in \{0,1\}^n$
let $\ket{z}=\ket{z_1}\otimes\cdots\otimes\ket{z_n}$ denote the computational basis vector. We are interested in the output distribution when such a noisy implementation is applied to the initial product state~$\ket{0^n}$ and a computational basis measurement is performed, i.e., in the probability distribution
\begin{align}
P^{(p)}(z)&=\langle z | \cU^{(p)} (\proj{0^n}) | z\rangle  . \label{eq:targetdistribution}
\end{align}

We note that
for ease of analysis, we have omitted   a layer of depolarizing noise immediately before
the final computational basis measurement
in the definition~\eqref{eq:noisyUimpelmentdefinition}
 of the noisy circuit.
 We point out that such an additional noise layer can be included in the simulation by
returning  $\widetilde Z_j=Z_j\oplus B_j$ with  $B_j\sim \mathsf{Ber}(p/2)$ 
chosen independently and identically for each $j\in[n]$, where $Z=(Z_1,\ldots,Z_n)\sim P^{(p)}$.
\subsection{Real-arithmetic model for parallel classical computation}
\label{sec:computational-model}
\label{it:precision}
For ease of presentation, we formulate our main algorithm for a model of parallel classical computation that allows exact manipulation of reals; we refer to this as the real-arithmetic model. We will subsequently translate this into Boolean circuits; see Section~\ref{sec:boolean-realization}.

In this model, individual registers  hold real numbers. An elementary real operation is addition,
subtraction, multiplication, division by a nonzero number, comparison,
selection between two values according to a bit, or a bounded-arity
Boolean operation. Bits are represented by $0$ and $1$, and complex
numbers by pairs of real numbers. Complex arithmetic and squared absolute
values therefore require a constant number of elementary real operations. 
Selections implement choices that
depend on computed values. A value may be used by several later
operations.

We also allow an elementary operation that returns a bit~$X$ distributed according to $X\sim \mathsf{Ber}(r)$
for any computed $r\in[0,1]$ (held in a register). Each such draw produces an independent random bit;
its bias may depend on earlier outcomes.

We are interested in the total number of elementary operations used by an algorithm. We also allow parallelism: Independent elementary operations can be applied simultaneously.
 The parallel runtime of an algorithm
is the longest chain of dependent operations.
We typically give upper bounds on the runtime and size of a parallel algorithm which are worst-case, i.e., hold for every choice of randomness.

We note that in this model, Bernoulli draws suffice to sample all distributions on finite alphabets used below.
\begin{lemma}[Sampling from computed weights]
\label{lem:discrete-sampling-cost}
Given $m\geq1$ nonnegative real weights $w_1,\ldots,w_m$ with positive sum,
one can sample $i\in[m]$ with probability $w_i/\sum_jw_j$ in parallel
runtime $O(\log(m+1))$, using $O(m)$ elementary real operations.
\end{lemma}
\begin{proof}
We give the elementary construction for completeness.
Place the weights at the leaves of a balanced binary tree and compute
subtree sums from the leaves to the root. At each internal node, draw a
bit whose bias is the right-child sum divided by the sum at that node.
All these draws can be made simultaneously once the sums are available.
Propagate an indicator from the root to the selected child at each level
and return the unique selected leaf. The conditional branch probabilities
telescope to $w_i/\sum_jw_j$. At a zero-mass node, use bias zero; replace
its denominator by one before division so that even unused branches are
defined. There are $O(m)$ operations and $O(\log(m+1))$ dependent levels.
\end{proof}

When addressing our weak sampling problem in this model, we assume the following. 
Matrix elements in the computational basis for the gates~$U^{(t)}_j$ of our circuit, as well as the depolarizing probability~$p$ are supplied exactly.
The cost of computing the weights is counted separately.

\section{Circuit cutting by noise}
\label{sec:circuit-cutting}
In our noise model (see Eq.~\eqref{eq:noisyUimpelmentdefinition}), depolarization can be associated with individual gates: we call a two-qubit unitary preceded by independent depolarizing noise on each of its qubits a noisy gate. Section~\ref{sec:noisycircuitdecompo} decomposes each noisy gate into a convex combination of product unitaries and a residual mixture of (possibly entangling) two-qubit unitaries; the first term amounts to a local cut of the circuit. Section~\ref{sec:convexdecompo} shows that with high probability, the cuts partition the line of $n$~qubits into short components on which the output distribution factorizes, and Section~\ref{sec:cutboundedcompontentssampling} translates this into the noise-sampling subprogram used by our algorithm.

The term circuit cutting, or circuit knitting, also refers to
running large circuits on smaller quantum devices~\cite{PHOW20,MF21,PS24,SPS25}.
There, gates across a partition are decomposed into quasiprobabilistic
mixtures of local operations, and expectation values are reconstructed
classically. In contrast, in our work, the cuts arise from fully depolarizing noise, incur no
sampling overhead, and leave subcircuits that we simulate classically.

\subsection{Circuit cut for a single noisy gate\label{sec:noisycircuitdecompo}}
Here we give the central decomposition of an individual noisy two-qubit gate. 
\begin{lemma}[Circuit-cutting decomposition]\label{lem:circuitcutting}
Let $p\in (0,1)$.  Let $\Pauli=\{I,X,Y,Z\}$ denote single-qubit Pauli operators. Let $U$ be a two-qubit unitary. Then the channel~$\cR_p = \unitarychannel{U} \circ \cN_p^{\otimes 2}$ can be written as
\begin{align}
\cR_p&=p^2\left(\sum_{(Q,R)\in\Pauli^2}\frac{1}{16}
\unitarychannel{Q}\otimes\unitarychannel{R}\right)
 +(1-p^2)\left(\sum_{(Q,R)\in\Pauli^2}
q_{(Q,R)}\unitarychannel{U(Q\otimes R)}\right)\ ,
\label{eq:rpdecompositionbasic}
\end{align}
where $\{q_{(Q,R)}\}_{(Q,R)\in\Pauli^2}$ is a probability distribution
over the $16$ pairs of single-qubit Pauli operators given by
\begin{align}
q_{(I,I)}&=\frac{1-p/2}{1+p}\ ,\notag\\
q_{(I,P)}=q_{(P,I)}&=\frac{p}{4(1+p)}
\quad\textrm{ for }P\in\{X,Y,Z\}\ ,\notag\\
q_{(Q,R)}&=0
\quad\textrm{ for }Q,R\in\{X,Y,Z\}\ .
\label{eq:residual-pauli-probabilities}
\end{align}
\end{lemma}

\begin{figure}[H]
\centering
\begingroup
\small
\tikzset{
  rpwire/.style={line width=0.45pt},
  rppauli/.style={draw,fill=white,line width=0.45pt,inner sep=1.5pt,
    minimum width=0.50cm,minimum height=0.42cm},
  rpgate/.style={draw,fill=blue!12,line width=0.45pt,inner sep=0pt,
    minimum width=0.55cm,minimum height=0.96cm},
  rpnoise/.style={text=red,fill=white,inner sep=0pt,
    font=\small},
  rplabel/.style={font=\small}
}
\newcommand{\rpNoisyCircuit}{%
  \begin{tikzpicture}[baseline=-0.5ex,every node/.style={rplabel}]
    \foreach \y in {-0.26,0.26}{
      \draw[rpwire] (0,\y) -- (1.65,\y);
      \node[rpnoise] at (0.40,\y) {$\bigstar$};
    }
    \node[rpgate] at (1.06,0) {$U$};
  \end{tikzpicture}%
}
\newcommand{\rpProductCircuit}{%
  \begin{tikzpicture}[baseline=-0.5ex,every node/.style={rplabel}]
    \foreach \y in {-0.26,0.26}{\draw[rpwire] (0,\y) -- (1.20,\y);}
    \node[rppauli] at (0.60,0.26) {$Q$};
    \node[rppauli] at (0.60,-0.26) {$R$};
  \end{tikzpicture}%
}
\newcommand{\rpResidualCircuit}{%
  \begin{tikzpicture}[baseline=-0.5ex,every node/.style={rplabel}]
    \foreach \y in {-0.26,0.26}{\draw[rpwire] (0,\y) -- (1.85,\y);}
    \node[rppauli] at (0.44,0.26) {$Q$};
    \node[rppauli] at (0.44,-0.26) {$R$};
    \node[rpgate] at (1.23,0) {$U$};
  \end{tikzpicture}%
}
\(\displaystyle
  \rpNoisyCircuit
  =p^2\sum_{(Q,R)\in\Pauli^2}\frac{1}{16}\,\rpProductCircuit
    +(1-p^2)\sum_{(Q,R)\in\Pauli^2}
    q_{(Q,R)}\,\rpResidualCircuit\,.
\)
\endgroup
\caption{Circuit identity for Eq.~\eqref{eq:rpdecompositionbasic}, with
$\Pauli=\{I,X,Y,Z\}$. Each red star denotes the depolarizing channel~$\cN_p$.
Circuits denote channels on arbitrary inputs, with time running from left
to right. The first term, of total weight~$p^2$, contains only independent
single-qubit Pauli gates and cuts the interaction. The second term, of
total weight~$1-p^2$, retains~$U$. The probabilities inside the sums are
$1/16$ and $q_{(Q,R)}$ from Eq.~\eqref{eq:residual-pauli-probabilities}.}
\label{fig:rpdecomposition-circuit}
\end{figure}

Fig.~\ref{fig:rpdecomposition-circuit} schematically depicts this two-branch circuit identity.
We note that only the second term in Eq.~\eqref{eq:rpdecompositionbasic} depends on the unitary~$U$ considered.
\begin{proof}
Writing the depolarizing channel as
\begin{align}
\cN_p&=(1-p)\mathsf{id}+p\cF\qquad\textrm{ where }\qquad
\cF(\rho)=\tr(\rho)/2\cdot I\ ,
\end{align}
we have
\begin{align}
\cN_p^{\otimes 2}&=p^2\cF\otimes\cF+(1-p^2)\cG
\end{align}
where
\begin{align}
\cG=(1-p^2)^{-1}\left(
(1-p)^2\mathsf{id}\otimes\mathsf{id}
+p(1-p)\mathsf{id}\otimes\cF+p(1-p)\cF\otimes\mathsf{id}
\right)\ .
\end{align}
Since $\cF\otimes\cF$ outputs the maximally mixed state, which is invariant
under any unitary, we have
$\unitarychannel{U}\circ(\cF\otimes\cF)=\cF\otimes\cF$ and thus
\begin{align}
\cR_p=\unitarychannel{U}\circ\cN_p^{\otimes 2}
&=p^2\,\cF\otimes\cF+(1-p^2)\,\unitarychannel{U}\circ\cG\ .
\label{eq:rectangle-two-branches}
\end{align}
It remains to exhibit the two branches in the claimed form. Using
\begin{align}
\cF(\rho)&=\frac{1}{4}\sum_{P\in\Pauli}P\rho P^\dagger\ ,
\label{eq:F-pauli-twirl}
\end{align}
the first branch is
\begin{align}
\cF\otimes\cF&=\frac{1}{16}\sum_{(Q,R)\in\Pauli^2}
\unitarychannel{Q}\otimes\unitarychannel{R}\ ,
\end{align}
i.e., the uniform mixture over the $16$ Pauli pairs $(Q,R)$ of product
unitary channels. For the second branch,
substituting Eq.~\eqref{eq:F-pauli-twirl} into the two mixed terms of~$\cG$ gives
\begin{align}
\unitarychannel{U}\circ(\mathsf{id}\otimes\cF)
&=\frac{1}{4}\sum_{R\in\Pauli}\unitarychannel{U(I\otimes R)}\ ,
\end{align}
a mixture of the unitary channels associated with $U(I\otimes R)$, and
similarly for $\unitarychannel{U}\circ(\cF\otimes\mathsf{id})$. Hence
\begin{align}
\unitarychannel{U}\circ\cG
&=\sum_{(Q,R)\in\Pauli^2}q_{(Q,R)}
\unitarychannel{U(Q\otimes R)}
\end{align}
with the probabilities defined in Eq.~\eqref{eq:residual-pauli-probabilities}.
Here the weight of $(I,I)$ collects the contribution $(1-p)^2$ of
$\mathsf{id}\otimes\mathsf{id}$ and the two identity terms of the mixed
contributions, $(1-p)^2+2p(1-p)/4$, normalized by $(1-p^2)$; one checks
$q_{(I,I)}+6\cdot\frac{p}{4(1+p)}=1$.
\end{proof}

\subsection{Convex decompositions of noisy brickwork circuits into product circuits\label{sec:convexdecompo}}
Applying Lemma~\ref{lem:circuitcutting} to every noisy gate of a brickwork circuit, and representing the remaining noise channels on the boundary qubits by random Paulis, expresses the noisy circuit as a mixture of unitary circuits. If the cut branch is selected at all $d/2$ gates crossing a given bond, no entangling operation acts across that bond, so that conditioned on the selected branches, the output distribution is a product of output distributions of unitary subcircuits (Lemma~\ref{lem:conditional-factorization}). The widths of these subcircuits are bounded with high probability (Lemma~\ref{lem:maxcircuitwidth}); this is what determines the runtime of our algorithm.

Let
\begin{align}
\gateLocations{\sU}=\{(t,j)\ |\ t\in [d]\textrm{ odd}, j\in [n/2]\}\cup\{(t,j)\ |\ t\in [d]\textrm{ even}, j\in [n/2-1]\}
\end{align}
be the set of gate locations of the circuit~$\sU$.
For a gate location~$(t,j)\in \gateLocations{\sU}$, let 
\begin{align} 
	\cR_p(t,j) = 
	\unitarychannel{U_j^{(t)}} \circ \cN_p^{\otimes 2} \label{eq:rectanglepchannel}
\end{align} 
be the channel of the noisy gate at this location, consisting of the unitary gate~$U_j^{(t)}$ and the two-qubit noise channel immediately preceding it. It acts on $Q_{2j-1},Q_{2j}$ for odd~$t$ and on $Q_{2j},Q_{2j+1}$ for even~$t$. We call $\cR_p(t,j)$ a rectangle. The noisy implementation~$\cU^{(p)}$ of~$\sU$ is the composition of the channels~$\cR_p(t,j)$, together with the noise channels acting on the boundary qubits~$Q_1$ and~$Q_n$ before each even layer, which are the only noise locations belonging to no rectangle.

Each channel~$\cN_p$ is a mixture of the four single-qubit Pauli channels with weights $(1-3p/4,p/4,p/4,p/4)$; in \routine{DrawNoise} (Algorithm~\ref{alg:draw-noise}) below we draw one additional independent Pauli variable~$S^{(t)}_j$ for each of these $d$ locations, include it in the noise fixing~$F$, and absorb the resulting single-qubit gate into the adjacent odd-layer gate of the substituted circuit. These draws act within the first and last cut-bounded component and affect neither the cut indicators nor the component widths. 
Let the variables $E^{(t)}_j$ be independent with distribution $\mathsf{Ber}(p^2)$, selecting the two branches of Lemma~\ref{lem:circuitcutting} at each rectangle.
Define the random variable
\begin{align}
F=\Bigl(\bigl\{(E^{(t)}_j,Q^{(t)}_j,R^{(t)}_j)\bigr\}_{(t,j)\in \gateLocations{\sU}},\ \bigl\{S^{(t)}_j\bigr\}_{t\textrm{ even},\,j\in\{1,n\}}\Bigr)\ .\label{eq:noisefixing}
\end{align}
Here $Q^{(t)}_j,R^{(t)}_j\in\Pauli$ are the two Pauli operators selected at
the rectangle~$(t,j)$. Conditioned on $E^{(t)}_j=1$, their pair is uniform
on~$\Pauli^2$; conditioned on $E^{(t)}_j=0$, it equals $(Q,R)$ with
probability~$q_{(Q,R)}$.
We call a realization $F=f$ a noise fixing and note that $\tsU=\tsU(F)$ is a random unitary circuit.
The circuit~$\tsU(F)$ has gates $\tilde{U}^{(t)}_j=Q^{(t)}_j\otimes R^{(t)}_j$ when $E^{(t)}_j=1$ and $\tilde{U}^{(t)}_j=U_j^{(t)}(Q^{(t)}_j\otimes R^{(t)}_j)$ otherwise, together with the boundary Pauli gates~$S^{(t)}_j$.

The choice of~\smash{$\{E^{(t)}_j\}_{(t,j)\in\gateLocations{\sU}}$} determines where~$\tsU$ generates entanglement. Define the cut indicator variables with half-integer indices by
\begin{align}
E_{1/2+(2j-1)}&=\prod_{t\textrm{ odd}} E^{(t)}_j\qquad\textrm{ for }\qquad j\in [n/2]\\
E_{1/2+2j}&=\prod_{t\textrm{ even}} E^{(t)}_j\qquad\textrm{ for }\qquad j\in [n/2-1]\ .
\end{align}
Then the following holds for any $k\in [n-1]$: If $E_{1/2+k}=1$, the circuit~$\tsU$ does not generate entanglement between the qubits~$Q_k$ and $Q_{k+1}$ since they never see a two-qubit gate. 

The random variables~$\{E_{1/2+k}\}_{k\in [n-1]}$ are independent and identically distributed: $E_{1/2+k}$ is the product of the $d/2$ rectangle variables of the gates crossing the boundary between $Q_k$ and~$Q_{k+1}$, these sets of rectangle variables are pairwise disjoint for distinct~$k$, and each rectangle variable is an independent $\mathsf{Ber}(p^2)$ variable. Hence
\begin{align}
  \Pr[E_{1/2+k}=1]=(p^2)^{d/2}=p^d=:q\ , \label{eq:q-pd}
\end{align}
i.e., $E_{1/2+k}\sim \mathsf{Ber}(q)$ for each $k\in [n-1]$.
Define 
\begin{align}
\{ C_1<C_2<\cdots<C_{\numCuts}\}=\left\{1/2+k\ |\ E_{1/2+k}=1\textrm{ for }k\in [n-1]\right\}\ .
\end{align}
We call each $C_s$ for $s\in \{1,\ldots,\numCuts\}$ a (location of a) cut. Note that the number~$\numCuts$ of cuts is a binomially distributed random variable~$\numCuts\sim\mathsf{Bin}(n-1,q)$. 
It will be convenient to also set
\begin{align}
  C_0=1/2\qquad\textrm{ and }  \qquad
  C_{\numCuts+1}=1/2+n\ .
\end{align}
Then the cuts define a partition of $[n]$ (the set of qubits) 
as 
\begin{align}
[n]=\bigcup_{r=0}^{\numCuts} J_r\qquad\textrm{ where }\qquad J_r=[C_r,C_{r+1}]\cap \mathbb{N}
\end{align}
into~$\numCuts+1$ contiguous pairwise disjoint intervals~$\{J_r\}_{r=0}^{\numCuts}$.
We will refer to each such $J_r$ as a cut-bounded component.
 
Observe that by definition, any two-qubit gate in the circuit~$\tsU$ acts on a pair of qubits belonging to a single interval only. In particular, the circuit~$\tsU$ factors as
\begin{align}
\tsU&=\bigotimes_{r=0}^{\numCuts} \tsU_r\label{eq:factorizationcircuit}
\end{align}
into $\numCuts+1$ parallel subcircuits~$\tsU_0,\ldots,\tsU_{\numCuts}$. For an $n$-bit string $z=(z_1,\ldots,z_n)\in \{0,1\}^n$ and a subset $J\subseteq [n]$, let us write $z_J=\{z_j\}_{j\in J}$ for the substring associated with bits belonging to~$J$.
 Eq.~\eqref{eq:factorizationcircuit} implies the following.
\begin{lemma}[Conditional factorization]
\label{lem:conditional-factorization}
Let 
\begin{align}
f&=\Bigl(\bigl\{(e^{(t)}_j,q^{(t)}_j,r^{(t)}_j)\bigr\}_{(t,j)\in \gateLocations{\sU}},\ \bigl\{s^{(t)}_j\bigr\}_{t\textrm{ even},\,j\in\{1,n\}}\Bigr)\  \label{eq:noisefixingspecific}
\end{align}
be any noise fixing, where $e^{(t)}_j\in\{0,1\}$ and
$q^{(t)}_j,r^{(t)}_j,s^{(t)}_j\in\Pauli$ are fixed.
 Conditioned on~$F=f$, the distribution~$P^{(p)}$ factors as 
 \begin{align}
 P^{(p)}(z|F=f)&=\prod_{r=0}^{\numCuts} P_{\tsU_r}(z_{J_r})\qquad\textrm{ where }\qquad z\in \{0,1\}^n\ ,
 \end{align}
 Here $P_{\tsU_r}$, defined by
 \begin{align}
 P_{\tsU_r}(y)&=\big|\langle y | \tsU_r | 0^{W_r}\rangle \big|^2\qquad\textrm{ for }\qquad y\in \{0,1\}^{W_r}
 \end{align}
 with $W_r=|J_r|$, is the output distribution of the subcircuit~$\tsU_r$.
\end{lemma}
\begin{proof}
By Lemma~\ref{lem:circuitcutting}, each rectangle channel $\cR_p(t,j)$ is the mixture, over the sampled triple $(E^{(t)}_j,Q^{(t)}_j,R^{(t)}_j)$ in \routine{DrawNoise} (Algorithm~\ref{alg:draw-noise}), of the unitary channels associated with the substituted gates~$\tilde{U}^{(t)}_j$; each boundary noise channel is likewise the mixture, over the draw of the corresponding variable~$S^{(t)}_j$ of the noise fixing, of the four single-qubit Pauli channels, absorbed into the adjacent odd-layer gates as described in Section~\ref{sec:circuit-cutting}. Composing over all rectangles and boundary locations, and noting that $F$ collects these independent rectangle and boundary draws, the noisy circuit is the average
\begin{align}
\cU^{(p)}(\rho)&=\mathbb{E}_F\left[\tsU(F)\,\rho\,\tsU(F)^\dagger\right]\ ,
\end{align}
so conditioned on $F=f$ the pre-measurement state is the pure state
$\tsU(f)|0^n\rangle$, with output distribution
\begin{align}
P^{(p)}(z\,|\,F=f)=|\langle z|\tsU(f)|0^n\rangle|^2.
\end{align}
By Eq.~\eqref{eq:factorizationcircuit}, $\tsU=\bigotimes_{r=0}^{\numCuts}\tsU_r$ with
$\tsU_r$ supported on the qubits of~$J_r$, and
$|0^n\rangle=\bigotimes_{r=0}^{\numCuts}|0^{W_r}\rangle$. Hence
\begin{align}
\langle z|\tsU(f)|0^n\rangle=\prod_{r=0}^{\numCuts}\langle z_{J_r}|\tsU_r|0^{W_r}\rangle\ ,
\end{align}
and taking absolute squares gives the claim.
\end{proof}


Let us argue that with probability at least~$1-\delta$, the
width~$W_r$ of each cut-bounded component is upper bounded.

\begin{lemma}[Maximum width]\label{lem:maxcircuitwidth}
Let $\delta\in (0,1)$ and set $q = p^d$. Define
\begin{align}
  W_*=\left\lceil \frac{\log(n/\delta)}{\log(1/(1-q))}\right\rceil +1\ .
  \label{eq:Wstar-choice}
\end{align}
Then
\begin{align}
  \Pr\left[W_r<W_*\textrm{ for all }r\in\{0,\ldots,\numCuts\}\right]\geq 1-\delta\ .
\end{align}
Furthermore, for constant $p\in(0,1)$,
\begin{align}
  W_* = \Theta\bigl(p^{-d}\log(n/\delta)\bigr).
  \label{eq:Wstar-scaling}
\end{align}
The probability guarantee also holds for any larger integer width.
\end{lemma}
\begin{proof}
Set $W_{\max}=\max_{0\leq r\leq \numCuts}W_r$, the maximum width of a cut-bounded component, and define the event
$\mathsf{Good}=\{W_{\max}< W_*\}$. 
A cut-bounded component of width~$W_r\geq W_*$ consists of at least $W_*$ consecutive sites, and the half-integer positions strictly between consecutive sites of a component carry no cut. Hence if $W_{\max}\geq W_*$, there are $W_*-1$ consecutive internal half-integer positions containing no cut. There are fewer than $n$ possible starting positions for such a run, and each fixed run of $W_*-1$ internal positions contains no cut with probability $(1-q)^{W_*-1}$. The union bound gives
\begin{align}
  \Pr[W_{\max}\geq W_*]
  \leq n(1-q)^{W_*-1}\ .
  \label{eq:gap-bound}
\end{align}
We have 
  $n(1-q)^{W_*-1}\leq \delta$ by  definition of~$W_*$, 
and hence $\Pr[\Good]\geq 1-\delta$. Finally, the standard bounds
\begin{align}
  q
  \leq \log\!\left(\frac{1}{1-q}\right)
  \leq \frac{q}{1-q}
\end{align}
give
\begin{align}
  (1-q)q^{-1}\log(n/\delta)+1
  \leq W_*
  \leq q^{-1}\log(n/\delta)+2.
  \label{eq:Wstar-two-sided}
\end{align}
Since $q=p^d\leq p$, the lower bound is at least
$(1-p)p^{-d}\log(n/\delta)+1$, while the upper bound is
$p^{-d}\log(n/\delta)+2$. For constant $p\in(0,1)$, this proves
Eq.~\eqref{eq:Wstar-scaling}. Increasing the width only decreases the
probability of an oversized component, proving the final assertion.
\end{proof}

\subsection{Sampling the noise fixing\label{sec:cutboundedcompontentssampling}}
Lemmas~\ref{lem:conditional-factorization} and~\ref{lem:maxcircuitwidth} reduce the weak simulation of a noisy brickwork circuit to two steps: sampling a noise fixing~$f$ from the distribution of~$F$, and sampling from the output distributions of the component subcircuits~$\tsU_r$, which are unitary circuits of width less than~$W_*$ with probability at least~$1-\delta$. The first step is carried out by the subprogram \routine{DrawNoise} (Algorithm~\ref{alg:draw-noise}), which applies Lemma~\ref{lem:circuitcutting} at every rectangle and returns both the noise fixing and the unitary circuit~$\tsU(f)$. The second step is exact: by Lemma~\ref{lem:conditional-factorization}, once the noise is fixed, the components can be sampled independently and in parallel, using the same noise fixing~$f$ in all calls to the local sampler \routine{LocalSample} (Algorithm~\ref{alg:local-marginal-sampler}) of Section~\ref{sec:circuitsampling}. Since the components have random positions and sizes, our algorithm instead samples the marginals of~$\tsU(f)$ on fixed intervals and assembles a sample from these; this is the subject of Section~\ref{sec:marginalsamples}.

Throughout the pseudocode, \textbf{Require} specifies the inputs and
preconditions, and \textbf{Ensure} specifies the returned values and their
guarantees. Randomized calls use fresh independent randomness unless a shared
noise fixing is explicitly supplied. Loops marked ``in parallel'' execute
their independent iterations simultaneously; successive tree levels remain
sequential.

\begin{algorithm}[!htb]
\caption{\routine{DrawNoise}: sample a noise fixing and construct its unitary circuit.}
\label{alg:draw-noise}
\small
\begin{algorithmic}[1]
\Require A depth-$d$ brickwork circuit $\sU$ on $n$ qubits and a noise probability $p\in(0,1)$; $\gateLocations{\sU}$ is its set of gate locations.
\Ensure A noise fixing $f$ distributed as $F$ in Eq.~\eqref{eq:noisefixing}, and the corresponding unitary circuit $\tsU(f)$.
\Function{DrawNoise}{$\sU,p$}
\ForAll{$(t,j)\in \gateLocations{\sU}$ \textbf{in parallel}}\label{algline:noise-fixing-start}
  \State Draw $E^{(t)}_j\sim \mathsf{Ber}(p^2)$ independently.
  \If{$E^{(t)}_j=1$}
    \State Draw $(Q^{(t)}_j,R^{(t)}_j)$ uniformly from $\Pauli^2$.
    \State Set $\tilde{U}^{(t)}_j\gets Q^{(t)}_j\otimes R^{(t)}_j$.
  \Else
    \State Draw $(Q^{(t)}_j,R^{(t)}_j)\in\Pauli^2$ with probabilities $q_{(Q,R)}$ from Eq.~\eqref{eq:residual-pauli-probabilities}.
    \State Set $\tilde{U}^{(t)}_j\gets U_j^{(t)}(Q^{(t)}_j\otimes R^{(t)}_j)$.
\EndIf
\EndFor
\ForAll{$t \in [d]$ even and $j\in\{1,n\}$ \textbf{in parallel}}
    \State Draw $S^{(t)}_j\in\Pauli$ independently with probabilities $(1-3p/4,\,p/4,\,p/4,\,p/4)$.
\EndFor \label{algline:noise-fixing-end}
\State Collect the drawn variables into $f$ as in Eq.~\eqref{eq:noisefixing}.
\ForAll{$t \in [d]$ \textbf{in parallel}}
       \If{$t$ odd}
	\State Define $\tilde{L}_t = \bigotimes_{j=1}^{n/2} \big(\tilde{U}_j^{(t)}\big)_{Q_{2j-1}Q_{2j}}$.
       \Else
       	\State Define $\tilde{L}_t = \big( S_1^{(t)} \big)_{Q_1} \otimes \left( \bigotimes_{j=1}^{n/2-1} \big( \tilde{U}_j^{(t)} \big)_{Q_{2j}Q_{2j+1}} \right) \otimes \big( S_n^{(t)} \big)_{Q_n}$. 
       \EndIf
\EndFor
\State Set $\tsU = \tilde{L}_d \tilde{L}_{d-1} \ldots \tilde{L}_1$.
\State \Return $(f,\tsU)$.
\EndFunction
\end{algorithmic}
\end{algorithm}

\section{Combining marginal samples\label{sec:marginalsamples}}
Fix a noise realization~$f$ as defined in Section~\ref{sec:convexdecompo}.
Our task is to generate an $n$-bit sample from the conditional distribution
\begin{align}
  Q_Z(z):=P^{(p)}(z\mid F=f).
\end{align}
Throughout this section, all distributions and independence statements
refer to this same fixed noise realization. By
Lemma~\ref{lem:conditional-factorization}, this distribution factorizes over
the cut-bounded components as
\begin{align}
  Q_Z=\prod_{s=0}^{\numCuts}Q_{Z_{J_s}}\ ,\label{eq:qzdistribm}
\end{align}
where each $J_s\subseteq [n]$ is a contiguous interval of qubits defined by the noise realization~$f$ (see Section~\ref{sec:convexdecompo}). 

We show that a sample from~$Q_Z$ can be assembled from independent samples of the marginals of~$Q_Z$ on two families of fixed intervals, chosen independently of the noise realization (Sections~\ref{sec:combining-independent-marginals} and~\ref{sec:two-forests}), and we give a parallel selection procedure which performs the assembly using the cut configuration determined by~$f$ (Sections~\ref{sec:subsetsamplingcombination} and~\ref{sec:boundarypass-algorithm}, with its message-passing implementation in Section~\ref{sec:messagepassingalgorithmcontainment}). The marginal samples themselves are produced by the algorithm \routine{LocalSample} of Section~\ref{sec:circuitsampling}.

\subsection{Sampling by combining independent marginal samples}
\label{sec:combining-independent-marginals}
The key observation is that a marginal sample on an interval containing
several cut-bounded components in their entirety already supplies independent, correctly
distributed samples for all those components, since discarding the other bits
does not change their joint distribution. Thus we may assign each
component to one interval containing it and copy its entire substring
from that interval's sample, provided that the assignment depends only on the noise
realization and the interval geometry, not on the sampled values.

\begin{figure}[H]
\centering
\begin{tikzpicture}[x=0.48cm,y=0.8cm,>=Latex]
  \tikzset{cut/.style={red,line width=1.1pt}, compA/.style={blue!60!black,line width=2.1pt}, compB/.style={green!45!black,line width=2.1pt}, winA/.style={draw=blue!50!black,line width=0.8pt,fill=blue!8,rounded corners=2pt}, winB/.style={draw=green!45!black,line width=0.8pt,fill=green!10,rounded corners=2pt}}
  \draw[line width=0.9pt] (0,0) -- (24,0);
  \foreach \x in {0,5,10,14,18,22,24}{
    \draw[cut] (\x,-0.32) -- (\x,0.32);
  }
  \draw[compA] (0,0) -- (5,0);
  \draw[compB] (5,0) -- (10,0);
  \draw[compA] (10,0) -- (14,0);
  \draw[compB] (14,0) -- (18,0);
  \draw[compA] (18,0) -- (22,0);
  \draw[compA] (22,0) -- (24,0);
  \node[below,font=\scriptsize,blue!60!black] at (2.5,-0.72) {$A_0$};
  \node[below,font=\scriptsize,green!45!black] at (7.5,-0.72) {$B_1$};
  \node[below,font=\scriptsize,blue!60!black] at (12,-0.72) {$A_1$};
  \node[below,font=\scriptsize,green!45!black] at (16,-0.72) {$B_2$};
  \node[below,font=\scriptsize,blue!60!black] at (20,-0.72) {$A_2$};
  \node[below,font=\scriptsize,blue!60!black] at (23,-0.72) {$A_2$};
  \node[above,font=\scriptsize,blue!60!black] at (2.5,0.18) {$J_0$};
  \node[above,font=\scriptsize,green!45!black] at (7.5,0.18) {$J_1$};
  \node[above,font=\scriptsize,blue!60!black] at (12,0.18) {$J_2$};
  \node[above,font=\scriptsize,green!45!black] at (16,0.18) {$J_3$};
  \node[above,font=\scriptsize,blue!60!black] at (20,0.18) {$J_4$};
  \node[above,font=\scriptsize,blue!60!black] at (23,0.18) {$J_5$};
  \foreach \x/\y/\lab in {0/8/$A_0$,8/16/$A_1$,16/24/$A_2$}{
    \path[winA] (\x,0.82) rectangle (\y,1.28);
    \node[font=\scriptsize] at ({(\x+\y)/2},1.05) {\lab};
  }
  \foreach \x/\y/\lab in {0/4/$B_0$,4/12/$B_1$,12/20/$B_2$,20/24/$B_3$}{
    \path[winB] (\x,-1.28) rectangle (\y,-0.82);
    \node[font=\scriptsize] at ({(\x+\y)/2},-1.05) {\lab};
  }
  \draw[decorate,decoration={brace,mirror,amplitude=4pt}] (5,-1.45) -- (10,-1.45) node[midway,below=4pt,font=\scriptsize] {$J_1$ is assigned to $B_1$};
\end{tikzpicture}
\caption{The two overlapping interval partitions. Red vertical lines mark cuts, including the physical boundaries. Each component is copied from one local sample: blue components from the unshifted partition and green components, which cross an unshifted boundary, from the shifted partition.}
\label{fig:overlap-intervals}
\end{figure}
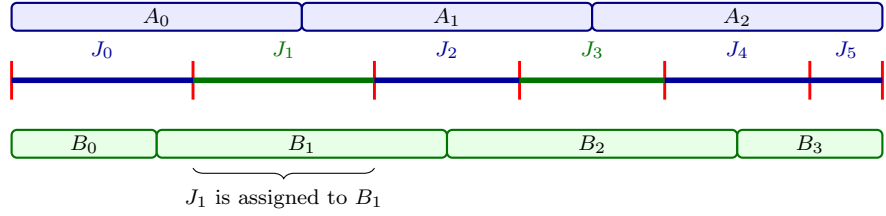

The following lemma formalizes the observation independently of the
choice of sampling intervals.

\begin{lemma}[Combining marginal samples]
\label{lem:combining-product-marginals}
Suppose $Q_Z=\prod_s Q_{Z_{J_s}}$ for the component partition~$\{J_s\}_s$.
Let $\{I_a\}_{a\in\mathcal A}$ be a finite indexed family of subsets of~$[n]$,
and let $V^{(a)}\sim Q_{Z_{I_a}}$ be mutually independent samples.
For each component~$J_s$, choose an index~$\alpha(s)\in\mathcal A$ such that
$J_s\subseteq I_{\alpha(s)}$, with the assignment depending only on~$f$
and the sets~$I_a$. Then the string defined by
\begin{align}
  Z_{J_s}=V^{(\alpha(s))}_{J_s}
  \qquad\textrm{for }s\in\{0,\ldots,\numCuts\}
\end{align}
has distribution~$Q_Z$.
\end{lemma}
\begin{proof}
For each source~$a$, let $D_a$ be the union of the components assigned to it.
The marginal of~$V^{(a)}$ on these components is
\begin{align}
  Q_{Z_{D_a}}=\prod_{s:\,\alpha(s)=a}Q_{Z_{J_s}},
\end{align}
because $D_a\subseteq I_a$ and~$Q_Z$ factorizes over the components.
The restrictions of distinct source samples are independent.
Since every component is assigned to exactly one source, concatenating
the selected substrings therefore gives the distribution
$\prod_s Q_{Z_{J_s}}=Q_Z$.
\end{proof}

It remains to choose short sampling intervals that cover every component
and to determine a valid source for every site in parallel. We first
establish the covering and assignment rules, then implement them by
\routine{BoundaryPass}.

\subsection{Choosing fixed sampling intervals}
\label{sec:two-forests}

Concretely, we use samples on two shifted interval partitions
$\{A_r\}_r$ and $\{B_t\}_t$, as illustrated in
Fig.~\ref{fig:overlap-intervals}. Concatenating the samples within each
partition gives two candidate strings~$X$ and~$Y$. Every sufficiently
short component is contained in a sampling interval of at least one
partition, and we select one source for that whole component.

Let $W_*\geq1$ be an integer. In Section~\ref{sec:full-routine} it will be
chosen at least as large as the value in Eq.~\eqref{eq:Wstar-choice}, so
that $W_{\max}<W_*$ with probability at least $1-\delta$ by
Lemma~\ref{lem:maxcircuitwidth}; the constructions of this section apply
to any integer $W_*\geq1$. Define
\begin{align}
  \ell=2W_*+4\ .
  \label{eq:ell-def}
\end{align}
We use two partitions 
\begin{align}
[n]&=\bigcup_{r=0}^{r_{\max}}A_r\qquad\textrm{ with }\qquad A_r\cap A_{r'}=\emptyset\textrm{ for }r\neq r'\\
[n]&=\bigcup_{s=0}^{s_{\max}}B_s\qquad\textrm{ with }\qquad B_s\cap B_{s'}=\emptyset\textrm{ for }s\neq s'
\end{align}
of $[n]$ into contiguous intervals.
 The unshifted intervals are
\begin{align}
  A_r=(r\ell,(r+1)\ell]\cap [n]
  \qquad\textrm{ for }\qquad
  r\in\left\{0,\ldots,r_{\max}\right\}\ 
  \label{eq:A-intervals}
\end{align}
where $r_{\max}=\lceil \frac{n}{\ell}\rceil -1$.
 The shifted intervals are
\begin{align}
  B_0&=(0,\ell/2]\cap [n]\\
  B_s&=\left(\left(s-\frac{1}{2}\right)\ell,\left(s+\frac{1}{2}\right)\ell\right]\cap [n]
  \qquad\textrm{ for }\qquad
  s\in\left\{1,\ldots,s_{\max}\right\}\ 
  \label{eq:B-intervals}
\end{align}
where $s_{\max}=\lceil\frac{n-\ell/2}{\ell}\rceil$. 

The construction is designed so that the two interval families jointly cover every short interval, in the following sense.
\begin{enumerate}[(i)]
\item\label{it:coveringproperty}
Covering. Every contiguous interval $J\subseteq[n]$ of length $|J|\leq\ell/2$ is contained in a single interval of one of the two families: either $J\subseteq A_r$ for some~$r$, or $J\subseteq B_s$ for some~$s$.
\item\label{it:lengthproperty}
Bounded length. Every interval of both families has length at most~$\ell$: $|A_r|\leq\ell$ and $|B_s|\leq\ell$ for all $r$ and~$s$.
\end{enumerate}
Property~\eqref{it:lengthproperty} is immediate.  For
property~\eqref{it:coveringproperty}, a short interval either crosses no
$A$-boundary and lies in one $A_r$, or crosses the unique boundary~$r\ell$ in
its range.  In the latter case its length bound places it inside the interval
$B_r=(r\ell-\ell/2,r\ell+\ell/2]$ centered at that boundary (with the evident
endpoint truncation).  Fig.~\ref{fig:overlap-intervals} illustrates the two
partitions and the resulting assignment of components.

Suppose  that we are given
a sample $X=(X_{A_0},\ldots,X_{A_{r_{\max}}})$ where the substring~$X_{A_r}$ is chosen independently according to
\begin{align}
X_{A_r}\sim Q_{Z_{A_r}}\qquad \textrm{for each }r\in \{0,\ldots,r_{\max}\}\ ,\label{eq:factorizationQdistributionA} 
\end{align}
and similarly, an independent sample 
$Y=(Y_{B_0},\ldots,Y_{B_{s_{\max}}})$ where the substring~$Y_{B_s}$ is chosen independently according to
\begin{align}
Y_{B_s}\sim Q_{Z_{B_s}}\qquad \textrm{for each }s\in \{0,\ldots,s_{\max}\}\  .\label{eq:factorizationQdistributionB} 
\end{align}
All these samples use the same fixed noise realization~$f$.
Since the interval samples are independent, the strings~$X$ and~$Y$ need
not be consistent on overlapping regions, and neither of them need have distribution~$Q_Z$, since a
partition may split a component between independently sampled intervals.
Lemma~\ref{lem:combining-product-marginals} applies once every whole
component has been assigned to one containing interval sample.

\subsection{Assigning components to interval samples}
\label{sec:subsetsamplingcombination}
We choose an unshifted sample whenever both bounding cuts of a component
are visible inside its interval, and otherwise use a shifted sample.
To express this rule, we associate with each cut-bounded component~$J_s$, $s\in \{0,\ldots,\numCuts\}$, its augmented component
\begin{align}
  \widehat{J}_s=\left[C_s-1/2,\, C_{s+1}+1/2\right]\cap [n]\ .
  \label{eq:augmented-component}
\end{align}
That is, $\widehat{J}_s$ is obtained from~$J_s$ by adding the site immediately left of the cut~$C_s$ (if $C_s>1/2$) and the site immediately right of the cut~$C_{s+1}$ (if $C_{s+1}<1/2+n$). If $W_{\max}<W_*$, then
\begin{align}
  |\widehat{J}_s|\leq W_s+2\leq W_*+1=\ell/2-1\ ,
  \label{eq:augmented-component-length}
\end{align}
where $W_s=|J_s|$.

For each component~$J_s$, choose the sample on~$A_r$ if
$\widehat J_s\subseteq A_r$ for some~$r$. Otherwise, choose the sample
on the shifted interval~$B_t$ containing~$\widehat J_s$.
Under the assumption~$W_{\max}<W_*$, such a shifted interval exists by
Eq.~\eqref{eq:augmented-component-length} and the covering property.

Let $s(j)$ and $r(j)$ denote the component and unshifted interval
containing site~$j$, respectively. Define the selection flags by
\begin{align}
  \chi_j=1
  \qquad\textrm{if and only if}\qquad
  \widehat{J}_{s(j)}\subseteq A_{r(j)}.
  \label{eq:augmented-containment-claim}
\end{align}
The output is assembled by
\begin{align}
  Z_j=
  \begin{cases}
    X_j\qquad &\textrm{if }\chi_j=1,\\
    Y_j\qquad &\textrm{if }\chi_j=0,
  \end{cases}
  \qquad\textrm{for }j\in\{1,\ldots,n\}.
  \label{eq:local-selection-rule}
\end{align}

\begin{lemma}[Selection of whole components]
\label{lem:component-source-assignment}
Assume~$W_{\max}<W_*$. The flags in
Eq.~\eqref{eq:augmented-containment-claim} are constant on every
component, and the rule in Eq.~\eqref{eq:local-selection-rule} copies
each component from a single interval sample whose interval contains it.
\end{lemma}
\begin{proof}
If $\widehat J_s\subseteq A_r$ for some~$r$, every site of~$J_s$ lies
in that same interval and has flag~$1$. The component is therefore
copied wholly from~$X_{A_r}$.
Otherwise, every site of~$J_s$ has flag~$0$.
The augmented component has length at most~$\ell/2$ by
Eq.~\eqref{eq:augmented-component-length}, so the covering property
places it in one shifted interval~$B_t$. The component is then copied
wholly from~$Y_{B_t}$.
\end{proof}

Together with Lemma~\ref{lem:combining-product-marginals}, this proves
that the selection rule gives the target distribution.

We determine the flags by testing, within each unshifted interval,
whether there is a cut on each side of a site. The required subprogram
\routine{IntervalContainment}$_N$ acts on an interval of~$N$ sites with
cut indicators $\{E_{1/2+k}\}_{k=0}^{N}$, where the two boundary indicators decide whether components reaching the
interval endpoints are accepted; in the calls below, $N=m_r=|A_r|$.
The interval-containment function
\begin{align}
\begin{matrix}
  \mathrm{IntervalContainment}_N\colon\ & \bits^{N+1}&\to&\bits^{N}\\
&   \left\{E_{1/2+k}\right\}_{k=0}^{N} & \mapsto &\left\{\chi_j\right\}_{j=1}^{N}
\end{matrix}
\label{eq:interval-containment-map}
\end{align}
maps a cut configuration to an indicator on the integer sites $j\in\{1,\ldots,N\}$, defined by
\begin{align}
  \chi_j=
  \left(\bigvee_{k=0}^{j-1}E_{1/2+k}\right)
  \land
  \left(\bigvee_{k=j}^{N}E_{1/2+k}\right)
  \label{eq:interval-containment-output}
\end{align}
for each $j\in\{1,\ldots,N\}$. Thus $\chi_j=1$ precisely when site~$j$
has a cut strictly to its left and a cut strictly to its right.
Section~\ref{sec:messagepassingalgorithmcontainment} gives a
message-passing implementation of~\routine{IntervalContainment}$_N$ with parallel runtime
$O(\log(N+1))$ and $O(N)$ elementary real operations (Theorem~\ref{cor:interval-containment-circuit}); its messages and the selection flags are bits, represented by the real
constants~$0$ and~$1$.

\subsection{The BoundaryPass algorithm}
\label{sec:boundarypass-algorithm}
The subprogram~\routine{BoundaryPass} calls~\routine{IntervalContainment}$_{m_r}$
separately on each unshifted interval~$A_r$, where $m_r=|A_r|$. Artificial boundary indicators
are set to~$0$, so a component is accepted only when both bounding cuts
are visible inside that interval. At a physical endpoint of the chain,
the corresponding boundary indicator is set to~$1$.

\begin{algorithm}[H]
\caption{\routine{BoundaryPass}: assemble interval samples by selecting one source per component.}
\label{alg:boundary-selection}
\small
\begin{algorithmic}[1]
\Require Internal cut indicators $E=\{E_{1/2+k}\}_{k=1}^{n-1}$, an even interval length $\ell\geq2$, and strings $X,Y\in\bits^n$ on the partitions in Eqs.~\eqref{eq:A-intervals}--\eqref{eq:B-intervals}.
\Ensure A string $Z\in\bits^n$ with $Z_j=X_j$ when $\chi_j=1$, and $Z_j=Y_j$ otherwise. Under the hypotheses of Theorem~\ref{thm:local-selection}, $Z\sim Q_Z$.
\Function{BoundaryPass}{$E,\ell,X,Y$}
\ForAll{$r=0,\ldots,\left\lceil n/\ell\right\rceil-1$ \textbf{in parallel}}
  \State Let $A_r=(r\ell,(r+1)\ell]\cap[n]=\{a_r^{\min},\ldots,a_r^{\max}\}$.
  \State $m_r\gets a_r^{\max}-a_r^{\min}+1$.
  \State $\tilde{E}_{1/2}\gets 0$ and $\tilde{E}_{1/2+m_r}\gets 0$.
  \ForAll{$k=1,\ldots,m_r-1$ \textbf{in parallel}}
    \State $\tilde{E}_{1/2+k}\gets E_{1/2+a_r^{\min}-1+k}$.
  \EndFor
  \If{$r=0$}
    \State $\tilde{E}_{1/2}\gets 1$.
  \EndIf
  \If{$r=\left\lceil n/\ell\right\rceil-1$}
    \State $\tilde{E}_{1/2+m_r}\gets 1$.
  \EndIf
  \State $(\chi_{a_r^{\min}},\ldots,\chi_{a_r^{\max}})\gets\Call{IntervalContainment\ensuremath{_{m_r}}}{\{\tilde{E}_{1/2+k}\}_{k=0}^{m_r}}$.
\EndFor
\ForAll{$j=1,\ldots,n$ \textbf{in parallel}}
  \If{$\chi_j=1$}
    \State $Z_j\gets X_j$.
  \Else
    \State $Z_j\gets Y_j$.
  \EndIf
\EndFor
\State \Return $\{Z_j\}_{j=1}^{n}$.
\EndFunction
\end{algorithmic}
\end{algorithm}

\begin{theorem}[Correctness of \routine{BoundaryPass}]
\label{thm:local-selection}
Let $W_*\geq1$ be an integer, let $\ell=2W_*+4$, and assume $W_{\max}<W_*$. Let $X$ and $Y$ be sampled independently as in Eqs.~\eqref{eq:factorizationQdistributionA} and~\eqref{eq:factorizationQdistributionB} from the marginals of~$Q_Z$. Then the output~$Z$ of \textnormal{\routine{BoundaryPass}} (Algorithm~\ref{alg:boundary-selection}) has distribution~$Q_Z$.
\end{theorem}
\begin{proof}[Proof of Theorem~\ref{thm:local-selection}]
For each~$A_r$, the cut string supplied to~\routine{IntervalContainment}$_{m_r}$
consists of the physical cuts strictly inside~$A_r$, together with a
boundary cut at each physical endpoint of the chain belonging to~$A_r$.
By Eq.~\eqref{eq:interval-containment-output}, the flag at site~$j$ is~$1$
exactly when a cut is visible on each side of~$j$.
If a more distant cut is visible on either side, the nearest bounding
cut of~$J_{s(j)}$ is visible as well. Thus the two bounding cuts are
visible exactly when
$\widehat J_{s(j)}\subseteq A_{r(j)}$, with the same statement holding
at the physical endpoints. The algorithm therefore computes the flags
in Eq.~\eqref{eq:augmented-containment-claim} and applies
Eq.~\eqref{eq:local-selection-rule}.

By Lemma~\ref{lem:component-source-assignment}, it assigns each whole
component to one containing interval sample. This assignment depends
only on the fixed cuts, and all the interval samples are independent.
Lemma~\ref{lem:combining-product-marginals} therefore gives the output
distribution~$Q_Z$.
\end{proof}

\begin{corollary}[Complexity of \routine{BoundaryPass}]
\label{cor:local-selection-circuit}
Given the cut indicators $\{E_{1/2+k}\}_{k=1}^{n-1}$ and candidate strings
$X,Y\in\bits^n$, \textnormal{\routine{BoundaryPass}} (Algorithm~\ref{alg:boundary-selection}), with
$\ell=2W_*+4$, runs in parallel runtime $O(\log(W_*+1))$ and uses
$O(n)$ elementary real operations. In particular, for $W_*$ as in
Eq.~\eqref{eq:Wstar-choice} we have
$W_*=O\bigl(p^{-d}\log(n/\delta)\bigr)$, and the parallel runtime is
$O\!\left(\log\!\left(2+p^{-d}\log(n/\delta)\right)\right)$.
\end{corollary}

\begin{proof}[Proof of Corollary~\ref{cor:local-selection-circuit}]
For each interval~$A_r$, copy the cut indicators and set the two boundary
bits. The boundary values depend only on~$r$. These operations take
constant parallel runtime and $O(n)$ elementary real operations in total.

Run \routine{IntervalContainment}$_{m_r}$ on all intervals~$A_r$ in parallel.
Each call has input length $m_r\leq\ell=2W_*+4$, so
Theorem~\ref{cor:interval-containment-circuit} gives parallel runtime
$O(\log(m_r+1))=O(\log(W_*+1))$ and $O(m_r)$ operations.
Since $\sum_r m_r=n$, the total number of operations is $O(n)$.
Lemma~\ref{lem:maxcircuitwidth} gives the bound in terms of~$n$.

Finally, select $X_j$ or $Y_j$ at every site in parallel. This adds
constant parallel runtime and $n$ operations.
\end{proof}

\section{Message passing for interval containment}
\label{sec:messagepassingalgorithmcontainment}
We give the implementation of the subprogram \routine{IntervalContainment}$_N$ used
in Section~\ref{sec:boundarypass-algorithm}, which computes the flags of
Eq.~\eqref{eq:interval-containment-output} from the $N+1$ cut indicators $\{E_{1/2+k}\}_{k=0}^{N}$ of an interval of~$N$ sites. It uses two passes on a balanced binary tree (Algorithm~\ref{alg:interval-containment}): an upward pass records whether each subtree
contains a cut, and a downward pass supplies the cuts outside that subtree.

\begin{algorithm}[!htbp]
\caption{\routine{IntervalContainment}$_N$: test interval containment by two passes on a binary tree.}
\label{alg:interval-containment}
\small
\begin{algorithmic}[1]
\Require A vector $E=\{E_{1/2+k}\}_{k=0}^{N}\in\bits^{N+1}$ of cut indicators, including both interval boundaries; $N\geq1$.
\Ensure $\chi\in\bits^N$, with $\chi_j=1$ exactly when there is a cut on each side of site $j$, as in Eq.~\eqref{eq:interval-containment-output}.
\Function{IntervalContainment\ensuremath{_N}}{$E$}
\State Set $\ell\gets \lceil \log_2(N+1)\rceil$.
\State Build a complete balanced binary tree whose ordered leaves are indexed by $0,1,\ldots,2^\ell-1$.
\ForAll{leaves $k\in\{0,\ldots,N\}$ \textbf{in parallel}}
  \State $h_k\gets E_{1/2+k}$.
\EndFor
\ForAll{dummy leaves $k\in\{N+1,\ldots,2^\ell-1\}$ \textbf{in parallel}}
  \State $h_k\gets 0$.
\EndFor
\ForAll{internal nodes $v$ in bottom-up order, \textbf{in parallel} within each level}
  \State Let $u$ and $w$ be the left and right children of~$v$.
  \State $h_v\gets h_u\lor h_w$.
\EndFor
\State $(L_{\mathrm{root}},R_{\mathrm{root}})\gets(0,0)$.
\ForAll{internal nodes $v$ in top-down order, \textbf{in parallel} within each level}
  \State Let $u$ and $w$ be the left and right children of~$v$.
  \State Send $(L_v,\ R_v\lor h_w)$ to~$u$.
  \State Send $(L_v\lor h_u,\ R_v)$ to~$w$.
\EndFor
\ForAll{$j=1,\ldots,N$ \textbf{in parallel}}
  \State $\chi_j\gets L_j\land\left(E_{1/2+j}\lor R_j\right)$.
\EndFor
\State \Return $\chi=\{\chi_j\}_{j=1}^{N}$.
\EndFunction
\end{algorithmic}
\end{algorithm}

\begin{theorem}[Correctness and complexity of \routine{IntervalContainment}$_N$]
\label{cor:interval-containment-circuit}
Given the cut indicators, \textnormal{\routine{IntervalContainment}}$_N$
(Algorithm~\ref{alg:interval-containment}) returns the containment flags~$\chi$
specified by Eq.~\eqref{eq:interval-containment-output}. It has parallel runtime $O(\log(N+1))$ and uses $O(N)$ elementary real operations.
\end{theorem}

Fig.~\ref{fig:interval-containment-examples} shows two explicit message-passing examples for $N=7$.

\begin{proof}
Let $I(v)=[a_v,b_v]$ be the interval of leaf indices below a node~$v$.
The upward phase maintains the invariant
\begin{align}
  h_v=\bigvee_{\substack{k\in I(v)\\ k\leq N}}E_{1/2+k},
\end{align}
with dummy leaves contributing~$0$.  In the downward phase, $L_v$ is the OR
of the cut indicators with index smaller than~$a_v$, while $R_v$ is the OR of
those with index larger than~$b_v$.  The root values are therefore~$(0,0)$,
and the two update rules in \routine{IntervalContainment}$_N$ (Algorithm~\ref{alg:interval-containment}) preserve
this invariant: the sibling summary supplies exactly the indices that become
external when passing to a child.

At leaf~$j$, the invariant gives
\begin{align}
  L_j=\bigvee_{k=0}^{j-1}E_{1/2+k},
  \qquad
  E_{1/2+j}\lor R_j=\bigvee_{k=j}^{N}E_{1/2+k}.
\end{align}
The output rule is therefore exactly Eq.~\eqref{eq:interval-containment-output}.
The tree is fixed by~$N$; it has $O(N)$ nodes and height
$O(\log(N+1))$. Each node performs a constant number of Boolean
operations in each pass, and each level is processed in parallel. Initialization and output take constant parallel
runtime and $O(N)$ operations, which gives the claimed bounds.
\end{proof}

Appendix~\ref{app:containment-bits} shows that the combining stage
introduces no finite-precision error. Appendix~\ref{app:noise-bits} gives the finite-bit
noise generation that supplies the cut indicators.

\newcommand{\intervalcontainmentfigure}{%
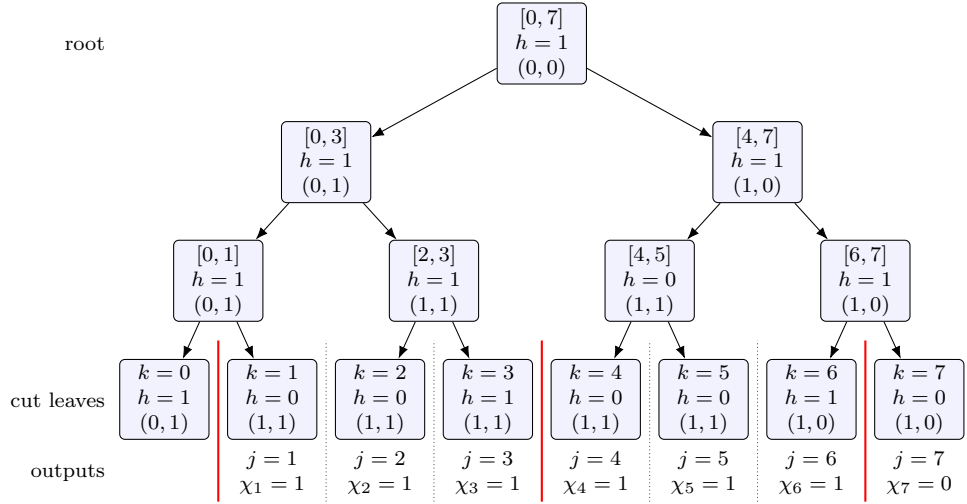
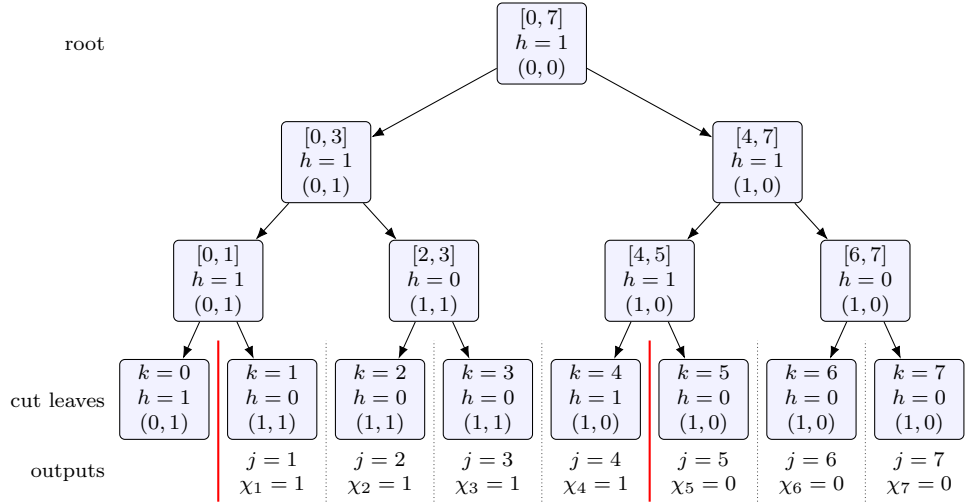
\begin{figure}[H]
\centering
\begin{subfigure}{\textwidth}
\centering
\resizebox{0.78\textwidth}{!}{%
\begin{tikzpicture}[
  x=1.45cm,
  y=1.1cm,
  every node/.style={font=\scriptsize},
  treenode/.style={draw,rounded corners=2pt,align=center,inner sep=2pt,minimum width=1.2cm,fill=blue!5},
  cutleaf/.style={treenode},
  outbit/.style={font=\scriptsize,align=center}
]
  \foreach \x in {1.5,2.5,4.5,5.5,7.5}{
    \draw[black!65,densely dotted] (\x,-1.25) -- (\x,0.72);
  }
  \foreach \x in {0.5,3.5,6.5}{
    \draw[red,thick] (\x,-1.25) -- (\x,0.72);
  }

  \node[cutleaf] (l0) at (0,0) {$k=0$\\$h=1$\\$(0,1)$};
  \node[treenode] (l1) at (1,0) {$k=1$\\$h=0$\\$(1,1)$};
  \node[treenode] (l2) at (2,0) {$k=2$\\$h=0$\\$(1,1)$};
  \node[cutleaf] (l3) at (3,0) {$k=3$\\$h=1$\\$(1,1)$};
  \node[treenode] (l4) at (4,0) {$k=4$\\$h=0$\\$(1,1)$};
  \node[treenode] (l5) at (5,0) {$k=5$\\$h=0$\\$(1,1)$};
  \node[cutleaf] (l6) at (6,0) {$k=6$\\$h=1$\\$(1,0)$};
  \node[treenode] (l7) at (7,0) {$k=7$\\$h=0$\\$(1,0)$};

  \node[treenode] (n01) at (0.5,1.45) {$[0,1]$\\$h=1$\\$(0,1)$};
  \node[treenode] (n23) at (2.5,1.45) {$[2,3]$\\$h=1$\\$(1,1)$};
  \node[treenode] (n45) at (4.5,1.45) {$[4,5]$\\$h=0$\\$(1,1)$};
  \node[treenode] (n67) at (6.5,1.45) {$[6,7]$\\$h=1$\\$(1,0)$};

  \node[treenode] (n03) at (1.5,2.9) {$[0,3]$\\$h=1$\\$(0,1)$};
  \node[treenode] (n47) at (5.5,2.9) {$[4,7]$\\$h=1$\\$(1,0)$};

  \node[treenode] (root) at (3.5,4.35) {$[0,7]$\\$h=1$\\$(0,0)$};

  \draw[-Latex] (root) -- (n03);
  \draw[-Latex] (root) -- (n47);
  \draw[-Latex] (n03) -- (n01);
  \draw[-Latex] (n03) -- (n23);
  \draw[-Latex] (n47) -- (n45);
  \draw[-Latex] (n47) -- (n67);
  \draw[-Latex] (n01) -- (l0);
  \draw[-Latex] (n01) -- (l1);
  \draw[-Latex] (n23) -- (l2);
  \draw[-Latex] (n23) -- (l3);
  \draw[-Latex] (n45) -- (l4);
  \draw[-Latex] (n45) -- (l5);
  \draw[-Latex] (n67) -- (l6);
  \draw[-Latex] (n67) -- (l7);

  \node[font=\scriptsize,anchor=east] at (-0.45,4.35) {root};
  \node[font=\scriptsize,anchor=east] at (-0.45,0) {cut leaves};

  \foreach \x/\j/\val in {
    1/1/1,2/2/1,3/3/1,4/4/1,5/5/1,6/6/1,7/7/0
  }{
    \node[outbit] at (\x,-0.9) {$j=\j$\\$\chi_{\j}=\val$};
  }
  \node[font=\scriptsize,anchor=east] at (-0.45,-0.9) {outputs};
\end{tikzpicture}%
}
\caption{$E_{1/2+0}=E_{1/2+3}=E_{1/2+6}=1$, giving $\chi=(1,1,1,1,1,1,0)$.}
\label{fig:interval-containment-example-three-cuts}
\end{subfigure}

\vspace{0.6em}

\begin{subfigure}{\textwidth}
\centering
\resizebox{0.78\textwidth}{!}{%
\begin{tikzpicture}[
  x=1.45cm,
  y=1.1cm,
  every node/.style={font=\scriptsize},
  treenode/.style={draw,rounded corners=2pt,align=center,inner sep=2pt,minimum width=1.2cm,fill=blue!5},
  cutleaf/.style={treenode},
  outbit/.style={font=\scriptsize,align=center}
]
  \foreach \x in {1.5,2.5,3.5,5.5,6.5,7.5}{
    \draw[black!65,densely dotted] (\x,-1.25) -- (\x,0.72);
  }
  \foreach \x in {0.5,4.5}{
    \draw[red,thick] (\x,-1.25) -- (\x,0.72);
  }

  \node[cutleaf] (l0) at (0,0) {$k=0$\\$h=1$\\$(0,1)$};
  \node[treenode] (l1) at (1,0) {$k=1$\\$h=0$\\$(1,1)$};
  \node[treenode] (l2) at (2,0) {$k=2$\\$h=0$\\$(1,1)$};
  \node[treenode] (l3) at (3,0) {$k=3$\\$h=0$\\$(1,1)$};
  \node[cutleaf] (l4) at (4,0) {$k=4$\\$h=1$\\$(1,0)$};
  \node[treenode] (l5) at (5,0) {$k=5$\\$h=0$\\$(1,0)$};
  \node[treenode] (l6) at (6,0) {$k=6$\\$h=0$\\$(1,0)$};
  \node[treenode] (l7) at (7,0) {$k=7$\\$h=0$\\$(1,0)$};

  \node[treenode] (n01) at (0.5,1.45) {$[0,1]$\\$h=1$\\$(0,1)$};
  \node[treenode] (n23) at (2.5,1.45) {$[2,3]$\\$h=0$\\$(1,1)$};
  \node[treenode] (n45) at (4.5,1.45) {$[4,5]$\\$h=1$\\$(1,0)$};
  \node[treenode] (n67) at (6.5,1.45) {$[6,7]$\\$h=0$\\$(1,0)$};

  \node[treenode] (n03) at (1.5,2.9) {$[0,3]$\\$h=1$\\$(0,1)$};
  \node[treenode] (n47) at (5.5,2.9) {$[4,7]$\\$h=1$\\$(1,0)$};

  \node[treenode] (root) at (3.5,4.35) {$[0,7]$\\$h=1$\\$(0,0)$};

  \draw[-Latex] (root) -- (n03);
  \draw[-Latex] (root) -- (n47);
  \draw[-Latex] (n03) -- (n01);
  \draw[-Latex] (n03) -- (n23);
  \draw[-Latex] (n47) -- (n45);
  \draw[-Latex] (n47) -- (n67);
  \draw[-Latex] (n01) -- (l0);
  \draw[-Latex] (n01) -- (l1);
  \draw[-Latex] (n23) -- (l2);
  \draw[-Latex] (n23) -- (l3);
  \draw[-Latex] (n45) -- (l4);
  \draw[-Latex] (n45) -- (l5);
  \draw[-Latex] (n67) -- (l6);
  \draw[-Latex] (n67) -- (l7);

  \node[font=\scriptsize,anchor=east] at (-0.45,4.35) {root};
  \node[font=\scriptsize,anchor=east] at (-0.45,0) {cut leaves};

  \foreach \x/\j/\val in {
    1/1/1,2/2/1,3/3/1,4/4/1,5/5/0,6/6/0,7/7/0
  }{
    \node[outbit] at (\x,-0.9) {$j=\j$\\$\chi_{\j}=\val$};
  }
  \node[font=\scriptsize,anchor=east] at (-0.45,-0.9) {outputs};
\end{tikzpicture}%
}
\caption{$E_{1/2+0}=E_{1/2+4}=1$ and all other $E_{1/2+k}=0$, giving $\chi=(1,1,1,1,0,0,0)$.}
\label{fig:interval-containment-example-two-cuts}
\end{subfigure}
\caption{Message passing for two cut strings. The vertical lines mark the half-integer positions $1/2+k$: solid red lines indicate $E_{1/2+k}=1$, and dotted black lines indicate $E_{1/2+k}=0$. For an internal vertex~$v$, write $I(v)=[a_v,b_v]$ for the associated interval of cut indices. Each internal vertex displays $I(v)$, the upward message~$h$, and the downward pair~$(L,R)$.}
\label{fig:interval-containment-examples}
\end{figure}
}

\intervalcontainmentfigure

\section{Exact local sampling}
\label{sec:circuitsampling}

Fix a noise realization~$f$ as defined in Section~\ref{sec:convexdecompo}.
To supply the input strings for \routine{BoundaryPass}
(Algorithm~\ref{alg:boundary-selection}) in Section~\ref{sec:marginalsamples}, we must
produce independent exact samples from the marginals of the distribution
\begin{align}
  Q^f_Z(z):=P^{(p)}(z\mid F=f)
  =
  \abs{\langle z|\tsU(f)|0^n\rangle}^2\ ,\qquad z\in \{0,1\}^n\ .
  \label{eq:fixed-noise-output}
\end{align}
This is the distribution denoted by $Q_Z$ in
Section~\ref{sec:marginalsamples}, with its dependence on~$f$ made explicit.
For a contiguous interval $A\subset[n]$, the target marginal is
\begin{align}
  Q^f_{Z_A}(a)
  =
  \sum_{z\in\bits^n:\ z_A=a} Q^f_Z(z),
  \qquad a\in\bits^A .
  \label{eq:local-marginal-target}
\end{align}
Let $\Lambda(A)$ be the backward light cone of~$A$ in~$\tsU(f)$.  Since the
circuit has depth~$d$ and nearest-neighbor gates, $\Lambda(A)$ is again a
contiguous interval and
\begin{align}
  |\Lambda(A)|\leq |A|+2d .
\end{align}
All gates outside this light cone can be discarded since the marginal distribution $Q^f_{Z_A}$ does not depend on these gates.  
 After relabelling
$\Lambda(A)$ as $[L]$ (where $L=|\Lambda(A)|$), we need to sample the full output of the resulting depth-$d$
circuit on $L$~qubits and retain only the sites corresponding to~$A$.  For an interval~$A$ as considered in Section~\ref{sec:two-forests}, we have $L\leq 2W_*+4+2d$.

\begin{algorithm}[!htbp]
\caption{\routine{LocalSample}: sample a marginal by restricting to its backward light cone.}
\label{alg:local-marginal-sampler}
\small
\begin{algorithmic}[1]
\Require A brickwork circuit $\sU$ on $n$ qubits, a noise fixing $f$ for that circuit, and a nonempty contiguous interval $A\subseteq[n]$.
\Ensure A string $z_A\in\bits^A$ with distribution $Q^f_{Z_A}=P^{(p)}_{Z_A}(\cdot\mid F=f)$.
\Function{LocalSample}{$\sU,f,A$}
\State Form the fixed unitary circuit $\tsU(f)$ specified by the noise realization.
\State Let $\Lambda(A)$ be the backward light cone of $A$ in $\tsU(f)$.
\State Discard all gates outside $\Lambda(A)$ and relabel $\Lambda(A)$ as $[L]$.
\State Coarse grain this circuit into a bilayer circuit $\sU_{\mathrm{bi}}$ using the divisibility convention of Section~\ref{sec:bilayer-reduction}; retain the original site map.
\State $z\gets\Call{TreeSample}{\sU_{\mathrm{bi}}}$.
\State Decode the qudit labels into qubit bits.
\State \Return the substring $z_A$ corresponding to the original interval $A$.
\EndFunction
\end{algorithmic}
\end{algorithm}
This is summarized by \routine{LocalSample}
(Algorithm~\ref{alg:local-marginal-sampler}). It recasts the restricted
circuit as a depth-two circuit and calls \routine{TreeSample}
(Algorithm~\ref{alg:joining-tree}). The latter produces samples by recursively joining neighboring
intervals along a balanced binary tree, using the resampling method of
Ref.~\cite{BravyiGossetLiu22}. In this section, we construct this subroutine
and analyze its parallel runtime and number of elementary operations
(Proposition~\ref{prop:bilayer-sampler}).

\subsection{Reduction to bilayer circuits}
\label{sec:bilayer-reduction}
\label{rem:bilayer-form}
The bilayer reduction groups consecutive blocks of physical qubits into
qudits and regroups the gates into two layers of nearest-neighbor
qudit gates. A block of $d$ qubits has dimension~$K=2^d$.
We call the resulting qudit gates coarse gates or coarse
unitaries, to distinguish them from the physical two-qubit gates
$U_j^{(t)}$ in Eq.~\eqref{eq:physical-gate-layers}.
Fig.~\ref{fig:coarsegrained} illustrates the reduction for the three-layer
circuit in Fig.~\ref{fig:bilayerqubits}: each block contains three qubits,
so $K=8$, and the gate groups of each color form one coarse layer.
In the noisy setting, the grouping is applied to the unitary
circuit~$\tsU(f)$ obtained after fixing the noise realization, so the
coarse gates depend on~$f$.

\begin{figure}[!htbp]
\centering
\includegraphics[width=0.4\textwidth]{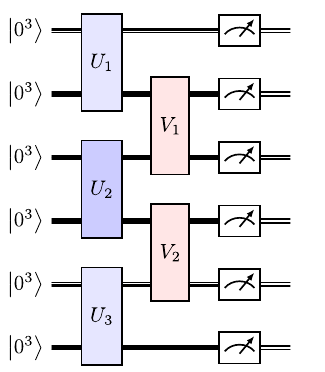}
\caption{Bilayer form of the circuit geometry in
Fig.~\ref{fig:bilayerqubits}. Each bundled wire represents three
consecutive physical qubits, initialized in $\ket{0^3}$, and hence one
qudit of dimension~$K=8$. After fixing the noise realization, the blue gate
groups form $U_1,U_2,U_3$ and the red groups form $V_1,V_2$; these coarse
unitaries depend on the noise fixing. Measuring a qudit in its
computational basis returns the three corresponding output bits.}
\label{fig:coarsegrained}
\end{figure}

\begin{samepage}
Let $K\geq 2$ and $M\in\bbN$. A bilayer circuit on $2M$ qudits has the form
\begin{align}
  \sU&=\mathsf{L}_2\mathsf{L}_1,
  \label{eq:basic-bilayer-circuit}
\end{align}
where the coarse layers are
\begin{align*}
  \mathsf{L}_1
  &=\bigotimes_{i=1}^{M}(U_i)_{Q_{2i-1}Q_{2i}},\\
  \mathsf{L}_2
  &=I_{Q_1}\otimes
    \left(\bigotimes_{i=1}^{M-1}(V_i)_{Q_{2i}Q_{2i+1}}\right)
    \otimes I_{Q_{2M}}.
\end{align*}
Here $Q_1,\ldots,Q_{2M}$ denote qudits of dimension~$K$, and
$U_i,V_i$ are two-qudit unitaries. The sans-serif notation
$\mathsf{L}_1,\mathsf{L}_2$ distinguishes these coarse layers from the
physical layers~$L_t$ of Eq.~\eqref{eq:physical-gate-layers}.
\end{samepage}
It defines a distribution on $[K]^{2M}$ by
\begin{align}
  P_{Z([2M])}(z)
  &:=
  \abs{\langle z|\sU|0^{2M}\rangle}^2,
  \qquad z\in[K]^{2M}.
  \label{eq:bilayer-output-distribution}
\end{align}
The goal is to sample this distribution exactly.

For the bilayer reduction, assume without loss of generality that the
number of physical qubits is divisible by~$2d$: otherwise consider a
larger circuit with fewer than $2d$ additional noninteracting qubits in
state~$|0\rangle$ and ignore their measurement outcomes. This convention
also applies separately to each local circuit in 
\textnormal{\routine{LocalSample}} (Algorithm~\ref{alg:local-marginal-sampler}).
For even~$d$ and width~$L=2Md$, grouping $d$ consecutive qubits into each
qudit gives Eq.~\eqref{eq:basic-bilayer-circuit} with $K=2^d$, by regrouping
the gates as illustrated in Fig.~\ref{fig:coarsegrained}.
Qudit measurement labels simply encode the corresponding qubit outcomes.
Each coarse gate contains at most $d^2$ physical gates, so assembling its
$K^2\times K^2$ matrix costs $\poly(2^d)$ elementary real operations, as
allowed in the bounds below.
We first sample the computational-basis outcomes of
$\mathsf{L}_1\ket{0^{2M}}$, independently for each two-qudit gate~$U_i$.
We then update this sample as the gates of~$\mathsf{L}_2$ are applied,
obtaining a sample from the full output distribution in
Eq.~\eqref{eq:bilayer-output-distribution}. Each update uses the following
resampling method.

\subsection{Sampling without computing marginals\label{sec:samplingwithoutmarginals}}
The method of Bravyi, Gosset, and
Liu~\cite{BravyiGossetLiu22} produces a sample from the measurement outcome distribution associated with a state~$U\ket{\Psi}$ for a local unitary~$U$, given a sample from the outcome distribution associated with~$\ket{\Psi}$.

To illustrate this approach, implemented by the routine \routine{Resample} (Algorithm~\ref{alg:resample}), consider a state~$\ket{\Psi}\in\cH_A\otimes\cH_B\cong \mathbb{C}^{D_A}\otimes\mathbb{C}^D$ on a bipartite Hilbert space with orthonormal basis~$\{\ket{j,k}\equiv \ket{j}\otimes\ket{k}\}_{j\in [D_A],k\in [D]}$.
Let $U$ be a unitary on~$\cH_A$.
Assume we have a sample $(y,z)\sim P$ from the distribution
\begin{align}
P(y,z)&=|\langle y,z|\Psi\rangle|^2\qquad\textrm{ for }\qquad (y,z)\in [D_A]\times [D]\ .
\end{align}
Our goal is to produce a sample $(y',z')\sim P'$ from  distribution
\begin{align}
P'(y,z)&=|\langle y,z|\Psi'\rangle|^2\qquad\textrm{ for }\qquad (y,z)\in [D_A]\times [D]\ 
\end{align} 
for the evolved state $\ket{\Psi'}=(U\otimes I_B)\ket{\Psi}$.
The subprogram \routine{Resample} (Algorithm~\ref{alg:resample}) keeps the
unaffected outcome $z$ and replaces only $y$, by sampling from the conditional distribution
\begin{align}
Q(u|z)&= \frac{P'(u,z)}{\sum_{v\in [D_A]}P'(v,z)}\qquad\textrm{ for }\qquad u\in [D_A]\ .\label{eq:uzexpression}
\end{align}
The unitary on $A$ leaves the marginal on $B$ unchanged, so the sampled
value of $z$ already has the required distribution.

\begin{algorithm}[!htbp]
\caption{\routine{Resample}: update the sampled subsystem after a local unitary.}
\label{alg:resample}
\small
\begin{algorithmic}[1]
\Require A sample $(y,z)\sim P$ on $[D_A]\times[D]$ and an evaluator for the evolved distribution $P'$; $P$ and $P'$ have the same marginal on $z$.
\Ensure A pair $(y',z)\sim P'$, retaining the supplied value of $z$.
\Function{Resample}{$P',(y,z)$}
\State Evaluate $P'(u,z)$ for all $u\in[D_A]$ and normalize to obtain $Q(u\mid z)$ from Eq.~\eqref{eq:uzexpression}.\label{it:steponesampling}
\State Draw $y'\sim Q(\cdot\mid z)$.\label{it:steptwosampling}
\State \Return $(y',z)$.\label{it:stepthreesampling}
\EndFunction
\end{algorithmic}
\end{algorithm}
In the real-arithmetic model, the number of elementary operations needed to realize \routine{Resample} is independent of~$D$, assuming
constant-cost access to the probabilities~$P'(u,z)$ for the fixed~$z$ and the $D_A$ outcomes $u\in[D_A]$.
In  our case, these probabilities are obtained from the vector of amplitudes $T(\cdot|z)\in\mathbb{C}^{D_A}$ of the current state,
\begin{align}
 T(y|z)=\langle y,z|\Psi\rangle\qquad\textrm{ for }\qquad y\in [D_A]\ ,
\end{align}
which we keep track of. Indeed, inserting a resolution $I_A=\sum_{s=1}^{D_A} \proj{s}$ of the identity on~$\cH_A$ gives
\begin{align}
 T'(y|z):=\langle y,z|(U\otimes I_B)|\Psi\rangle
 =\sum_{s=1}^{D_A} U_{y,s} T(s|z)\qquad\textrm{ for }\qquad y\in [D_A]\ ,
\end{align}
and thus 
\begin{align}
P'(y,z)=|T'(y|z)|^2\qquad\textrm{for }\qquad y\in [D_A]\  .\label{eq:pprimeyzdefinition}
\end{align}
Therefore, given a sample~$(y,z)\sim P$ and the vector
$T(\cdot|z)\in\mathbb{C}^{D_A}$, we compute $T'(\cdot|z)=UT(\cdot|z)$
and then call \routine{Resample}. Combining
Eqs.~\eqref{eq:pprimeyzdefinition} and~\eqref{eq:uzexpression} gives
\begin{align}
  Q(u\mid z)
  &=\frac{|T'(u|z)|^2}{\|T'(\cdot|z)\|_2^2}
   =\mathcal B\bigl(T'(\cdot|z)\bigr)_u .
  \label{eq:resample-born-distribution}
\end{align}
Here $\|\cdot\|_2$ is the Euclidean norm, and
$\mathcal B(a)_j:=|a_j|^2/\|a\|_2^2$ denotes the Born distribution
of a nonzero vector~$a$. The same normalized
amplitude formula gives the interface update in
Eq.~\eqref{eq:join-conditional-born} and is used in the finite-precision
analysis of Appendix~\ref{app:finite-precision-tree}.

For bilayer circuits, we compute  vectors analogous to $B\bigl(T'(\cdot|z)\bigr)\in \mathbb{R}^{D_A}$ from tables of boundary
amplitudes, defined below. Appendix~\ref{app:finite-precision-tree} gives
the finite-precision analysis, including word lengths and the resulting
total variation error.

\begin{samepage}
\subsection{Sampling from bilayer qudit circuits}
\label{sec:bilayer-subsec}
For each sampled interior string, we retain the amplitudes for all
boundary outcomes. These amplitudes allow \routine{JoinIntervals}
(Algorithm~\ref{alg:join-routine}) to
update the sample when two intervals are joined.
\end{samepage}

Throughout Sections~\ref{sec:bilayer-subsec}--\ref{sec:coarsening-subsec}, the symbols $I$ and~$J$ denote intervals of qudit indices in the bilayer chain; they are unrelated to the cut-bounded components~$J_r$ of Section~\ref{sec:circuit-cutting}. A valid interval is an even interval
\begin{align}
  I=\{2a+1,2a+2,\ldots,2b\},
  \qquad a<b .
\end{align}
Its boundary sites are $2a+1$ and $2b$, and its interior is
$I^\circ=\{2a+2,\ldots,2b-1\}$.  For intervals touching a physical endpoint,
we use the evident one-sided convention. The restricted interval state
$\ket{\Psi(I)}$ is obtained from $\ket{0^{|I|}}$ by applying all first-layer
gates~$U_i$ fully supported in~$I$, followed by all second-layer gates~$V_j$
fully supported in~$I^\circ$. The output distribution on~$I$ is
\begin{align}
  P_{Z^I}(z^I)=\abs{\langle z^I|\Psi(I)\rangle}^2 .
\end{align}
For a fixed interior string $z^I_{\circ}\in[K]^{I^\circ}$, define the boundary
amplitude table
\begin{align}
  T_I(z^I_-,z^I_+\given z^I_{\circ})
  =
  \langle z^I_-,z^I_{\circ},z^I_+|\Psi(I)\rangle,
  \qquad
  z^I_-,z^I_+\in[K].
  \label{eq:T-table}
\end{align}
The table has $K^2$ entries, or $K$ entries at a one-sided endpoint.
Fig.~\ref{fig:boundary-amplitude-tensors} depicts the tables for two
neighboring intervals and their union. Fig.~\ref{fig:bilayer-interval-states}
shows valid intervals and the restricted states that define these tables.

\begin{figure}[H]
\centering
\begin{subfigure}[t]{0.25\textwidth}
  \centering
  \includegraphics[width=\linewidth]{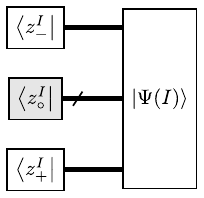}
  \caption{$T_I(z^I_-,z^I_+\mid z^I_{\circ})$.}
  \label{fig:TI-tensor}
\end{subfigure}
\hfill
\begin{subfigure}[t]{0.25\textwidth}
  \centering
  \includegraphics[width=\linewidth]{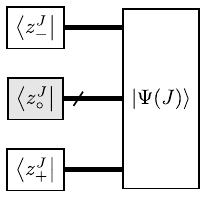}
  \caption{$T_J(z^J_-,z^J_+\mid z^J_{\circ})$.}
  \label{fig:TJ-tensor}
\end{subfigure}
\hfill
\begin{subfigure}[t]{0.30\textwidth}
  \centering
  \includegraphics[width=\linewidth]{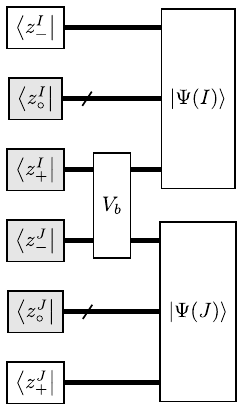}
  \caption{$T_{I\cup J}(z^{I\cup J}_-,z^{I\cup J}_+\mid z^{I\cup J}_{\circ})$.}
  \label{fig:TIJ-union-tensor}
\end{subfigure}
\caption{Coefficients
$T_I(z^I_-,z^I_+\mid z^I_{\circ})$, $T_J(z^J_-,z^J_+\mid z^J_{\circ})$,
and $T_{I\cup J}(z^{I\cup J}_-,z^{I\cup J}_+\mid z^{I\cup J}_{\circ})$.
Here we identify $(z^{I\cup J}_-,z^{I\cup J}_{\circ},z^{I\cup J}_+)=
(z^I_-,(z^I_{\circ},z^I_+,z^J_-,z^J_{\circ}),z^J_+)$.
}
\label{fig:boundary-amplitude-tensors}
\end{figure}

\begin{figure}[H]
\centering
\begin{subfigure}[b]{0.26\textwidth}
  \centering
  \includegraphics[width=\linewidth]{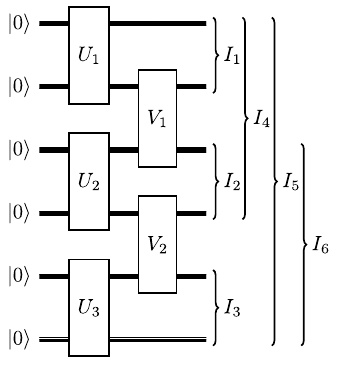}
  \caption{Valid intervals in the bilayer circuit.}
  \label{fig:bilayer-valid-intervals}
\end{subfigure}
\hfill
\begin{subfigure}[b]{0.68\textwidth}
  \centering
  \includegraphics[width=0.29\linewidth]{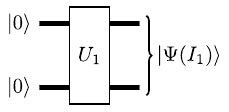}\hspace{0.7em}
  \includegraphics[width=0.29\linewidth]{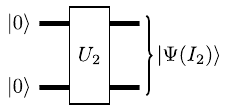}\hspace{0.7em}
  \includegraphics[width=0.29\linewidth]{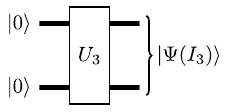}\\[0.45em]
  \includegraphics[width=0.29\linewidth]{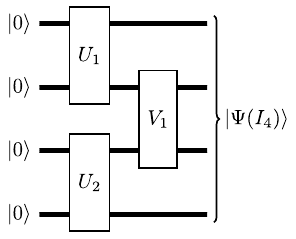}\hspace{0.7em}
  \includegraphics[width=0.29\linewidth]{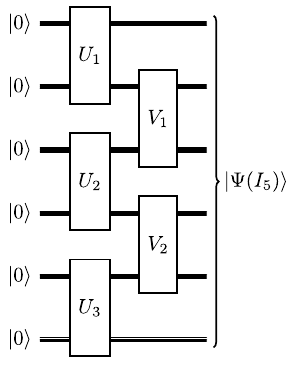}\hspace{0.7em}
  \includegraphics[width=0.29\linewidth]{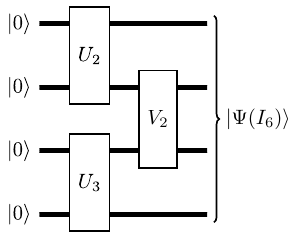}
  \caption{Associated restricted states $\ket{\Psi(I)}$.}
\end{subfigure}
\caption{Valid intervals for the bilayer circuit and the corresponding
restricted interval states.  For each depicted interval~$I_j$, the state
$\ket{\Psi(I_j)}$ is formed by keeping all $U$ gates supported on~$I_j$ and all
$V$ gates supported on~$I_j^\circ$.}
\label{fig:bilayer-interval-states}
\end{figure}

\subsection{Joining two neighboring intervals}
\label{sec:joining-intervals}

The routine \routine{JoinIntervals} combines samples from two neighboring
intervals by applying \routine{Resample} to the two interface qudits.
The boundary tables give the required probabilities.
Let
\begin{align}
  I=\{2a+1,\ldots,2b\},
  \qquad
  J=\{2b+1,\ldots,2c\}
\end{align}
be consecutive valid intervals. Their adjacent boundary sites $(2b,2b+1)$
form the interface pair, and the second-layer gate~$V_b$ acting on
these sites is the interface gate. Suppose we have independent samples
\begin{align}
  z^I=(z^I_-,z^I_{\circ},z^I_+)\sim P_{Z^I},
  \qquad
  z^J=(z^J_-,z^J_{\circ},z^J_+)\sim P_{Z^J}.
\end{align}
Here $z^I_-$ and $z^I_+$ are the values on the left and right boundary sites of
$I$, while $z^I_{\circ}$ is the interior string; the notation for~$J$ is
analogous.  The old interface values $z^I_+,z^J_-$ are discarded.  For a
candidate interface output
$(u,v)\in[K]^2$, the coefficients for the union, illustrated in
Fig.~\ref{fig:TIJ-union-identity}, are
\begin{align}
  T_{I\cup J}((z^I_-,z^J_+)\given z^I_{\circ},(u,v),z^J_{\circ})
  =
  \sum_{u^I_+,v^J_-\in[K]}
  (V_b)_{(u,v),(u^I_+,v^J_-)}
  T_I(z^I_-,u^I_+\given z^I_{\circ})
  T_J(v^J_-,z^J_+\given z^J_{\circ}).
  \label{eq:union-table}
\end{align}
Then sample $(u,v)\in[K]^2$ with probability
\begin{align}
  \Pr[(u,v)\mid z^I_-,z^I_{\circ},z^J_{\circ},z^J_+]
  =
  \frac{\abs{T_{I\cup J}((z^I_-,z^J_+)\given z^I_{\circ},(u,v),z^J_{\circ})}^2}
       {\sum_{(u',v')\in[K]^2}
       \abs{T_{I\cup J}((z^I_-,z^J_+)\given z^I_{\circ},(u',v'),z^J_{\circ})}^2}.
  \label{eq:join-conditional}
\end{align}
Equivalently, collect the retained values in
$\zeta=(z^I_-,z^I_{\circ},z^J_{\circ},z^J_+)$ and define the vector
$a_\zeta\in\mathbb C^{K^2}$ by
\begin{align*}
  a_\zeta(u,v)
  :=T_{I\cup J}((z^I_-,z^J_+)\given z^I_{\circ},(u,v),z^J_{\circ}).
\end{align*}
Then Eq.~\eqref{eq:join-conditional} is the instance of
Eq.~\eqref{eq:resample-born-distribution} given by
\begin{align}
  \Pr[(u,v)\mid\zeta]
  =\frac{|a_\zeta(u,v)|^2}{\|a_\zeta\|_2^2}
  =\mathcal B(a_\zeta)_{(u,v)}.
  \label{eq:join-conditional-born}
\end{align}
The squared norm $\|a_\zeta\|_2^2$ is the probability of the retained
values~$\zeta$. Appendix~\ref{app:finite-precision-tree} uses these same
vectors to compare exact and finite-precision joins.
The joined sample is
\begin{align}
  z^{I\cup J}=(z^I_-,z^I_{\circ},u,v,z^J_{\circ},z^J_+),
\end{align}
with the coefficients given by Eq.~\eqref{eq:union-table}.

\begin{algorithm}[!htbp]
\caption{\routine{JoinIntervals}: resample an interface and return the joined sample and table.}
\label{alg:join-routine}
\small
\begin{algorithmic}[1]
\Require Consecutive valid intervals $I,J$, independent samples
$z^I=(z^I_-,z^I_{\circ},z^I_+)\sim P_{Z^I}$ and
$z^J=(z^J_-,z^J_{\circ},z^J_+)\sim P_{Z^J}$, and coefficient tables
$T_I(\cdot\given z^I_{\circ})$ and $T_J(\cdot\given z^J_{\circ})$; the interface unitary $V_b$ on the last site of $I$ and first site of $J$.
\Ensure A sample $z^{I\cup J}\sim P_{Z^{I\cup J}}$ and the coefficient table
$T_{I\cup J}(\cdot\given z^{I\cup J}_{\circ})$.
\Function{JoinIntervals}{$I,J,z^I,z^J,T_I,T_J,V_b$}
\State Discard the old interface values $z^I_+$ and $z^J_-$.
\State Compute $T_{I\cup J}$ from $T_I,T_J$, and~$V_b$ using
Eq.~\eqref{eq:union-table} for all boundary and interface values.
\ForAll{$(u,v)\in[K]^2$ \textbf{in parallel}}
  \State $q_{u,v}\gets
  \abs{T_{I\cup J}((z^I_-,z^J_+)\given z^I_{\circ},(u,v),z^J_{\circ})}^2.$
\EndFor
\State Sample $(u,v)\in[K]^2$ with probability
$q_{u,v}/\sum_{(u',v')\in[K]^2}q_{u',v'}$.
\State Set $z^{I\cup J}\gets(z^I_-,z^I_{\circ},u,v,z^J_{\circ},z^J_+)$.
\State \Return $z^{I\cup J}$ and $T_{I\cup J}(\cdot\given z^I_{\circ},(u,v),z^J_{\circ})$.
\EndFunction
\end{algorithmic}
\end{algorithm}

\begin{figure}[!t]
\centering
\[
  T_{I\cup J}((z^I_-,z^J_+)\given z^I_{\circ},(u,v),z^J_{\circ})
  =\sum_{u^I_+,v^J_-\in[K]}
  \vcenter{\hbox{\includegraphics[width=0.38\textwidth]{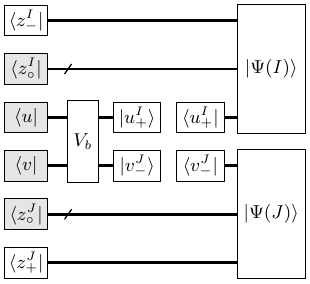}}}\,.
\]
\caption{Resolution of the identity on the two interface indices before
applying $V_b$, giving the coefficient in Eq.~\eqref{eq:union-table}.
The interface bras $\bra{u}$ and $\bra{v}$ specify the candidate outputs.}
\label{fig:TIJ-union-identity}
\end{figure}

\begin{lemma}[Correctness of \routine{JoinIntervals}]
\label{lem:join}
\textnormal{\routine{JoinIntervals}} (Algorithm~\ref{alg:join-routine}) outputs a sample from $P_{Z^{I\cup J}}$ and
the correct coefficients $T_{I\cup J}(\cdot\given z^{I\cup J}_{\circ})$.
Its parallel runtime and number of elementary real operations are
both $\poly(K)$.
\end{lemma}

\begin{proof}
By the definition of a restricted interval state, joining~$I$ and~$J$ adds
exactly the interface gate~$V_b$.  Hence
\begin{align}
  |\Psi(I\cup J)\rangle
  &=
  V_b\bigl(|\Psi(I)\rangle\otimes|\Psi(J)\rangle\bigr)\ .
  \label{eq:union-state-identity}
\end{align}
Inserting the identity on the interface pair before~$V_b$ gives
Eq.~\eqref{eq:union-table}.

Fix the retained values $(z^I_-,z^I_{\circ},z^J_{\circ},z^J_+)$ and define
$w\in\bbC^{K^2}$ by
\begin{align}
  w_{(u^I_+,v^J_-)}
  =T_I(z^I_-,u^I_+\given z^I_{\circ})
  T_J(v^J_-,z^J_+\given z^J_{\circ}) .
\end{align}
The product sample assigns these retained values total weight~$\|w\|_2^2$.
After the interface gate their weight is~$\|V_bw\|_2^2=\|w\|_2^2$, so their
marginal distribution is unchanged.  Eq.~\eqref{eq:join-conditional} resamples the
two interface values from the transformed amplitudes~$V_bw$.  The resampling
step in \routine{Resample} (Section~\ref{sec:samplingwithoutmarginals})
therefore gives the target distribution on~$I\cup J$.

There are $K^4$ choices of outer boundary and interface values.
Each entry in Eq.~\eqref{eq:union-table} is a sum of $K^2$
products, giving $O(K^6)$ elementary real operations. Selecting the
probabilities and the returned table entries takes $\poly(K)$ operations.
By Lemma~\ref{lem:discrete-sampling-cost}, the sampling step takes
$O(K^2)$ operations and $O(\log(K+1))$ parallel runtime. These estimates
give the claimed bounds.

Store one outcome register per qudit. Each join updates only the two
interface registers and the boundary table, so interior strings need
not be copied.
\end{proof}

\subsection{Parallel joining}
\label{sec:coarsening-subsec}
The routine \routine{TreeSample} (Algorithm~\ref{alg:joining-tree})
starts with two-qudit intervals and applies \routine{JoinIntervals}
along a balanced binary tree, processing each level in parallel.
We call this balanced binary tree the joining tree.
Fig.~\ref{fig:joining-tree} illustrates it for $M=4$.
The intervals $\{2i-1,2i\}$ are the leaves of the tree, initialized with
$U_i|00\rangle$; each internal node represents one join of neighboring
child intervals.
Each call to \routine{JoinIntervals} returns both a sample and its boundary
coefficient table; the table supplies the amplitudes needed at the next
level. Disjoint pairs at a level use independent randomness.

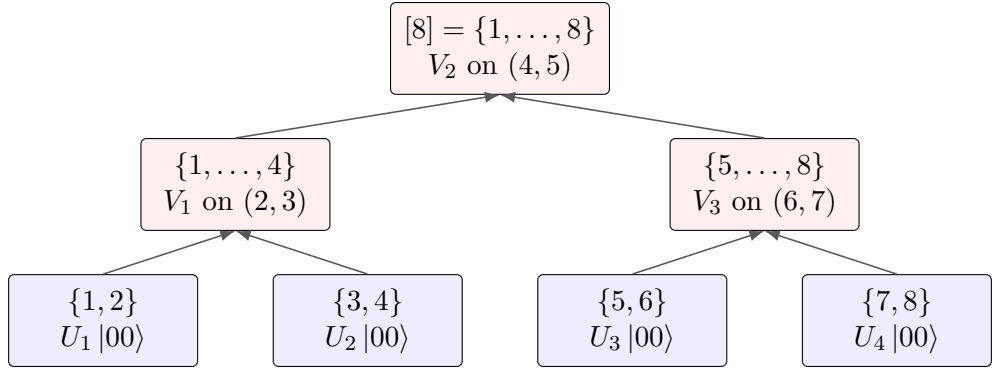
\begin{figure}[H]
\centering
\begin{tikzpicture}[
  x=1cm,y=1cm,
  jtNode/.style={draw,rounded corners=2pt,align=center,font=\small,
    inner sep=5pt,minimum width=2.5cm,minimum height=0.95cm},
  jtLeaf/.style={jtNode,fill=blue!7},
  jtJoin/.style={jtNode,fill=red!6},
  jtEdge/.style={-Latex,line width=0.6pt,draw=black!65}
]
  \node[jtLeaf] (leaf1) at (0,0) {$\{1,2\}$\\$U_1\ket{00}$};
  \node[jtLeaf] (leaf2) at (3.5,0) {$\{3,4\}$\\$U_2\ket{00}$};
  \node[jtLeaf] (leaf3) at (7,0) {$\{5,6\}$\\$U_3\ket{00}$};
  \node[jtLeaf] (leaf4) at (10.5,0) {$\{7,8\}$\\$U_4\ket{00}$};
  \node[jtJoin] (join12) at (1.75,1.8)
    {$\{1,\ldots,4\}$\\$V_1$ on $(2,3)$};
  \node[jtJoin] (join34) at (8.75,1.8)
    {$\{5,\ldots,8\}$\\$V_3$ on $(6,7)$};
  \node[jtJoin] (root) at (5.25,3.6)
    {$[8]=\{1,\ldots,8\}$\\$V_2$ on $(4,5)$};
  \draw[jtEdge] (leaf1.north) -- (join12.south);
  \draw[jtEdge] (leaf2.north) -- (join12.south);
  \draw[jtEdge] (leaf3.north) -- (join34.south);
  \draw[jtEdge] (leaf4.north) -- (join34.south);
  \draw[jtEdge] (join12.north) -- (root.south);
  \draw[jtEdge] (join34.north) -- (root.south);
\end{tikzpicture}
\caption{Joining tree for $M=4$. Each node is labelled by its qudit
interval. Leaves are initialized with $U_i\ket{00}$; internal nodes show
the interface gate and the pair of qudits on which it acts. Arrows carry
samples and boundary tables. The $V_1$ and $V_3$ joins run in parallel,
followed by the $V_2$ join.}
\label{fig:joining-tree}
\end{figure}

\begin{algorithm}[!htbp]
\caption{\routine{TreeSample}: sample a bilayer chain by joining neighboring intervals.}
\label{alg:joining-tree}
\label{alg:joining-tree-sampler}
\small
\begin{algorithmic}[1]
\Require A bilayer circuit $\sU=(\{U_r\}_{r=1}^{M},\{V_b\}_{b=1}^{M-1})$ on $2M$ qudits of dimension $K$, as in Eq.~\eqref{eq:basic-bilayer-circuit}.
\Ensure A string $z\in[K]^{2M}$ with distribution $P_{Z([2M])}$.
\Function{TreeSample}{$\sU$}
\ForAll{$r=1,\ldots,M$ \textbf{in parallel}}
  \State Set $I_r^{(0)}\gets\{2r-1,2r\}$.
  \State Draw $z_r^{(0)}\in[K]^2$ with probability $\abs{\langle z_r^{(0)}|U_r|00\rangle}^2$.
  \State Store $T_{I_r^{(0)}}(a,b\given \emptyset)\gets\langle a,b|U_r|00\rangle$ for all $(a,b)\in[K]^2$.
\EndFor
\For{$h=0,\ldots,\lceil\log_2 M\rceil-1$}
  \ForAll{disjoint neighboring pairs $(I,J)$ at level $h$ \textbf{in parallel}}
    \State Let $V_b$ be the gate across the interface between $I$ and $J$.
    \State $(z^{I\cup J},T_{I\cup J})\gets\Call{JoinIntervals}{I,J,z^I,z^J,T_I,T_J,V_b}$.
    \State Store this sample and table as the active record for $I\cup J$ at level $h+1$.
  \EndFor
  \State Carry any unpaired endpoint interval unchanged to the next level.
\EndFor
\State \Return the final string $z$ on $[2M]$.
\EndFunction
\end{algorithmic}
\end{algorithm}

\begin{samepage}
\begin{proposition}[Correctness and complexity of \routine{TreeSample}]
\label{prop:bilayer-sampler}
For a bilayer chain of $2M$ qudits of local dimension~$K$,
\textnormal{\routine{TreeSample}} samples its output distribution exactly
in parallel runtime
\begin{align}
  O\!\left(\poly(K)\log(M+1)\right)
  \label{eq:bilayer-depth}
\end{align}
using $O(M\,\poly(K))$ elementary real operations.
\end{proposition}
\end{samepage}

\begin{proof}
At every level, the active intervals partition~$[2M]$ and carry mutually
independent exact samples and their boundary tables.  This holds initially
because $|\Psi(I_r^{(0)})\rangle=U_r|00\rangle$.  Lemma~\ref{lem:join}
preserves the invariant when disjoint neighboring pairs are joined in
parallel; an unpaired endpoint interval is simply carried forward.  The number
of active intervals changes from~$m$ to~$\lceil m/2\rceil$, so after
$\lceil\log_2M\rceil$ levels the remaining sample has distribution~$P_{Z([2M])}$.
Initialization takes $\poly(K)$ parallel runtime and $\poly(K)$
operations per interval, including sampling by
Lemma~\ref{lem:discrete-sampling-cost}. Each subsequent level takes
$\poly(K)$ parallel runtime by Lemma~\ref{lem:join}. This proves
Eq.~\eqref{eq:bilayer-depth}, also for $M=1$. There are $M$ initial
intervals and $M-1$ joins, each requiring $\poly(K)$ operations.
Retaining outcome registers and copying the final $2M$ outcomes once
gives $O(M\,\poly(K))$ operations in total.
\end{proof}

\begin{corollary}[Noiseless shallow circuits]\label{cor:noiseless}
The output distribution of an arbitrary depth-$d$ brickwork circuit~$\sU$
on $n$ qubits, without noise, can be sampled exactly in parallel runtime $O(\poly(2^d)\log n)$ using $O(n\,\poly(2^d))$ elementary real
operations.
\end{corollary}
\begin{proof}
The reduction in Section~\ref{sec:bilayer-reduction}, with $L=n$, gives
$K=2^d$ and $M=O(n/d+1)$. Constructing the coarse gates takes
$\poly(2^d)$ parallel runtime and $O(M\,\poly(2^d))$ operations.
Apply \routine{TreeSample} and decode the qubit outcomes.
Proposition~\ref{prop:bilayer-sampler} gives parallel runtime
$O(\poly(K)\log(M+1))=O(\poly(2^d)\log n)$ and
$O(M\,\poly(K))=O(n\,\poly(2^d))$ operations, including decoding.
\end{proof}

\begin{corollary}[Complexity of \routine{LocalSample}]
\label{cor:local-marginal-sampling-circuit}
For a fixed noise realization~$f$ and a contiguous interval~$A\subset[n]$ with
backward light-cone width at most~$L$, \textnormal{\routine{LocalSample}} (Algorithm~\ref{alg:local-marginal-sampler})
has parallel runtime
\begin{align}
  O\!\left(\poly(2^d)\log(L+d)\right)
\end{align}
and uses $O\!\left(\poly(2^d)(L+d)\right)$ elementary real operations.
\end{corollary}

\begin{proof}
Construct the gates of $\tsU(f)$ on the backward light cone of~$A$.
This takes $O(Ld)$ elementary real operations and $O(d)$ parallel
runtime; gates outside the light cone are not processed.
Section~\ref{sec:bilayer-reduction} gives a bilayer circuit with
$2M=O(L/d+1)$ qudits of dimension $K=2^d$. Constructing its coarse
gate matrices takes $O(M\,\poly(2^d))$ operations and $\poly(2^d)$
parallel runtime. Apply \routine{TreeSample}, decode the qubit outcomes,
and retain those on~$A$. Proposition~\ref{prop:bilayer-sampler} gives
the claimed bounds.
\end{proof}


\section{The sampling algorithm}
\label{sec:full-routine}

In this section we assemble the subprograms of
Sections~\ref{sec:circuit-cutting}--\ref{sec:circuitsampling} into the
algorithm \Sample{} (Algorithm~\ref{alg:full-noisy-interval-sampler}) of Theorem~\ref{thm:main}. The proof of
Theorem~\ref{thm:main} is completed in Section~\ref{sec:simulation-proof}.

The inputs of \Sample{} are the brickwork circuit~$\sU$, the noise
probability~$p$, and an integer width bound~$W_*\geq1$. The algorithm
draws a noise fixing~$f$ with \routine{DrawNoise}
(Algorithm~\ref{alg:draw-noise}) and computes the cut indicators
$\{E_{1/2+k}\}_{k=1}^{n-1}$ of Section~\ref{sec:convexdecompo}. It then
sets $\ell=2W_*+4$ as in Eq.~\eqref{eq:ell-def} and calls
\routine{LocalSample} (Algorithm~\ref{alg:local-marginal-sampler}) once
for every interval of the two partitions $\{A_r\}_{r=0}^{r_{\max}}$
and $\{B_s\}_{s=0}^{s_{\max}}$ of Eqs.~\eqref{eq:A-intervals}
and~\eqref{eq:B-intervals}. All these calls share the noise fixing~$f$,
but each call uses its own independent randomness. Finally,
\routine{BoundaryPass} (Algorithm~\ref{alg:boundary-selection}) selects
one interval sample for every cut-bounded component.

\begin{algorithm}[!htbp]
\caption{\Sample{}: sample local marginals and assemble a noisy-circuit output.}
\label{alg:full-noisy-interval-sampler}
\small
\begin{algorithmic}[1]
\Require A depth-$d$ brickwork circuit $\sU$ on $n$ qubits, a noise probability $p\in(0,1)$, and an integer width bound $W_*\geq1$.
\Ensure A string $Z\in\bits^n$ whose conditional distribution given the noise fixing $F=f$ is $Q^f_Z$ whenever $W_{\max}<W_*$ (Proposition~\ref{prop:assembled-routine}). If $W_*$ is at least the value in Eq.~\eqref{eq:Wstar-choice}, the distribution of $Z$ is within total variation distance~$\delta$ of~$P^{(p)}$ (Corollary~\ref{cor:l1-approximation}).
\Function{\Sample}{$\sU,p,W_*$}
\State $(f,\tsU)\gets\Call{DrawNoise}{\sU,p}$.
\State Compute the cut indicators $\{E_{1/2+k}\}_{k=1}^{n-1}$ as products of the rectangle variables of~$f$, as in Section~\ref{sec:convexdecompo}.
\State Set $\ell\gets 2W_*+4$, $r_{\max}\gets \lceil n/\ell\rceil-1$, and $s_{\max}\gets \lceil(n-\ell/2)/\ell\rceil$.
\ForAll{$r=0,\ldots,r_{\max}$ \textbf{in parallel}}
  \State Set $A_r\gets(r\ell,(r+1)\ell]\cap[n]$.
  \State $X_{A_r}\gets\Call{LocalSample}{\sU,f,A_r}$.
\EndFor
\ForAll{$s=0,\ldots,s_{\max}$ \textbf{in parallel}}
  \If{$s=0$}
    \State Set $B_s\gets(0,\ell/2]\cap[n]$.
  \Else
    \State Set $B_s\gets((s-\frac{1}{2})\ell,(s+\frac{1}{2})\ell]\cap[n]$.
  \EndIf
  \State $Y_{B_s}\gets\Call{LocalSample}{\sU,f,B_s}$.
\EndFor
\State $Z\gets\Call{BoundaryPass}{\{E_{1/2+k}\}_{k=1}^{n-1},\ell,X,Y}$.
\State \Return $Z$.
\EndFunction
\end{algorithmic}
\end{algorithm}

Throughout this section, $W_{\max}=\max_{0\leq r\leq\numCuts}W_r$ denotes
the maximum width of a cut-bounded component of the noise fixing, and
\begin{align}
  \Good=\{W_{\max}<W_*\}
  \qquad\textrm{and}\qquad
  \Bad=\{W_{\max}\geq W_*\}
  \label{eq:good-bad-events}
\end{align}
are the two complementary events considered in the proof of
Lemma~\ref{lem:maxcircuitwidth}. We use $Q^f_Z$ and $Q^f_{Z_C}$ for the
output distribution of~$\tsU(f)$ and its marginal on~$C$, respectively,
as defined in Eqs.~\eqref{eq:fixed-noise-output}
and~\eqref{eq:local-marginal-target}.

\begin{proposition}[Correctness of \Sample{}]
\label{prop:assembled-routine}
Let $W_*\geq1$ be an integer, and let $f$ be a noise fixing with
$W_{\max}<W_*$. Conditioned on $F=f$, the output~$Z$ of
\textnormal{\Sample{}}$(\sU,p,W_*)$ has distribution~$Q^f_Z$ of
Eq.~\eqref{eq:fixed-noise-output}.
\end{proposition}
\begin{proof}
Fix~$f$. By the light-cone reduction of Section~\ref{sec:circuitsampling}
and Proposition~\ref{prop:bilayer-sampler}, the sample returned by
\routine{LocalSample}$(\sU,f,C)$ for an interval~$C$ has
distribution~$Q^f_{Z_C}$. The calls to
\routine{LocalSample} use independent randomness, so the samples
$\{X_{A_r}\}_r$ and $\{Y_{B_s}\}_s$ are mutually independent. Note that
the backward light cones of different intervals overlap; the independence
of the samples is a consequence of the independent randomness of the
calls, not of any disjointness. Hence $X$ and~$Y$ are distributed as in
Eqs.~\eqref{eq:factorizationQdistributionA}
and~\eqref{eq:factorizationQdistributionB} with $Q_Z=Q^f_Z$. Since
$W_{\max}<W_*$, Theorem~\ref{thm:local-selection} shows that the output
of \routine{BoundaryPass} has distribution~$Q^f_Z$.
\end{proof}

The only source of error is the event~$\Bad$, on which the output of
\Sample{} may have an arbitrary distribution.

\begin{corollary}[Error of \Sample{}]
\label{cor:l1-approximation}
Let $W_*\geq1$ be an integer, and let $\widehat P$ be the output
distribution of \textnormal{\Sample{}}$(\sU,p,W_*)$. Then
\begin{align}
  d_{\mathrm{TV}}(\widehat P,P^{(p)})\leq\Pr[\Bad] .
  \label{eq:sample-tv-error}
\end{align}
In particular, for $\delta\in(0,1)$, if
\begin{align}
  W_*\geq\left\lceil\frac{\log(n/\delta)}{\log(1/(1-p^d))}\right\rceil+1,
  \label{eq:sample-width-condition}
\end{align}
then $d_{\mathrm{TV}}(\widehat P,P^{(p)})\leq\delta$.
\end{corollary}
\begin{proof}
The noise fixing~$F$ drawn by \routine{DrawNoise} has the distribution of
Eq.~\eqref{eq:noisefixing}. Let $\widehat P(\cdot\mid F=f)$ be the
conditional distribution of the output of \Sample{} given $F=f$. Then
$\widehat P=\mathbb{E}_F[\widehat P(\cdot\mid F)]$, and
$P^{(p)}=\mathbb{E}_F[Q^F_Z]$ by the representation of the noisy circuit as
a mixture of the unitary circuits~$\tsU(f)$ in the proof of
Lemma~\ref{lem:conditional-factorization}. Since the total variation
distance is jointly convex,
\begin{align}
  d_{\mathrm{TV}}(\widehat P,P^{(p)})
  \leq\mathbb{E}_F\bigl[d_{\mathrm{TV}}(\widehat P(\cdot\mid F),Q^F_Z)\bigr]
  \leq\Pr[\Bad] ,
\end{align}
where we used Proposition~\ref{prop:assembled-routine} on the event~$\Good$
and the bound $d_{\mathrm{TV}}\leq1$ on its complement. If
$W_*\geq\left\lceil\log(n/\delta)/\log(1/(1-p^d))\right\rceil+1$,
then $\Pr[\Bad]\leq\delta$ by Lemma~\ref{lem:maxcircuitwidth} with $q=p^d$.
\end{proof}

Note that the error in Eq.~\eqref{eq:sample-tv-error} can be removed
altogether: the event~$\Bad$ is determined by the cut indicators, and on
this event one may return the output of
$\routine{LocalSample}(\sU,f,[n])$ instead. The resulting algorithm
samples exactly from~$P^{(p)}$, but its worst-case parallel runtime is
that of the noiseless sampler of Corollary~\ref{cor:noiseless}.

\begin{corollary}[Complexity of \Sample{}]
\label{cor:full-interval-sampler-circuit}
For every integer $W_*\geq1$, \textnormal{\Sample{}}$(\sU,p,W_*)$ has parallel
runtime
\begin{align}
  O\bigl(\poly(2^d)\log(W_*+d)\bigr)
  \label{eq:sample-runtime}
\end{align}
and uses $O(n\,\poly(2^d))$ elementary real operations.
\end{corollary}
\begin{proof}
We consider the four stages of \Sample{} (Algorithm~\ref{alg:full-noisy-interval-sampler})
in turn. \routine{DrawNoise} draws $O(nd)$ independent random variables
with constant support and forms the substituted gates~$\tilde U^{(t)}_j$,
which are $4\times4$ matrices. By Lemma~\ref{lem:discrete-sampling-cost},
this takes constant parallel runtime and $O(nd)$ operations. Each cut
indicator is the product of the $d/2$ rectangle variables at one bond,
and all cut indicators are computed along balanced binary trees in
parallel runtime $O(\log d)$ with $O(nd)$ operations.

Every interval~$C$ of the two partitions has length $|C|\leq\ell$, so
its backward light cone has width at most $\ell+2d$. By
Corollary~\ref{cor:local-marginal-sampling-circuit}, the calls to
\routine{LocalSample}, which are executed in parallel, have parallel
runtime $O(\poly(2^d)\log(\ell+3d))=O(\poly(2^d)\log(W_*+d))$, and the
call for~$C$ uses $O(\poly(2^d)(|C|+d))$ operations. The two partitions
consist of $O(1+n/\ell)$ intervals whose lengths sum to~$2n$. Hence the
calls use
\begin{align}
  O\Bigl(\poly(2^d)\Bigl(n+d\Bigl(1+\frac{n}{\ell}\Bigr)\Bigr)\Bigr)
  =O\bigl(nd\,\poly(2^d)\bigr)
  =O\bigl(n\,\poly(2^d)\bigr)
\end{align}
operations in total, where we used $d\leq n$ and $d\leq2^d$. Finally,
\routine{BoundaryPass} has parallel runtime $O(\log(W_*+1))$ and uses
$O(n)$ operations by Corollary~\ref{cor:local-selection-circuit}. Summing
the four contributions and using $\log d\leq d\leq\poly(2^d)$ gives the
claim.
\end{proof}

\subsection{Proof of Theorem~\ref{thm:main}}
\label{sec:simulation-proof}

\begin{proof}[Proof of Theorem~\ref{thm:main}]
We run \Sample{} with the width bound
\begin{align}
  W_*=\bigl\lceil p^{-d}\log(n/\delta)\bigr\rceil+1 .
  \label{eq:main-effective-width}
\end{align}
Since $\log(1/(1-q))\geq q=p^d$, this is at least the value in
Eq.~\eqref{eq:Wstar-choice}, and Corollary~\ref{cor:l1-approximation}
gives Eq.~\eqref{eq:main-total-variation}. By
Corollary~\ref{cor:full-interval-sampler-circuit}, the algorithm uses
$O(n\,\poly(2^d))=n\,2^{O(d)}$ elementary operations, and its parallel
runtime is $O(\poly(2^d)\log(W_*+d))$. Here
$W_*+d\leq(2+d)\bigl(2+p^{-d}\log(n/\delta)\bigr)$, since $a+b\leq ab$ for
$a,b\geq2$. It follows that
\begin{align}
  \poly(2^d)\log(W_*+d)
  &\leq\poly(2^d)\Bigl(\log(2+d)+\log\bigl(2+p^{-d}\log(n/\delta)\bigr)\Bigr)
  \notag\\
  &=2^{O(d)}\log\bigl(2+p^{-d}\log(n/\delta)\bigr),
\end{align}
where we used $\log(2+d)=O(d)$ and
$\log(2+p^{-d}\log(n/\delta))\geq\log2$. Finally, $2+x\leq4\max\{1,x\}$
for $x\geq0$, and $\log\log(n/\delta)\geq\log\log2>-1/2$ since
$n/\delta>2$. With $x=p^{-d}\log(n/\delta)$ this gives
$\log\bigl(2+p^{-d}\log(n/\delta)\bigr)
\leq3\bigl(1+d\log\frac1p+\log\log\frac n\delta\bigr)$, which is the
runtime bound in Eq.~\eqref{eq:main-complexity}.
\end{proof}

\paragraph{Classical inputs.}
Let $x\in\{0,1\}^k$ with $k\leq n$, and consider the noisy circuit applied
to the input state $\ket{x}\otimes\ket{0^{n-k}}$, with output
distribution $P^{(p)}(\cdot\mid x)$ as in
Section~\ref{sec:simulation-overview}. Set $x_j=0$ for $k<j\leq n$ and
$X^x=\bigotimes_{j=1}^nX^{x_j}$, so that
$\ket{x}\otimes\ket{0^{n-k}}=X^x\ket{0^n}$. The depolarizing channel
commutes with Pauli conjugation, $\cN_p(X\rho X)=X\cN_p(\rho)X$, and
therefore
\begin{align}
  P^{(p)}(\cdot\mid x)=P^{(p)}_{\sU_x},
\end{align}
where $P^{(p)}_{\sU_x}$ denotes the distribution of
Eq.~\eqref{eq:targetdistribution} for the depth-$d$ brickwork
circuit~$\sU_x$ obtained from~$\sU$ by replacing the first layer~$L_1$
by~$L_1X^x$, i.e., by replacing each gate $U^{(1)}_j$ of the first layer
by $U^{(1)}_j(X^{x_{2j-1}}\otimes X^{x_{2j}})$.
The matrix elements of these gates are those of~$U^{(1)}_j$ with columns
permuted according to two bits of~$x$, and they are computed from the
matrix elements of~$\sU$ and~$x$ by $O(n)$ selections in constant
parallel runtime. Applying \Sample{} to~$\sU_x$ therefore gives the
extension of Theorem~\ref{thm:main} to classical inputs stated in
Section~\ref{sec:simulation-overview}, with the same runtime and size
bounds.

\section{Circuit implementations}
\label{sec:circuit-realizations}

In this section we show how to implement \Sample{}
(Algorithm~\ref{alg:full-noisy-interval-sampler}) by classical Boolean circuits. We first
observe that the algorithm itself is an arithmetic circuit
(Section~\ref{sec:arithmetic-realization}). We then give exact Boolean
circuits for the combining stage
(Section~\ref{sec:combining-boolean-realization}) and describe the
finite-precision implementation of the remaining stages
(Section~\ref{sec:boolean-realization}), whose analysis is given in
Appendix~\ref{app:boolean-implementation}. Throughout, the size of a
circuit is its number of gates and its depth is the maximum number of
gates on a directed path from an input to an output; fan-out is
unrestricted.

\subsection{Arithmetic circuits}
\label{sec:arithmetic-realization}

The sequence of elementary operations executed by \Sample{}, and the
registers on which they act, depend only on $n$, $d$, and~$W_*$. Indeed,
the interval partitions, the backward light cones, the bilayer
reductions, and the joining trees are determined by the brickwork circuit geometry;
here the light cone of an interval is taken in~$\sU$, which contains its
light cone in~$\tsU(f)$ and leaves the marginal on the interval
unchanged. Moreover,  all data-dependent choices,
namely the substituted gates~$\tilde U^{(t)}_j$, the branches taken in the
sampling trees of Lemma~\ref{lem:discrete-sampling-cost}, and the table
entries retained by \routine{JoinIntervals}, are made by selection
operations. Replacing every elementary operation by a gate (with real-valued inputs and outputs on wires) therefore
gives a randomized arithmetic circuit. Its gates perform real
arithmetic, comparisons, selections, and Boolean operations of bounded
arity, and a random gate with input $r\in[0,1]$ returns an independent random bit with distribution~$\mathsf{Ber}(r)$. The depth of the circuit is the parallel runtime
of the algorithm, and its size is the number of elementary operations.
The inputs of the circuit are the matrix elements of the gates of~$\sU$
and the noise probability~$p$; for a fixed noisy quantum circuit they are
constants. This gives the following reformulation.
\begin{corollary}[Arithmetic circuits]
\label{cor:arithmetic-circuit-realization}
Under the assumptions of Theorem~\ref{thm:main}, there is a randomized
arithmetic circuit of bounded fan-in with depth
$2^{O(d)}\log\bigl(2+p^{-d}\log(n/\delta)\bigr)$ and size $n\,2^{O(d)}$
whose output distribution is within total variation distance~$\delta$
of~$P^{(p)}$. For constant~$d$ and~$p$ and constant or inverse-polynomial~$\delta$,
its depth is $O(\log\log n)$ and its size is $O(n)$. Without noise, there
is a randomized arithmetic circuit of depth $O(\poly(2^d)\log n)$ and size
$O(n\,\poly(2^d))$ which samples exactly from the output distribution
of~$\sU$.
\end{corollary}
\begin{proof}
Apply the construction above to \Sample{} with the width bound of
Eq.~\eqref{eq:main-effective-width}, and to \routine{TreeSample}
(Algorithm~\ref{alg:joining-tree}), together with the bilayer reduction
and decoding used in Corollary~\ref{cor:noiseless}. The circuits perform the same operations on
the same random bits as the algorithms and hence have the same output
distributions. The bounds are those of Theorem~\ref{thm:main} and
Corollary~\ref{cor:noiseless}.
\end{proof}

\subsection{Boolean circuits for the combining stage}
\label{sec:combining-boolean-realization}

The subprograms \routine{BoundaryPass} (Algorithm~\ref{alg:boundary-selection})
and \routine{IntervalContainment}$_N$
(Algorithm~\ref{alg:interval-containment}) operate on bits only. Their
arithmetic circuits are therefore Boolean circuits, and with unbounded
fan-in their depth becomes constant.

\begin{corollary}[Boolean circuits for containment and selection]
\label{cor:combining-boolean-realization}
\textnormal{\routine{IntervalContainment}}$_N$ on an interval of $N$ sites is
computed by a Boolean circuit with fan-in-two AND and OR gates and NOT
gates of depth $O(\log N)$ and size $O(N)$, and by a circuit with
unbounded fan-in AND and OR gates of depth two and size $O(N)$.
\textnormal{\routine{BoundaryPass}} with interval length $\ell=2W_*+4$ is
computed by a Boolean circuit with fan-in-two AND and OR gates and NOT
gates of depth $O(\log\ell)=O(\log(W_*+1))$ and size $O(n)$, and by a
circuit with unbounded fan-in AND and OR gates and NOT gates of constant
depth and size $O(n)$. All these circuits compute the respective maps
exactly.
\end{corollary}
\begin{proof}
The inputs, the messages $h_v$, $L_v$, and~$R_v$, the flags~$\chi_j$, and
the sample bits are bits, and every update in the two algorithms is a
Boolean operation of bounded arity; the selection of $X_j$ or~$Y_j$
according to~$\chi_j$ is $Z_j=(\chi_j\land X_j)\lor(\neg\chi_j\land Y_j)$.
The fan-in-two bounds are therefore those of
Theorem~\ref{cor:interval-containment-circuit}
and Corollary~\ref{cor:local-selection-circuit}. With unbounded fan-in, evaluate the
two disjunctions in Eq.~\eqref{eq:interval-containment-output} for every
site~$j$ by one OR gate each, followed by one AND gate. This gives depth
two and at most $3N$ gates. \routine{BoundaryPass} applies this circuit
to each of the $\lceil n/\ell\rceil$ intervals~$A_r$, with the boundary
indicators replaced by the constants specified in
\routine{BoundaryPass} (Algorithm~\ref{alg:boundary-selection}), and then performs the selection
at every site. Since $\sum_r|A_r|=n$, this takes constant depth and
$O(n)$ gates.
\end{proof}

\subsection{Finite precision and fair coins}
\label{sec:boolean-realization}

The remaining stages of \Sample{}, namely the generation of the noise
fixing, the construction of the coarse gates, and \routine{TreeSample},
operate on real numbers and biased random bits. Appendix~\ref{app:boolean-implementation} implements them with
independent fair coins and integer words of $\beta$~bits, using a single precision parameter~$s$: the noise
probability~$p$ and the real and imaginary parts of the gate entries are given as multiples of~$2^{-s}$, the noise fixing is drawn from fair coins by simulating the depolarizing channel at each noise location directly
(Appendix~\ref{app:noise-bits}), and each draw of \routine{TreeSample} is
realized with $s$~fair coins by integer comparisons against
cumulative weights (Appendix~\ref{app:finite-precision-tree}). All other computations, namely the products forming the coarse
gates and the contractions of amplitude tables in
\routine{JoinIntervals}, are carried out exactly in integer arithmetic.
Appendix~\ref{app:full-implementation} shows that
$s=O(d+\log(n/\delta))$ suffices for total variation error~$\delta$ and
that a word length $\beta=O(ds(W_*+d))$ accommodates every
intermediate integer, so that for constant~$d$ the
gates and the noise probability need to be specified with
$O(\log(n/\delta))$ bits only.

Appendix~\ref{app:boolean-arithmetic} converts the resulting algorithm into Boolean circuits. With fan-in two, the depth is $O(\log(W_*+d)\,(d+\log \beta))$, dominated by the $O(d+\log \beta)$ depth of the arithmetic at each of the $O(\log(W_*+d))$ levels of coarse-gate multiplications and joins, and the size is $n\,2^{O(d)}\beta^2$ (Theorem~\ref{thm:full-bounded-boolean}). For constant~$d$ and~$p$ and constant or inverse-polynomial~$\delta$, this gives the circuit~$\cC_2$ of Corollary~\ref{cor:overview-bounded-boolean}, and collapsing groups of $\lfloor\log_2\log_2n\rfloor$ consecutive levels into two levels of unbounded fan-in gates gives the circuit~$\cC_\infty$ of Corollary~\ref{cor:overview-unbounded-boolean}. Both circuits can be constructed in polynomial time, and both extend to classical inputs by permuting the entries of the first-layer gates according to the input bits as in Section~\ref{sec:simulation-proof}. Appendix~\ref{app:lookup-tables} replaces each call to \routine{LocalSample} by a single table lookup, which gives the constant-depth circuit~$\cC_{\mathrm{lt}}$ of Corollary~\ref{cor:compiled-constant-depth}.

\section{Conclusion and outlook}
\label{sec:conclusion}

We have shown that the output distribution of a depth-$d$ brickwork
circuit on $n$ qubits subject to depolarizing noise of constant strength
can be sampled to within total variation error~$\delta$ in parallel
runtime $2^{O(d)}\log\log(n/\delta)$ with $n\,2^{O(d)}$ elementary
operations, and that for constant depth and constant or
inverse-polynomial error the algorithm can be implemented by Boolean
circuits of depth $O(\log\log n)$ with unbounded fan-in, or of depth
$O((\log\log n)^2)$ with fan-in two. Tabulating the local marginals
gives randomized $\mathsf{AC}^0$-circuits, of constant depth and larger
polynomial size. Noise is what makes these depth bounds possible.
Without noise, our method gives parallel runtime $O(\log n)$, and by the
separations of~\cite{BGKT20,CCR26} for
noiseless one-dimensional circuits
(Table~\ref{tab:quantum-advantage-comparison}), no polynomial-size
Boolean circuit of depth $o(\log n/\log\log n)$ reproduces the
input/output behavior of every noiseless constant-depth circuit on a
line. We also
note that the dependence on~$n$ is small at any realistic system size,
since $\log_2\log_2n\leq6$ for $n\leq2^{64}$. At such sizes the depth of
the circuits of Corollaries~\ref{cor:overview-bounded-boolean}
and~\ref{cor:overview-unbounded-boolean} is governed by the
factor~$2^{O(d)}$, whose exponent we have not optimized.

Our argument extends to any single-qubit Pauli channel in which all four
Paulis occur with nonzero probability. Such a channel is a convex
combination of the completely depolarizing channel and another Pauli
channel, so that Lemma~\ref{lem:circuitcutting} holds with modified
weights and the noise fixing again yields a unitary circuit. More
generally, every channel whose Choi matrix has full rank contains the
completely depolarizing channel as a convex component, and the cutting
argument of Section~\ref{sec:circuit-cutting} applies. The residual
channel is then no longer unitary, however, so the cut-bounded components
are noisy circuits themselves, and the exact local sampler of
Section~\ref{sec:circuitsampling} would have to be replaced.

Our model does not include mid-circuit measurements, classical
feedforward, qubit resets, or fresh ancillas. Already local classical
feedforward introduces classical communication between qubits which the
noise does not interrupt, so that the circuit no longer decomposes into
components of bounded width, and adaptive circuits require a different
analysis.

Finally, we do not know whether the doubly logarithmic parallel runtime
of Theorem~\ref{thm:main} is optimal for algorithms with a linear number
of operations, whereas the runtime $O(\log n)$ of the noiseless sampler
is optimal~\cite{BGKT20,CCR26}.
Output bits whose backward light cones
are disjoint are independent, so the output distribution has
correlations of range $O(d)$ only. Nevertheless, producing a globally
consistent sample appears to require coordination across the cut-bounded
components, whose widths are logarithmic in~$n$, and this is the source
of the dependence on~$n$ in the runtime of Theorem~\ref{thm:main}. The
circuits of Corollary~\ref{cor:compiled-constant-depth} avoid this
dependence in the depth with unbounded fan-in, at the price of a polynomial size of degree
$O(dp^{-d})$.
\paragraph*{Acknowledgments:}
The conceptual and technical ideas underlying this work were developed jointly by the authors during a visit of RK at CQT. RK thanks CQT for their hospitality. He also  gratefully acknowledges support from the Simons Center for Geometry and Physics, Stony Brook University,
where some 
of the subsequent work on this paper was performed, 
as well as the Isaac Newton Institute for Mathematical Sciences, Cambridge, for support and hospitality during the programme
``Mathematics of many-body entanglement''. MT is supported by the NRF Investigatorship award (NRF-NRFI10-2024-0006).
RK gratefully acknowledges support by the European Research Council under Grant No. 101001976 (project EQUIPTNT) and the Munich Quantum Valley, which is supported by the Bavarian state government through the Hightech Agenda Bayern Plus.

\paragraph*{Declaration of use of artificial intelligence:}
Claude Fable 5 (Anthropic), GPT-5.6 Sol and GPT-6 Astra (OpenAI), working under the direction of the authors, helped to fix notational inconsistencies and typos throughout and helped with final revisions. The translation of the main algorithm from the arithmetic model to Boolean circuits, together with the corresponding analysis in Appendices~\ref{app:boolean-implementation} and~\ref{app:lookup-tables}, was carried out primarily by AI through several iterations of prompting by the authors. The authors have checked the arguments and take responsibility for any remaining inaccuracies.

\begingroup
\small
\bibliographystyle{alpha}
\bibliography{q}
\endgroup

\clearpage

\appendix

\section{Finite-precision analysis and Boolean circuits}
\label{app:boolean-implementation}
In this appendix we give the detailed constructions and analysis of the
finite-precision implementation of \Sample{}
(Algorithm~\ref{alg:full-noisy-interval-sampler}) outlined in
Section~\ref{sec:boolean-realization}. This implementation uses finite binary
words and independent fair coins in place of the exact real numbers and
biased random bits of the real-arithmetic model of
Section~\ref{sec:computational-model}. We bound its approximation error,
word lengths, and Boolean circuit complexity, and prove
Corollaries~\ref{cor:overview-bounded-boolean}
and~\ref{cor:overview-unbounded-boolean}.

We use the following representation and sampling rules, controlled by an
integer precision parameter~$s\geq1$ (to be chosen):
\begin{enumerate}[(i)]
\item The noise probability~$p$ and the real and imaginary parts of the
matrix elements of the physical two-qubit gates are given as integer multiples
of~$2^{-s}$; we refer to these approximate matrix elements as those of rounded physical gates.
\item Every sampling step of \routine{TreeSample} uses $s$~independent uniform coins.
\item Approximations to entries of coarse gates (i.e., rounded coarse gates) are computed by  taking (exact) products of at most~$d^2$
rounded physical gates. Each entry can therefore be written as
$2^{-c}(x+\mathrm{i}y)$ with $x,y\in\mathbb Z$ and $c=d^2s$.
Thus $2^c$ is a common denominator for the real and imaginary parts,
and $c$ is its exponent; this denominator need not be reduced.
\item We choose an integer word length
$\beta=O\bigl(ds(W_*+d)\bigr)$ bits, large enough to store every
intermediate integer computed in the algorithm exactly. In particular, all intermediate sums and
products fit in these words, so no overflow handling or truncation is
needed. Proposition~\ref{thm:full-bit-precision}
proves that this many bits suffice for the choice of~$s$ given in
Eq.~\eqref{eq:precision-choice} below.
\end{enumerate}
All arithmetic after rounding the inputs, in particular the products of
gates and the contractions of amplitude tables, is carried out exactly
on integer numerators. The parameters $c$ and~$\beta$ specify the
representation and introduce no additional rounding. The errors from
rounding the noise probability, rounding the gate entries, and using
finitely many coins add up, and total variation error~$\delta$ is achieved
by choosing
\begin{align}
  s=\left\lceil c_0\bigl(d+\log_2(n/\delta)\bigr)\right\rceil
   =O\bigl(d+\log(n/\delta)\bigr),
  \label{eq:precision-choice}
\end{align}
where $c_0$ is a sufficiently large absolute constant, as in
Proposition~\ref{thm:full-bit-precision}.
Constants are not optimized. Throughout, $n\geq d\geq2$ are even, as in
Theorem~\ref{thm:main}.

Appendix~\ref{app:noise-bits} generates the noise fixing from fair
coins, and Appendix~\ref{app:finite-precision-tree} analyzes
\routine{TreeSample} with rounded gates; this is the main part of the
analysis. Appendix~\ref{app:full-implementation} fixes $s$ and~$\beta$ and
bounds the total error, and Appendix~\ref{app:boolean-arithmetic}
converts the resulting algorithm into Boolean circuits.

We write $d_{\mathrm{TV}}(R,S)$ for the total variation distance of two distributions~$(R,S)$ 
and $\mathcal B(a)$ for the Born
distribution of a nonzero vector $a\in\mathbb C^q$,
$\mathcal B(a)_j=|a_j|^2/\|a\|_2^2$. We use two standard facts: the total variation distance~$d_{\mathrm{TV}}$ does not increase when the same stochastic map is
applied to both distributions, and
$d_{\mathrm{TV}}(\bigotimes_iR_i,\bigotimes_iS_i)\leq\sum_id_{\mathrm{TV}}(R_i,S_i)$
for product distributions. Here $\bigotimes_iR_i$ denotes the distribution
of independent variables with respective distributions~$R_i$; explicitly,
$(\bigotimes_iR_i)(z_1,\ldots,z_m)=\prod_{i=1}^mR_i(z_i)$.

\begin{lemma}
\label{lem:born-probability-stability}
For nonzero $a,b\in\mathbb C^q$, we have
\begin{align}
  d_{\mathrm{TV}}(\mathcal B(a),\mathcal B(b))
  \leq\frac{\|a-b\|_2}{\|a\|_2}.
  \label{eq:normalized-amplitude-comparison}
\end{align}
\end{lemma}

\begin{proof}
Let
\[
x=\frac{|a|}{\|a\|_2},
\qquad
y=\frac{|b|}{\|b\|_2}.
\]
Then $\|x\|_2=\|y\|_2=1$, and
\begin{align}
d_{\mathrm{TV}}(\mathcal B(a),\mathcal B(b))
&=\frac12\sum_j |x_j^2-y_j^2|\\
&\leq \frac12\|x-y\|_2\|x+y\|_2\\
&=\sqrt{1-\langle x,y\rangle^2}.
\end{align}
Since $y$ is a unit vector,
\[
\sqrt{1-\langle x,y\rangle^2}
=\min_{s\in\mathbb R}\|x-sy\|_2,
\]
because the closest point on the line spanned by $y$ is the orthogonal
projection $\langle x,y\rangle y$. In particular, choosing
$s=\|b\|_2/\|a\|_2$ gives
\[
\sqrt{1-\langle x,y\rangle^2}
\leq
\left\|x-\frac{\|b\|_2}{\|a\|_2}y\right\|_2
=
\frac{\||a|-|b|\|_2}{\|a\|_2}
\leq
\frac{\|a-b\|_2}{\|a\|_2},
\]
where the last inequality follows from
$\bigl||a_j|-|b_j|\bigr|\leq |a_j-b_j|$.
\end{proof}

\subsection{Noise fixing with fair coins}
\label{app:noise-bits}

\routine{DrawNoise} (Algorithm~\ref{alg:draw-noise}) draws the noise
fixing~$F$ of Eq.~\eqref{eq:noisefixing} from Bernoulli variables with
parameter~$p^2$ and Pauli variables with $p$-dependent weights. Let
$\mu_p$ denote the distribution of~$F$; see the description after Eq.~\eqref{eq:noisefixing} for its definition. The fair-coin construction uses
the representation of~$\cN_p$ as a mixture of Pauli channels to generate
classical noise variables at each of the $nd$ noise locations
$(t,a)$, $t\in[d]$, $a\in[n]$ of the brickwork circuit.
By Eq.~\eqref{eq:F-pauli-twirl}, this mixture chooses a uniformly random
Pauli with probability~$p$ and the identity otherwise. The lemma below
samples these choices using fair coins with the rounded probability~$p_s$. These samples can then be used to produce a sample~$F\sim \mu_{p_s}$, approximating an ideal sample from the distribution~$\mu_p$. 

\begin{lemma}
\label{lem:finite-bit-noise-fixing}
Let $s\geq1$  and let $\kappa\in \{0,\ldots,2^s\}$ be such that  $p_s=\kappa2^{-s}$ satisfies $|p_s-p|\leq2^{-s}$. Using $nd(s+2)$
independent fair coins one can draw a noise fixing~$F\sim \mu_{p_s}$, together with its cut indicators
$\{E_{1/2+k}\}_{k=1}^{n-1}$, by a Boolean circuit of size $O(nds)$ and
depth $O(\log s+\log d)$ with fan-in two, or of depth~$O(1)$ with
unbounded fan-in gates. Moreover, the distribution~$\mu_{p_s}$ satisfies
\begin{align}
  d_{\mathrm{TV}}(\mu_{p_s},\mu_p)\leq nd\,2^{-s}.
  \label{eq:noise-fixing-bit-error}
\end{align}
\end{lemma}
\begin{proof}
At each location $(t,a)$, draw a uniform $s$-bit integer~$T^{(t)}_a$ (using $s$ coins) and
a uniform Pauli $P^{(t)}_a\in\Pauli$ (using $2$~coins). Let
$B^{(t)}_a=\mathbf 1\{T^{(t)}_a<\kappa\}$, where $\mathbf 1\{\cdot\}$
is the indicator function: it equals $1$ when its condition holds and
$0$ otherwise. Set
\begin{align}
  \Pi^{(t)}_a=
  \begin{cases}
    P^{(t)}_a, &\text{if } B^{(t)}_a=1,\\
    I, &\text{if } B^{(t)}_a=0.
  \end{cases}
\end{align}
Then $B^{(t)}_a\sim\mathsf{Ber}(p_s)$.
For a rectangle $\cR_p(t,j)$ defined in Eq.~\eqref{eq:rectanglepchannel},
let $a=2j-1$ for odd~$t$ and $a=2j$ for even~$t$, so that the rectangle
acts on qubits $a,a+1$. Set
\begin{align*}
  E^{(t)}_j=B^{(t)}_aB^{(t)}_{a+1}
  \qquad\textrm{and}\qquad
  (Q^{(t)}_j,R^{(t)}_j)=(\Pi^{(t)}_a,\Pi^{(t)}_{a+1}) .
\end{align*}
At each boundary noise location~$(t,j)$ with even~$t\in[d]$ and
$j\in\{1,n\}$, set $S^{(t)}_j=\Pi^{(t)}_j$.

The resulting noise fixing is the collection
\begin{align*}
  F=\Bigl(
    \bigl\{(E^{(t)}_j,Q^{(t)}_j,R^{(t)}_j)\bigr\}_{(t,j)\in\gateLocations{\sU}},\
    \bigl\{S^{(t)}_j\bigr\}_{t\in[d]\text{ even},\,j\in\{1,n\}}
  \Bigr)
\end{align*}
of Eq.~\eqref{eq:noisefixing}, and it has distribution~$\mu_{p_s}$.
Indeed, distinct
rectangles and boundary locations use distinct coins and are therefore
independent, as in \routine{DrawNoise}. At a rectangle, $E^{(t)}_j\sim\mathsf{Ber}(p_s^2)$, and conditioning
on~$E^{(t)}_j$ reproduces the two branches $\cF\otimes\cF$ and~$\cG$ of
Eq.~\eqref{eq:rectangle-two-branches} in the proof of
Lemma~\ref{lem:circuitcutting}: given $E^{(t)}_j=1$, the pair
$(Q^{(t)}_j,R^{(t)}_j)$ is uniform on~$\Pauli^2$, and given
$E^{(t)}_j=0$ it has the weights $q_{(Q,R)}$ of
Eq.~\eqref{eq:residual-pauli-probabilities} at parameter~$p_s$. At a
boundary location, $S^{(t)}_j$ has the weights
$(1-3p_s/4,p_s/4,p_s/4,p_s/4)$. This is the distribution of the variables drawn
by \routine{DrawNoise} at parameter~$p_s$.

For Eq.~\eqref{eq:noise-fixing-bit-error}, note that the same
construction with $\mathsf{Ber}(p)$ in place of $\mathsf{Ber}(p_s)$
produces~$\mu_p$. Both fixings are the same deterministic function of
the $nd$ Bernoulli variables and the $nd$ uniform Paulis. Hence
$d_{\mathrm{TV}}(\mu_{p_s},\mu_p)
\leq d_{\mathrm{TV}}(\mathsf{Ber}(p_s)^{\otimes nd},\mathsf{Ber}(p)^{\otimes nd})
\leq nd\,|p_s-p|$.

Each comparison $T^{(t)}_a<\kappa$ is a prefix computation of size
$O(s)$ and depth $O(\log s)$ with fan-in two, or  of depth~$O(1)$ with
unbounded fan-in gates (test in parallel, for every bit position, whether it
is the most significant position where $T^{(t)}_a$ and~$\kappa$
differ). The cut indicator $E_{1/2+k}$ is the AND of the $d$ variables
$B^{(t)}_a$ with $a\in\{k,k+1\}$ at the $d/2$ layers in which a gate
crosses the bond~$k$; it has depth $O(\log d)$ or $O(1)$, respectively.
\end{proof}

\subsection{Local sampling at finite precision}
\label{app:finite-precision-tree}

We analyze \routine{TreeSample} (Algorithm~\ref{alg:joining-tree}) when
its gates are given only approximately because of rounding and its sampling steps use
finitely many coins. The argument adapts the robustness analysis of
Bravyi, Gosset, and Liu~\cite[Lemma~1]{BravyiGossetLiu22} to the joining
tree introduced in Section~\ref{sec:coarsening-subsec}.

Consider a bilayer circuit on $2M$ qudits of dimension~$K$ with gates
$U_1,\ldots,U_M$ and $V_1,\ldots,V_{M-1}$ as in
Eq.~\eqref{eq:basic-bilayer-circuit}. Each node of the joining tree is
associated with a valid interval of qudit indices. There are $M$ leaves:
leaf~$i$, for $i\in\{1,\ldots,M\}$, is associated with the
two-qudit interval~$\{2i-1,2i\}$. Each internal node is associated with
the union~$I\cup J$ of the neighboring intervals of its two children
and with the interface gate~$V_b$ acting between them, as in
Section~\ref{sec:joining-intervals}. There are exactly $M-1$ internal
nodes, since each join reduces the number of active intervals by one;
carrying an unpaired interval to the next level creates no new node.
The root is associated with the full interval~$[2M]=\{1,\ldots,2M\}$.
Fig.~\ref{fig:joining-tree} illustrates these interval labels and interface
gates for $M=4$.

Suppose that every coarse gate
$G\in\{U_1,\ldots,U_M,V_1,\ldots,V_{M-1}\}$, expressed as a matrix in
the computational basis, is replaced by a rounded matrix~$\widetilde G$ with
\begin{align}
  \|\widetilde G-G\|_{2\to2}\leq\nu,
  \label{eq:join-local-roundoff}
\end{align}
where $\|\cdot\|_{2\to2}$ is the operator norm induced by the Euclidean
norm, and that all arithmetic is exact. Define the rounded interval
vectors, which need not be normalized, recursively by
\begin{align}
  \widehat\Psi_{\{2i-1,2i\}}=\widetilde U_i\ket{00}
  \qquad\textrm{and}\qquad
  \widehat\Psi_{I\cup J}=\widetilde V_b\bigl(\widehat\Psi_I\otimes\widehat\Psi_J\bigr).
  \label{eq:join-amplitude-invariant}
\end{align}
The first equality applies to every leaf~$i\in\{1,\ldots,M\}$, and
the second to every pair of neighboring child intervals~$I,J$ at an
internal node of the joining tree, with~$V_b$ their interface gate.
The restricted interval states $\Psi_I:=\ket{\Psi(I)}$ of
Section~\ref{sec:bilayer-subsec} satisfy the same recursion with the
exact gates, by Eq.~\eqref{eq:union-state-identity}; $\Psi_I$ is a unit
vector, and $P_{Z^I}=\mathcal B(\Psi_I)$ is the output distribution
on~$I$, which at the root is the target distribution~$P_{Z([2M])}$. By
Eq.~\eqref{eq:union-table}, when \routine{TreeSample} uses the rounded
gates and exact arithmetic, its table at node~$I$ contains the amplitudes
of~$\widehat\Psi_I$ for all boundary outcomes with the sampled interior
string~$z^I_\circ$ held fixed. Denoting this table by~$\widehat T_I$, its
entries are
\begin{align*}
  \widehat T_I(z^I_-,z^I_+\given z^I_\circ)
  =\langle z^I_-,z^I_\circ,z^I_+|\widehat\Psi_I\rangle,
  \qquad z^I_-,z^I_+\in[K],
\end{align*}
with the endpoint conventions of Section~\ref{sec:bilayer-subsec}.
The full vector~$\widehat\Psi_I$ is used only for the analysis; the
algorithm retains only the table for the sampled interior string. Rounded gates
need not be unitary, so \routine{TreeSample} need not sample the Born
distribution of the rounded vector; we therefore track the amplitude
error $\|\widehat\Psi_I-\Psi_I\|_2$ and the sampling error separately.
In the following lemma, $c$ specifies the denominators of the gate
entries, while $s$ is the number of fair coins used per draw.

\begin{lemma}[Finite-precision \routine{TreeSample}]
\label{cor:tree-bit-precision}
Let $c,s\geq1$ be integers. Suppose that the gates $\widetilde U_i$ and
$\widetilde V_b$ satisfy Eq.~\eqref{eq:join-local-roundoff} with
$2M\nu\leq1$, and that their entries are of the form
$2^{-c}(x+\mathrm iy)$ with $x,y\in\mathbb Z$. Then \routine{TreeSample} on $2M$ qudits of dimension~$K$
can be implemented with exact integer arithmetic on words of
$O(cM+s+\log K)$ bits, without divisions, using $s$~fresh fair coins at
every node, such that its output distribution $\widehat P_{Z([2M])}$
satisfies
\begin{align}
  d_{\mathrm{TV}}\bigl(\widehat P_{Z([2M])},P_{Z([2M])}\bigr)
  \leq8M^2\nu+MK^2\,2^{-s} .
  \label{eq:finite-precision-tree-error}
\end{align}
\end{lemma}
\begin{proof}
Since no rounding takes place after the input, only the growth of the
denominators and numerators has to be controlled. By induction along
Eq.~\eqref{eq:join-amplitude-invariant}, every entry
of $\widehat\Psi_I$ is of the form $2^{-c(|I|-1)}(x+\mathrm iy)$ with
$x,y\in\mathbb Z$: a leaf involves one gate, and at a join each term of
Eq.~\eqref{eq:union-table} is a product of two child amplitudes and one
gate entry, so the exponents add up to $c(|I|-1)+c(|J|-1)+c=c(|I\cup J|-1)$.
The tables of \routine{TreeSample} are therefore computed exactly by
integer additions and multiplications of numerators. Since
$\|\widetilde G\|_{2\to2}\leq1+\nu$ for every gate,
$\|\widehat\Psi_I\|_2\leq(1+\nu)^{|I|-1}\leq e^{2M\nu}\leq e$, so every
numerator has $O(cM)$ bits. Squaring numerators, summing $K^2$ such
terms, and multiplying by an $s$-bit integer, as done below, therefore
requires $O(cM+\log K+s)$ bits.

Each sampling step draws from $[K]^2$ with weights proportional to the
squared moduli of $K^2$ computed amplitudes, which we collect in a
vector $b\in\mathbb C^{K^2}$: the entries of $\widetilde U_i\ket{00}$
at a leaf, and the table entries for the retained values at a join. Index the candidates by $j\in[K^2]$, let
$\omega_j=x_j^2+y_j^2$ be the integer weights, let
$S_j=\omega_1+\cdots+\omega_j$ be their partial sums, $S_0=0$, and let
$S=S_{K^2}$. If $S=0$, output a fixed value. Otherwise draw a uniform
$R\in\{0,\ldots,2^s-1\}$ from $s$ coins and output the unique~$j$ with
\begin{align}
  2^sS_{j-1}\leq RS<2^sS_j .
  \label{eq:integer-sampling-thresholds}
\end{align}
The admissible values of~$R$ form a half-open interval of length
$2^s\omega_j/S$, which contains between $\lfloor2^s\omega_j/S\rfloor$
and $\lceil2^s\omega_j/S\rceil$ integers. Hence the probability of~$j$
differs from $\omega_j/S$ by at most~$2^{-s}$, and the step draws within
total variation distance $\gamma:=K^22^{-s-1}$ of $\mathcal B(b)$
whenever $b\neq0$.

For every node~$I$ let $\widehat P_{Z^I}$ be the distribution of the
sample returned at~$I$, and let
$\epsilon_I:=\|\widehat\Psi_I-\Psi_I\|_2$ and
$\alpha_I:=d_{\mathrm{TV}}(\widehat P_{Z^I},P_{Z^I})$. We claim that
\begin{align}
\begin{split}
  \alpha_I&\leq\epsilon_I+\gamma\qquad\textrm{at a leaf,}\\
  \alpha_{I\cup J}&\leq\alpha_I+\alpha_J+\epsilon_{I\cup J}+\gamma
  \qquad\textrm{at a join.}
\end{split}
  \label{eq:join-sampling-error-recurrence}
\end{align}
At a leaf with $\widehat\Psi_I\neq0$, the sample is drawn within
distance~$\gamma$ of $\mathcal B(\widehat\Psi_I)$, and
Lemma~\ref{lem:born-probability-stability} with $\|\Psi_I\|_2=1$ gives
$d_{\mathrm{TV}}(\mathcal B(\widehat\Psi_I),\mathcal B(\Psi_I))\leq\epsilon_I$.
If $\widehat\Psi_I=0$, then $\epsilon_I=1\geq\alpha_I$.

Consider a join of $I$ and~$J$ with interface gate~$V$. Write an outcome
on~$I\cup J$ as $(\zeta,(u,v))$, where $(u,v)\in[K]^2$ are the interface
values and $\zeta=(z^I_-,z^I_\circ,z^J_\circ,z^J_+)$ collects the
retained values, and let $a_\zeta,b_\zeta\in\mathbb C^{K^2}$ be the
slices of $\Psi_{I\cup J}$ and~$\widehat\Psi_{I\cup J}$ with retained values~$\zeta$,
\begin{align*}
  a_\zeta(u,v)=\langle z^I_-,z^I_\circ,u,v,z^J_\circ,z^J_+|\Psi_{I\cup J}\rangle
  \qquad\textrm{and}\qquad
  b_\zeta(u,v)=\langle z^I_-,z^I_\circ,u,v,z^J_\circ,z^J_+|\widehat\Psi_{I\cup J}\rangle .
\end{align*}
The slices partition the coordinates of the full vectors, so
\begin{align}
  \sum_\zeta\|a_\zeta\|_2^2=1
  \qquad\textrm{and}\qquad
  \sum_\zeta\|b_\zeta-a_\zeta\|_2^2=\epsilon_{I\cup J}^2 .
  \label{eq:join-squared-residuals}
\end{align}
By Eq.~\eqref{eq:union-table}, $b_\zeta$ is the vector of computed table
entries for the retained values~$\zeta$, so the implemented sampling
step draws $(u,v)$ within distance~$\gamma$ of $\mathcal B(b_\zeta)$ if
$b_\zeta\neq0$, and returns the fixed value if $b_\zeta=0$. The exact
step draws $(u,v)$ from $\mathcal B(a_\zeta)$, Eq.~\eqref{eq:join-conditional-born}.
Both steps retain~$\zeta$ and use fresh coins.

Replacing the child distributions $\widehat P_{Z^I}\otimes\widehat P_{Z^J}$
at the input of the implemented step by $P_{Z^I}\otimes P_{Z^J}$ changes
its output by at most $\alpha_I+\alpha_J$, by the two facts stated at
the beginning of this appendix. It remains to compare the implemented
step with the exact step on the same exact input.

With exact child distributions, the components of $\Psi_I\otimes\Psi_J$
with retained values~$\zeta$ are mapped to~$a_\zeta$ by the unitary~$V$,
so $\zeta$ has probability $\|a_\zeta\|_2^2$, and the exact step
produces~$P_{Z^{I\cup J}}$ (Lemma~\ref{lem:join}). Let $R_{Z^{I\cup J}}$ be the output
distribution of the implemented step on these exact child distributions,
and let $R_\zeta$ be its conditional distribution of $(u,v)$ given~$\zeta$.
Since both steps retain~$\zeta$, they have the same marginal
$\|a_\zeta\|_2^2$. Thus
\begin{align*}
  d_{\mathrm{TV}}(R_{Z^{I\cup J}},P_{Z^{I\cup J}})
  &=\frac12\sum_{\zeta:a_\zeta\neq0}\|a_\zeta\|_2^2
    \sum_{u,v\in[K]}\bigl|R_\zeta(u,v)-\mathcal B(a_\zeta)_{(u,v)}\bigr|\\
  &=\sum_{\zeta:a_\zeta\neq0}\|a_\zeta\|_2^2
    d_{\mathrm{TV}}(R_\zeta,\mathcal B(a_\zeta)) .
\end{align*}
Terms with $a_\zeta=0$ have zero probability and are omitted. For
$a_\zeta\neq0$, the triangle inequality and
Lemma~\ref{lem:born-probability-stability} give
\begin{align*}
  d_{\mathrm{TV}}(R_\zeta,\mathcal B(a_\zeta))
  \leq\gamma+\frac{\|b_\zeta-a_\zeta\|_2}{\|a_\zeta\|_2} .
\end{align*}
This also holds when $b_\zeta=0$, since the fraction then equals~$1$
and total variation distance is at most~$1$. Multiplying by
$\|a_\zeta\|_2^2$ leaves a factor~$\|a_\zeta\|_2$, so rare retained outcomes,
whose conditional distributions are sensitive to small amplitude errors,
cause no difficulty. By the Cauchy--Schwarz inequality and
Eq.~\eqref{eq:join-squared-residuals},
\begin{align*}
  \alpha_{I\cup J}
  &\leq\alpha_I+\alpha_J+d_{\mathrm{TV}}(R_{Z^{I\cup J}},P_{Z^{I\cup J}})\\
  &\leq\alpha_I+\alpha_J+\gamma+\sum_\zeta\|a_\zeta\|_2\,\|b_\zeta-a_\zeta\|_2
  \leq\alpha_I+\alpha_J+\gamma+\epsilon_{I\cup J} .
\end{align*}

$\Psi_I$ and $\widehat\Psi_I$ are products of the same $|I|-1$ gates
applied to $\ket{0^{|I|}}$, exact and rounded respectively. Insert
intermediate products by replacing the exact gates by their rounded
versions one at a time, starting with the leftmost factor. Consecutive
products differ in exactly one gate. Their differences telescope to
the difference between the fully rounded and exact products. The term
with $|I|-1-i$ rounded gates to the left of the differing gate and
$i-1$ exact gates to its right has operator norm at most
$(1+\nu)^{|I|-1-i}\nu$, by submultiplicativity of the operator norm,
since exact gates have norm~$1$, rounded gates have norm at most
$1+\nu$, and their differences have norm at most~$\nu$. Applying the
triangle inequality to the sum and using $\|\ket{0^{|I|}}\|_2=1$ gives
\begin{align}
  \epsilon_I\leq\sum_{i=1}^{|I|-1}(1+\nu)^{|I|-1-i}\nu
  =(1+\nu)^{|I|-1}-1\leq e^{2M\nu}-1\leq4M\nu ,
  \label{eq:join-amplitude-error-recurrence}
\end{align}
where the last step uses $e^x-1\leq2x$ for $0\leq x\leq1$. Applying
Eq.~\eqref{eq:join-sampling-error-recurrence} at the root and then
recursively at every join, and summing over the $2M-1$ nodes (an
unpaired interval carried to the next level of \routine{TreeSample} is
not a node and contributes nothing), gives
\begin{align*}
  \alpha_{[2M]}
  \leq\sum_I(\epsilon_I+\gamma)
  \leq(2M-1)(4M\nu+\gamma)
  \leq8M^2\nu+MK^2\,2^{-s},
\end{align*}
which is Eq.~\eqref{eq:finite-precision-tree-error}.
\end{proof}

\subsection{Choice of the precision}
\label{app:full-implementation}
\label{app:containment-bits}

Let $W_*\geq1$ and $\ell=2W_*+4$ be as in \Sample{}, and let
$\mathcal L$ be the family of intervals of the two partitions
$\{A_r\}_r$ and $\{B_s\}_s$ of Section~\ref{sec:two-forests}. Each
partition covers~$[n]$, and the shifted partition has at most one
interval more than the unshifted one, so that
\begin{align}
  \sum_{C\in\mathcal L}|C|=2n
  \qquad\textrm{and}\qquad
  |\mathcal L|\leq2\lceil n/\ell\rceil+1=O(n) .
  \label{eq:interval-family}
\end{align}
Each $C\in\mathcal L$ is handled by one call to \routine{LocalSample},
which runs \routine{TreeSample} on the bilayer reduction of the backward
light cone of~$C$ in~$\sU$ (Section~\ref{sec:arithmetic-realization}).
The light cone has width at most $|C|+2d\leq\ell+2d$ and is padded to a
multiple of~$2d$, so the bilayer chain has qudit dimension $K=2^d$ and
$M_C\leq|C|/(2d)+2$ leaves. Hence
\begin{align}
  M_C=O\Bigl(\frac{W_*}{d}+1\Bigr)
  \qquad\textrm{and}\qquad
  \sum_{C\in\mathcal L}M_C\leq\frac nd+2|\mathcal L|=O(n) .
  \label{eq:tree-size-bounds}
\end{align}
Each coarse gate is a product of at most~$d^2$ two-qubit gates
of~$\tsU(f)$, i.e., of gates of~$\sU$ and Paulis of the noise fixing.

The combining stage of \Sample{}, namely \routine{BoundaryPass}
(Algorithm~\ref{alg:boundary-selection}) with its calls to
\routine{IntervalContainment}, acts on bits only, the cut indicators of
the noise fixing and the interval samples, and
Corollary~\ref{cor:combining-boolean-realization} gives exact Boolean
circuits for it. Finite precision affects this stage only through the
interval samples. Conditional on the noise fixing these are
independent, so their errors add up, and the deterministic combining
step cannot increase the total variation distance. This is used in the
following proposition.

\begin{proposition}
\label{thm:full-bit-precision}
There is an absolute constant $c_0$ such that the following holds. Let
$\delta\in(0,1)$, let $W_*$ be the value of Eq.~\eqref{eq:Wstar-choice}
with $\delta/2$ in place of~$\delta$, and let
$s\geq c_0\,(d+\log_2(n/\delta))$ be an integer. Suppose that the noise
probability is given as $p_s=\kappa2^{-s}$ with $|p_s-p|\leq2^{-s}$, and
that every two-qubit gate~$U$ of~$\sU$ is given as a matrix~$\widetilde U$
whose entries have real and imaginary parts in $2^{-s}\mathbb Z$ and
differ from the entries of~$U$ by at most~$2^{-s}$ in modulus. Then
\Sample{} can be implemented with exact integer arithmetic on words of
\begin{align}
  \beta=O\bigl(ds\,(W_*+d)\bigr)
  \label{eq:full-bit-prescription}
\end{align}
bits, using $O(s)$ fair coins per noise location and per node of the
joining trees, such that its output distribution satisfies
$d_{\mathrm{TV}}(\widehat P,P^{(p)})\leq\delta$.
\end{proposition}

\begin{proof}
The operator norm is at most the Frobenius norm, so
$\|\widetilde U-U\|_{2\to2}\leq4\cdot2^{-s}$ for every two-qubit gate
(16 entries, each within~$2^{-s}$); multiplication by Paulis is exact
and preserves this bound. The rounded coarse gate~$\widetilde G$ is
computed exactly as a product of at most~$d^2$ rounded elementary gates,
whose entries are multiples of~$2^{-s}$; its entries are therefore of
the form $2^{-c}(x+\mathrm iy)$ with $c:=d^2s$ and $x,y\in\mathbb Z$.
For the original coarse gate~$G$, the telescoping argument of
Eq.~\eqref{eq:join-amplitude-error-recurrence} gives
\begin{align}
  \|\widetilde G-G\|_{2\to2}\leq(1+4\cdot2^{-s})^{d^2}-1\leq8d^22^{-s}=:\nu ,
  \label{eq:coarse-gate-from-elementary-error}
\end{align}
using $e^x-1\leq2x$ with $x=4d^22^{-s}$, which is at most~$1$ for
$c_0\geq2$. Since $M_C=O(n)$ by Eq.~\eqref{eq:tree-size-bounds}, we have
$2M_C\nu=O(nd^22^{-s})\leq1$ for every~$C$ if $c_0$ is large enough.
Lemma~\ref{cor:tree-bit-precision} therefore applies to every joining
tree, with $K=2^d$ and $s$~coins per node.

Fix a noise fixing~$f$, and let $\widehat P^f$ and $P^f$ be the
conditional output distributions given $F=f$ of the implementation and
of the real-arithmetic algorithm \Sample{}, respectively. For every
$C\in\mathcal L$, the implemented sample on~$C$ is a fixed function of
the output of \routine{TreeSample} on the light cone of~$C$, and the
same function of the exact output has distribution~$Q^f_{Z_C}$;
Lemma~\ref{cor:tree-bit-precision} therefore bounds their distance by
$8M_C^2\nu+M_C4^d2^{-s}$. Conditional on~$f$, the interval samples are
independent and \routine{BoundaryPass} is a deterministic function of
them, so the two facts stated at the beginning of this appendix give
\begin{align*}
  d_{\mathrm{TV}}(\widehat P^f,P^f)
  \leq\sum_{C\in\mathcal L}\bigl(8M_C^2\nu+M_C4^d2^{-s}\bigr)
  \qquad\textrm{for every noise fixing~$f$.}
\end{align*}
The output distribution~$\widehat P$ of the implementation is the
average of $\widehat P^F$ over $F\sim\mu_{p_s}$, and $d_{\mathrm{TV}}$
is jointly convex, so the same bound holds between $\widehat P$ and the
average of $P^F$ over $F\sim\mu_{p_s}$. Since $f\mapsto P^f$ is a
stochastic map, replacing $\mu_{p_s}$ by~$\mu_p$ in this average changes
it by at most $d_{\mathrm{TV}}(\mu_{p_s},\mu_p)\leq nd\,2^{-s}$
(Lemma~\ref{lem:finite-bit-noise-fixing}). The average of $P^F$ over
$F\sim\mu_p$ is the output distribution of \Sample{}, which is within
$\Pr_p[\Bad]$ of~$P^{(p)}$ by Corollary~\ref{cor:l1-approximation},
where $\Bad$ is the event of Eq.~\eqref{eq:good-bad-events} and the
probability is taken at the original noise parameter~$p$. Altogether,
\begin{align}
  d_{\mathrm{TV}}(\widehat P,P^{(p)})
  \leq\underbrace{\Pr\nolimits_p[\Bad]}_{\textrm{width cutoff}}
  +\underbrace{nd\,2^{-s}}_{\textrm{noise}}
  +\underbrace{8\nu\sum\nolimits_CM_C^2}_{\textrm{gates}}
  +\underbrace{4^d2^{-s}\sum\nolimits_CM_C}_{\textrm{sampling}} .
  \label{eq:full-error-budget}
\end{align}
The first term is at most $\delta/2$ by Lemma~\ref{lem:maxcircuitwidth}
and the choice of~$W_*$. By Eq.~\eqref{eq:tree-size-bounds},
$\sum_CM_C^2\leq(\sum_CM_C)^2=O(n^2)$ and $\sum_CM_C=O(n)$, and
$\nu=8d^22^{-s}$, so the other three terms sum to $2^{O(d)}n^2\,2^{-s}$,
which is at most $\delta/2$ if $c_0$ is large enough.

The integer word length~$\beta$ is determined by the largest integers that occur.
By Lemma~\ref{cor:tree-bit-precision} and
Eq.~\eqref{eq:tree-size-bounds}, the joining trees use integers of
$O(cM_C+s+d)=O(d^2s(W_*/d+1))=O(ds(W_*+d))$ bits. The partial products
of at most~$d^2$ two-qubit gates which form the coarse gates have
denominators at most~$2^c$ and operator norm at most $1+\nu\leq2$, so
they use integers of $O(c+d)$ bits. The noise generation uses $s$-bit
integers and $s+2$ coins per location
(Lemma~\ref{lem:finite-bit-noise-fixing}).
\end{proof}

From now on we take $s$ as in Eq.~\eqref{eq:precision-choice}, so that
$s=O(d+\log(n/\delta))$ and $\beta=O(ds(W_*+d))$. For constant~$d$, the
gate entries and the noise probability thus need to be specified with
$O(\log(n/\delta))$ bits only, as stated in
Section~\ref{sec:boolean-overview}. If moreover $p$ is constant and
$\delta$ is constant or inverse-polynomial in~$n$, then
$W_*=O(\log n)$ by Lemma~\ref{lem:maxcircuitwidth}, hence
\begin{align}
  s=O(\log n)
  \qquad\textrm{and}\qquad
  \beta=O\bigl((\log n)^2\bigr).
  \label{eq:constant-depth-parameters}
\end{align}

\subsection{Boolean circuit bounds}
\label{app:boolean-arithmetic}

We now convert the implementation of
Proposition~\ref{thm:full-bit-precision} into Boolean circuits. Fan-in
two gives Corollary~\ref{cor:overview-bounded-boolean}, and collapsing
groups of levels with unbounded fan-in gives
Corollary~\ref{cor:overview-unbounded-boolean}.

\begin{theorem}[Boolean circuits with fan-in two]
\label{thm:full-bounded-boolean}
Let $W_*$, $s$, and $\beta$ be as in Proposition~\ref{thm:full-bit-precision},
with $s$ chosen as in Eq.~\eqref{eq:precision-choice}.
Then \Sample{} has a randomized Boolean circuit~$\cC_2$ with fan-in-two
AND and OR gates and NOT gates whose output distribution is within total
variation distance~$\delta$ of $P^{(p)}$, and whose depth and size satisfy
\begin{align}
\begin{matrix}
  \mathsf{depth}(\cC_2)&=&O\bigl(\log(W_*+d)\,(d+\log \beta)\bigr),\\
  \mathsf{size}(\cC_2)&=&n\,2^{O(d)}\,\beta^2.
\end{matrix}
  \label{eq:full-bounded-boolean}
\end{align}
The circuit can be constructed from $n$, $d$, $W_*$, and~$s$ in time
polynomial in its size. In particular, for constant~$p$, constant or
inverse-polynomial~$\delta$, and depth $d=O(\log\log n)$,
\begin{align}
\begin{matrix}
  \mathsf{depth}(\cC_2)&=&O\bigl((\log\log n)^2\bigr),\\
  \mathsf{size}(\cC_2)&=&n(\log n)^{O(1)} .
\end{matrix}
  \label{eq:loglog-depth-bounded-boolean}
\end{align}
\end{theorem}
\begin{proof}
Addition and comparison of $\beta$-bit integers are carry computations,
i.e., parallel prefix computations of depth $O(\log \beta)$ and size $O(\beta)$,
and selection between two words according to a bit is a multiplexer of
constant depth. The sum of $q$ words of $\beta$ bits, if it fits into $\beta$
bits, is computed by repeatedly replacing three words by two words of
the same sum using bitwise full adders, which takes $O(\log q)$ rounds
of constant depth and $O(q\beta)$ gates, followed by one addition; applied
to the $\beta$ shifted partial products of two $\beta$-bit integers, this is
the multiplier of Wallace~\cite{Wallace64}, of depth $O(\log \beta)$ and
size $O(\beta^2)$.

With $K=2^d$, the product of two $K^2\times K^2$ matrices has $K^4$
entries, each a sum of $K^2$ products of $\beta$-bit integers, giving depth
$O(\log K+\log \beta)=O(d+\log \beta)$ and size $O(K^6\beta^2)$. The table of
Eq.~\eqref{eq:union-table} computed by \routine{JoinIntervals} has $K^4$
entries of the same form and the same bounds.

The sampling step of Lemma~\ref{cor:tree-bit-precision} consists of
$K^2$ squarings, the $K^2$ partial sums~$S_j$, each a sum of at most
$K^2$ words, the products $RS$ and $2^sS_j$, and $K^2$ comparisons,
computed in a constant number of stages with operations parallelized
within each stage, followed by a selection; it has depth $O(d+\log \beta)$ and
size $O(K^4\beta^2)$. Hence one call to \routine{JoinIntervals} including
its sampling step, and likewise the initialization and sampling at a
leaf, have depth $O(d+\log \beta)$ and size $O(K^6\beta^2)$.

The circuit consists of four stages, and its random inputs are the fair
coins of Proposition~\ref{thm:full-bit-precision}. By
Eq.~\eqref{eq:tree-size-bounds}, the joining trees have
$\sum_C(2M_C-1)=O(n)$ nodes in total and
$\lceil\log_2M_C\rceil+1=O(\log(W_*/d+2))$ levels each, and there is
one coarse gate per node. Each coarse gate is the product of at
most~$d^2$ matrices, which we compute along a balanced binary tree in
$O(\log d)$ rounds of matrix multiplication. With
Lemma~\ref{lem:finite-bit-noise-fixing}, the arithmetic bounds above,
and Corollary~\ref{cor:combining-boolean-realization}, the four stages
have the following depth and total size.
\begin{center}
\begin{tabular}{lll}
\toprule
Stage & Depth & Size \\
\midrule
Noise fixing and cut indicators & $O(\log s+\log d)$ & $O(nds)$ \\
Coarse gates & $O(\log d\,(d+\log \beta))$ & $O(nd^2K^6\beta^2)$ \\
Joining trees & $O(\log(W_*/d+2)\,(d+\log \beta))$ & $O(nK^6\beta^2)$ \\
\routine{BoundaryPass} & $O(\log\ell)$ & $O(n)$ \\
\bottomrule
\end{tabular}
\end{center}
Since $\log d+\log(W_*/d+2)=\log(W_*+2d)$, $\log\ell=O(\log(W_*+d))$,
and $s\leq \beta$, the column sums are Eq.~\eqref{eq:full-bounded-boolean}.
The error bound is Proposition~\ref{thm:full-bit-precision}. The
interval families, the light cones, the joining trees, and the word
length are fixed by $n$, $d$, $W_*$, and~$s$, and the arithmetic
circuits above are standard constructions, so the circuit can be
constructed in time polynomial in its size.

Let $p$ be constant. Then $2^{O(d)}=(\log n)^{O(1)}$,
Lemma~\ref{lem:maxcircuitwidth} gives
$W_*=O(p^{-d}\log(n/\delta))=(\log n)^{O(1)}$, and
Proposition~\ref{thm:full-bit-precision} gives $s=O(\log n)$ and
$\beta=O(ds(W_*+d))=(\log n)^{O(1)}$. Hence $\log(W_*+d)$ and $d+\log \beta$ are
both $O(\log\log n)$, and Eq.~\eqref{eq:full-bounded-boolean} gives
Eq.~\eqref{eq:loglog-depth-bounded-boolean}.
\end{proof}

\begin{proof}[Proof of Corollary~\ref{cor:overview-bounded-boolean}]
Let $d$ and $p$ be constants and let $\delta$ be constant or
inverse-polynomial in~$n$. By Eq.~\eqref{eq:constant-depth-parameters},
$W_*=O(\log n)$ and $\beta=O((\log n)^2)$, so that $\log(W_*+d)$ and
$d+\log \beta$ are both $O(\log\log n)$. Theorem~\ref{thm:full-bounded-boolean}
therefore gives a circuit~$\cC_2$ of depth $O((\log\log n)^2)$ and size
$O(n\beta^2)=O(n(\log n)^4)$ whose output distribution is within total
variation distance~$\delta$ of~$P^{(p)}$.
\end{proof}

\begin{proof}[Proof of Corollary~\ref{cor:overview-unbounded-boolean}]
Let $\cC_2$ be the circuit of Corollary~\ref{cor:overview-bounded-boolean},
of depth $D=O((\log\log n)^2)$ and size $S=O(n(\log n)^4)$, and let
$n\geq4$. Partition the levels of~$\cC_2$ into $\lceil D/t\rceil$
groups of $t=\lfloor\log_2\log_2n\rfloor\geq1$ consecutive levels.
Every gate in a group depends on at most $2^t\leq\log_2n$ signals
entering the group, so it is a function of these signals which can be
written as an OR of at most $2^{2^t}\leq n$ ANDs of literals. Replacing
every gate by this depth-two circuit, in parallel within each group and
keeping the signals that later groups need, gives a circuit~$\cC_\infty$
with unbounded fan-in AND and OR gates and NOT gates which computes the
same function of its inputs and coins as~$\cC_2$, and hence has the same
output distribution. Its depth is $O(1+D/t)=O(\log\log n)$, since
$t\geq\frac12\log_2\log_2n$ by $\lfloor x\rfloor\geq x/2$ for $x\geq1$,
and its size is $O(nS)=O(n^2(\log n)^4)$. The depth-two replacements are
obtained by evaluating a subcircuit with at most $\log_2n$ inputs on all
its inputs, which takes time polynomial in~$n$, so $\cC_\infty$ can also
be constructed in polynomial time.
\end{proof}

For a classical input $x\in\{0,1\}^k$, the circuits $\cC_2$
and~$\cC_\infty$ take the bits of~$x$ as additional inputs and first
permute the entries of the first-layer gates according to~$x$, as
described in Section~\ref{sec:simulation-proof}. This is a selection of
constant depth and size $O(ns)$, which does not affect the stated
bounds. In particular, Corollaries~\ref{cor:overview-bounded-boolean}
and~\ref{cor:overview-unbounded-boolean} hold for circuits with
classical inputs, as stated in Section~\ref{sec:boolean-overview}.
\clearpage

\section{Constant-depth simulation by lookup tables}
\label{app:lookup-tables}

In this appendix we prove Corollary~\ref{cor:compiled-constant-depth}.
The circuit~$\cC_{\mathrm{lt}}$ keeps the noise fixing, the sampling
intervals, and the combining stage of \Sample{}
(Algorithm~\ref{alg:full-noisy-interval-sampler}), and replaces each call
to \routine{LocalSample} by a single lookup in a precomputed table.

This is possible because the marginal on a sampling interval~$C$ depends
on few bits. Once the noise is fixed, it is determined by the classical
input bits and the noise labels in the backward light cone of~$C$. For
constant $d$ and~$p$ and constant or inverse-polynomial~$\delta$, these
are $O(\log n)$ bits. A table indexed by
these bits and by $O(\log n)$ fair coins therefore has polynomial size,
and it can be evaluated in constant depth. We emphasize that each table
produces the entire string on~$C$ at once. Converting
\routine{TreeSample} step by step would retain the $O(\log\log n)$
dependent levels of its joining tree, and a table indexed by all the
coins of its finite-precision implementation
(Appendix~\ref{app:finite-precision-tree}) would be indexed by a number
of bits of order $(\log n)^2$. For the same reason, the tables are
indexed by the noise labels, which take a constant number of bits per
location, and not by the coins from which these labels are generated.

Appendix~\ref{app:lookup-sampling} gives the lookup-table sampler,
Appendix~\ref{app:lookup-marginals} computes the tables from
finite-precision gate descriptions, and
Appendix~\ref{app:compiled-simulation} assembles the circuit. Constants
are not optimized.

\subsection{Sampling by table lookup}
\label{app:lookup-sampling}

We first show that a fixed family of distributions indexed by $b$~bits
can be sampled in constant depth, at a size exponential in~$b$ and in the
number of coins.

\begin{lemma}
\label{lem:lookup-table-sampling}
Let $b\geq0$, $m\geq1$, and $r\geq1$ be integers, and let
$\{P_a\}_{a\in\bits^b}$ be a family of probability distributions
on~$[m]$. There is a Boolean circuit~$\cT$ with unbounded fan-in AND and
OR gates and NOT gates which takes as input $a\in\bits^b$ and $r$
independent fair coins, and outputs the binary encoding of some
$j\in[m]$. Its output distribution~$P^{\cT}_a$ for fixed input~$a$
satisfies
\begin{align}
  d_{\mathrm{TV}}\bigl(P^{\cT}_a,P_a\bigr)\leq m\,2^{-r-1}
  \qquad\text{for every }a\in\bits^b .
  \label{eq:lookup-tv-error}
\end{align}
The circuit has depth three, including negations, and size
\begin{align}
  O\bigl(2^{b+r}+\log m\bigr).
  \label{eq:lookup-size}
\end{align}
If all probabilities $P_a(j)$ are rational, the circuit can be
constructed in time polynomial in $2^{b+r}$, $m$, and their bit length.
\end{lemma}
\begin{proof}
For each~$a$, define the cumulative probabilities and integer thresholds
\begin{align}
  F_a(j)=\sum_{i=1}^jP_a(i),\qquad
  t_a(j)=\bigl\lceil2^rF_a(j)\bigr\rceil,
  \qquad 0\leq j\leq m .
  \label{eq:lookup-thresholds}
\end{align}
The thresholds are nondecreasing, with $t_a(0)=0$ and $t_a(m)=2^r$.
Identify $u\in\bits^r$ with an integer in $\{0,\ldots,2^r-1\}$, and let
$T(a,u)$ be the unique $j\in[m]$ with $t_a(j-1)\leq u<t_a(j)$. For a
uniformly random~$u$,
\begin{align}
  \Pr[T(a,u)=j]=2^{-r}\bigl(t_a(j)-t_a(j-1)\bigr)
  =P_a(j)+2^{-r}\bigl(\epsilon_a(j)-\epsilon_a(j-1)\bigr),
\end{align}
where $\epsilon_a(j)=t_a(j)-2^rF_a(j)\in[0,1)$. Each probability is
therefore within~$2^{-r}$ of~$P_a(j)$, and summing over~$j$ and dividing
by two gives Eq.~\eqref{eq:lookup-tv-error}.

It remains to implement the fixed function $T\colon\bits^{b+r}\to[m]$.
For $z\in\bits^{b+r}$, the AND of the literals $v_i$ (if $z_i=1$) and
$\neg v_i$ (if $z_i=0$), $i\in[b+r]$, is the indicator $E_z(v)$ of the
event $v=z$. Output bit~$k$ of the binary encoding of $T(v)-1$ is
\begin{align}
  D_k(v)=\bigvee_{z:\ \text{bit $k$ of }T(z)-1\text{ is }1}E_z(v),
  \label{eq:lookup-dnf}
\end{align}
where an empty disjunction is the constant~$0$. Exactly one indicator
$E_z(v)$ equals one, so the outputs encode~$T(v)$. The circuit consists
of $b+r$ negations, $2^{b+r}$ AND gates shared by all output bits, and
$\max\{1,\lceil\log_2m\rceil\}$ OR gates, in three layers. For rational
probabilities, the thresholds and the function~$T$ are computed by exact
rational arithmetic within the stated time.
\end{proof}

Unlike in Lemma~\ref{lem:discrete-sampling-cost}, the probabilities are
not inputs of the circuit: they are fixed when the circuit is
constructed, and only the index~$a$ and the coins are inputs. This is
what makes constant depth possible, at the price of a size exponential
in $b+r$. Below we use Lemma~\ref{lem:lookup-table-sampling} with
$b,r=O(\log n)$, for which the size is polynomial.

\subsection{Tables of interval marginals}
\label{app:lookup-marginals}

Next we define the tables to which Lemma~\ref{lem:lookup-table-sampling}
is applied, and show that they can be computed in polynomial time from
finite-precision gate descriptions.

Let $x\in\bits^k$ with $k\leq n$ be a classical input, and set $x_j=0$
for $k<j\leq n$, so that $\ket{x}\otimes\ket{0^{n-k}}=\ket{x_1\cdots x_n}$.
For a noise fixing~$f$, let
\begin{align}
  Q^{f,x}_Z(z)=\bigl|\bra{z}\tsU(f)\ket{x_1\cdots x_n}\bigr|^2,
  \qquad z\in\bits^n,
  \label{eq:input-fixed-noise-output}
\end{align}
and let $Q^{f,x}_{Z_C}$ denote its marginal on an interval~$C\subseteq[n]$.
For $k=0$ these are the distributions of
Eqs.~\eqref{eq:fixed-noise-output} and~\eqref{eq:local-marginal-target}.
The proof of Lemma~\ref{lem:conditional-factorization} applies to every
input state: the noisy circuit is the average of the unitary channels
of~$\tsU(F)$, so
\begin{align}
  P^{(p)}(\cdot\mid x)=\mathbb{E}_{F\sim\mu_p}\bigl[Q^{F,x}_Z\bigr],
  \label{eq:input-mixture}
\end{align}
and, since the input is a product state, $Q^{f,x}_Z$ factorizes over the
cut-bounded components of~$f$.

Let $C\subseteq[n]$ be an interval, and let $\Lambda(C)$ be its backward
light cone in the brickwork geometry of~$\sU$. It is an interval of width
\begin{align}
  L_C:=|\Lambda(C)|\leq|C|+2d ,
  \label{eq:compiled-light-cone}
\end{align}
and it contains the backward light cone of~$C$ in~$\tsU(f)$ for
every~$f$ (Section~\ref{sec:arithmetic-realization}). Let $f_C$ denote
the part of a noise fixing~$f$ located in~$\Lambda(C)$, namely the
labels $(e^{(t)}_j,q^{(t)}_j,r^{(t)}_j)$ of the rectangles acting on two
qubits of~$\Lambda(C)$ and the Paulis $s^{(t)}_j$ of the boundary noise
locations in~$\Lambda(C)$. The gates of~$\tsU(f)$ at these locations
form a circuit on~$\Lambda(C)$ which contains every gate in the backward
light cone of~$C$. Applied to the input bits of~$x$ in~$\Lambda(C)$, it
therefore has the output marginal~$Q^{f,x}_{Z_C}$ on~$C$. The
\emph{address} of~$C$ is
\begin{align}
  a_C=\bigl(x_{\Lambda(C)\cap[k]},f_C\bigr)\in\bits^{b_C},
  \qquad b_C=O(dL_C),
  \label{eq:compiled-address}
\end{align}
where each label is encoded by at most two bits; the bound on~$b_C$
holds because there are at most $dL_C/2$ such rectangles and at most $d$
such boundary noise locations.

Every $a\in\bits^{b_C}$, including addresses which have probability
zero under~$\mu_p$, specifies in the same way a unitary circuit on~$\Lambda(C)$, with gates
obtained from those of~$\sU$ by the substitution rule of
Section~\ref{sec:convexdecompo}, and a computational basis input state
on~$\Lambda(C)$. Let $Q_{C,a}$ denote the marginal on~$C$ of the output
distribution of this circuit. By the preceding observation,
\begin{align}
  Q_{C,a_C}=Q^{f,x}_{Z_C}
  \qquad\text{whenever }a_C=\bigl(x_{\Lambda(C)\cap[k]},f_C\bigr).
  \label{eq:table-is-marginal}
\end{align}

\begin{lemma}
\label{lem:interval-tables}
Let $C\subseteq[n]$ be an interval, let $g_C$ be the number of
rectangles acting on two qubits of~$\Lambda(C)$, let $\eta\in(0,1)$, and
let $s\geq\log_2(8g_C/\eta)$ be an integer. Suppose that every such gate~$U$ is
given as a matrix~$\widehat U$ whose entries have real and imaginary
parts in $2^{-s}\mathbb Z$ and differ from the entries of~$U$ by at
most~$2^{-s}$ in modulus. Then one can compute, in time polynomial in
$2^{b_C+L_C}$ and~$s$, rational distributions $\widetilde Q_{C,a}$
on~$\bits^C$ such that
\begin{align}
  d_{\mathrm{TV}}\bigl(\widetilde Q_{C,a},Q_{C,a}\bigr)\leq\eta
  \qquad\text{for every }a\in\bits^{b_C}.
  \label{eq:compiled-table-precision}
\end{align}
\end{lemma}
\begin{proof}
Fix an address~$a$, and let $\psi\in\mathbb C^{2^{L_C}}$ be the output
state of the corresponding unitary circuit on~$\Lambda(C)$. Let
$\widehat\psi$ be the vector obtained by replacing every two-qubit
gate~$U$ by~$\widehat U$ and computing exactly; the input state and the
Paulis are exact. The operator norm is at most the Frobenius norm, so
$\|\widehat U-U\|_{2\to2}\leq4\cdot2^{-s}$. Replacing the gates one at a
time gives
\begin{align}
  \|\widehat\psi-\psi\|_2\leq(1+4\cdot2^{-s})^{g_C}-1\leq8g_C\,2^{-s}\leq\eta ,
  \label{eq:compiled-state-error}
\end{align}
where we used $e^y-1\leq2y$ for $0\leq y=4g_C\,2^{-s}\leq1$. In
particular, $\widehat\psi\neq0$. Let $\widetilde Q_{C,a}$ be the
marginal on~$C$ of the Born distribution of~$\widehat\psi$,
\begin{align}
  \widetilde Q_{C,a}(z_C)
  =\frac{\sum_{y:\,y_C=z_C}|\widehat\psi(y)|^2}{\sum_y|\widehat\psi(y)|^2}.
  \label{eq:compiled-rational-marginal}
\end{align}
Since $\psi$ is a unit vector, Lemma~\ref{lem:born-probability-stability}
and Eq.~\eqref{eq:compiled-state-error} bound the distance between the
Born distributions of $\widehat\psi$ and~$\psi$ by~$\eta$, and
marginalization to~$C$ does not increase the total variation distance.
This proves Eq.~\eqref{eq:compiled-table-precision}.

The entries of~$\widehat\psi$ are Gaussian integers divided by the common
denominator~$2^{g_Cs}$, with numerators of $O(g_Cs)$ bits since
$\|\widehat\psi\|_2\leq2$. They are obtained by $g_C$ exact products with
sparse matrices of dimension~$2^{L_C}$, and
Eq.~\eqref{eq:compiled-rational-marginal} requires $2^{L_C}$ squared
moduli. Since every rectangle counted by~$g_C$ carries a label in the address,
$g_C\leq b_C$, and doing this for all $2^{b_C}$ addresses takes time
polynomial in $2^{b_C+L_C}$ and~$s$.
\end{proof}

\subsection{Proof of Corollary~\ref{cor:compiled-constant-depth}}
\label{app:compiled-simulation}

We prove the corollary in its form for classical inputs: there is a
single circuit~$\cC_{\mathrm{lt}}$ which takes $x\in\bits^k$, $k\leq n$,
as an input and whose output distribution $\widehat P(\cdot\mid x)$
satisfies
\begin{align}
  \sup_{x\in\bits^k}
  d_{\mathrm{TV}}\bigl(\widehat P(\cdot\mid x),P^{(p)}(\cdot\mid x)\bigr)
  \leq\delta .
  \label{eq:compiled-total-variation}
\end{align}
The case $k=0$ is Corollary~\ref{cor:compiled-constant-depth} as stated.
As there, $d\geq2$ is an even constant, $p\in(0,1)$ is constant, and
$\delta$ is constant or inverse-polynomial in~$n$.

Let $W_*=\lceil p^{-d}\log(4n/\delta)\rceil+1$ be the width of
Eq.~\eqref{eq:main-effective-width} with $\delta/4$ in place of~$\delta$.
As in the proof of Theorem~\ref{thm:main}, Lemma~\ref{lem:maxcircuitwidth}
gives $\Pr_p[\Bad]\leq\delta/4$ for the event~$\Bad$ of
Eq.~\eqref{eq:good-bad-events}. Let $\ell=2W_*+4$, let $\mathcal L$ be the
family of intervals of the two partitions of
Section~\ref{sec:two-forests}, and let $N=|\mathcal L|=O(n)$
(Eq.~\eqref{eq:interval-family}). Every $C\in\mathcal L$ has
$|C|\leq\ell$, so Eqs.~\eqref{eq:compiled-light-cone}
and~\eqref{eq:compiled-address} give
\begin{align}
  L_C\leq\ell+2d=O\bigl(p^{-d}\log(n/\delta)+d\bigr),
  \qquad
  b_C=O\bigl(dp^{-d}\log(n/\delta)+d^2\bigr),
  \label{eq:compiled-parameters}
\end{align}
both of which are $O(\log n)$. We use a single precision parameter, an
integer~$s$ with
\begin{align}
  2^s\geq32Nnd(\ell+2d)/\delta ,\qquad s=O(\log(n/\delta)),
  \label{eq:compiled-precision}
\end{align}
and assume that $p$ and the real and imaginary parts of the gate entries
are given as multiples of~$2^{-s}$, each within~$2^{-s}$ of its exact
value; for $p$ we write $p_s$ for this approximation. Since any larger
width may be used (Lemma~\ref{lem:maxcircuitwidth}), $W_*$ can also be
computed from these data: replace $p$ by the lower bound
$p_-=p_s-2^{-s}$ and $\log(4n/\delta)$ by
$\lceil\log_2(4n/\delta)\rceil$. For $n\geq1/p$ we have
$2^{1-s}\leq p/(2d)$, hence $p_-^{-d}\leq(1-1/(2d))^{-d}p^{-d}\leq2p^{-d}$,
and $W_*$ changes by at most an absolute constant factor.

The circuit~$\cC_{\mathrm{lt}}$ has three stages. All its random inputs
are independent fair coins.
\begin{enumerate}[(1)]
\item Noise. Lemma~\ref{lem:finite-bit-noise-fixing} generates a noise
fixing~$\widehat F$ with distribution~$\mu_{p_s}$ together with its cut
indicators, and by Eq.~\eqref{eq:compiled-precision}
\begin{align}
  d_{\mathrm{TV}}(\mu_{p_s},\mu_p)\leq nd\,2^{-s}\leq\delta/4 .
  \label{eq:compiled-noise-error}
\end{align}
\item Lookups. For every $C\in\mathcal L$, let $\cT_C$ be the circuit of
Lemma~\ref{lem:lookup-table-sampling} for the distributions
$\widetilde Q_{C,a}$, $a\in\bits^{b_C}$, of
Lemma~\ref{lem:interval-tables} with $\eta=\delta/(4N)$; the hypothesis
$s\geq\log_2(32Ng_C/\delta)$ of that lemma holds by
Eq.~\eqref{eq:compiled-precision}, since $g_C\leq dL_C\leq d(\ell+2d)$. It has
$m_C=2^{|C|}$ outcomes and uses
\begin{align}
  r_C=|C|+\left\lceil\log_2\frac{2N}{\delta}\right\rceil=O(\log n)
  \label{eq:compiled-seed-length}
\end{align}
coins. Its input is the address $a_C=(x_{\Lambda(C)\cap[k]},\widehat F_C)$,
which consists of input bits and noise labels computed in stage~(1), and
its output is a string on~$C$, with distribution~$P^{\cT_C}_a$ for
address~$a$. The rounding error of the table is
$m_C2^{-r_C-1}\leq\delta/(4N)$, so
Eqs.~\eqref{eq:lookup-tv-error} and~\eqref{eq:compiled-table-precision}
and the triangle inequality give
\begin{align}
  d_{\mathrm{TV}}\bigl(P^{\cT_C}_a,Q_{C,a}\bigr)
  \leq\frac{\delta}{2N}
  \qquad\text{for every }a\in\bits^{b_C}.
  \label{eq:compiled-lookup-error}
\end{align}
The circuits~$\cT_C$ use mutually independent coins, which are also
independent of the coins of stage~(1).
\item Assembly. Concatenating the outputs of the circuits~$\cT_C$ on the
intervals $\{A_r\}_r$ and $\{B_s\}_s$ gives candidate strings
$X,Y\in\bits^n$. The circuit outputs the result of the constant-depth
circuit for \routine{BoundaryPass} of
Corollary~\ref{cor:combining-boolean-realization}, applied to $X$, $Y$,
and the cut indicators of~$\widehat F$.
\end{enumerate}

The argument follows the proof of
Proposition~\ref{thm:full-bit-precision}. Fix~$x$ and a noise
fixing~$f$. Conditioned on $\widehat F=f$, the strings on the intervals
$C\in\mathcal L$ are independent, and by
Eqs.~\eqref{eq:table-is-marginal} and~\eqref{eq:compiled-lookup-error}
each is within $\delta/(2N)$ of~$Q^{f,x}_{Z_C}$. Let $R^{f,x}$ be the
distribution obtained by applying \routine{BoundaryPass} to independent
exact samples from these marginals. Since \routine{BoundaryPass} is a
deterministic function of the samples and of~$f$, the conditional output
distribution of~$\cC_{\mathrm{lut}}$ is within $N\cdot\delta/(2N)=\delta/2$
of~$R^{f,x}$. Averaging over~$f$ with joint convexity, and replacing
$\mu_{p_s}$ by~$\mu_p$ with Eq.~\eqref{eq:compiled-noise-error}, gives
$d_{\mathrm{TV}}(\widehat P(\cdot\mid x),\mathbb E_{F\sim\mu_p}R^{F,x})
\leq\delta/2+\delta/4$. Finally, $Q^{f,x}_Z$ factorizes over the
cut-bounded components of~$f$, so Theorem~\ref{thm:local-selection}
gives $R^{f,x}=Q^{f,x}_Z$ whenever $W_{\max}<W_*$. With
Eq.~\eqref{eq:input-mixture}, the argument of
Corollary~\ref{cor:l1-approximation} bounds the remaining distance
by~$\Pr_p[\Bad]\leq\delta/4$. Altogether,
\begin{align}
  d_{\mathrm{TV}}\bigl(\widehat P(\cdot\mid x),P^{(p)}(\cdot\mid x)\bigr)
  \leq\underbrace{\delta/2}_{\text{tables}}
  +\underbrace{\delta/4}_{\text{noise}}
  +\underbrace{\delta/4}_{\text{width cutoff}}
  =\delta
  \label{eq:compiled-error-budget}
\end{align}
for every~$x$, which is Eq.~\eqref{eq:compiled-total-variation}.

Regard the coins as inputs. In stage~(1), each noise location compares
$s$~coins $T\in\bits^s$ with the fixed threshold~$\kappa=2^sp_s$ of
Lemma~\ref{lem:finite-bit-noise-fixing}. The indicator of $T<\kappa$ is
the OR, over the positions~$i$ at which $\kappa$ has a one, of the AND
of $\neg T_i$ and the literals expressing $T_j=\kappa_j$ for the more
significant positions~$j$; this takes three layers including negations. One further layer of AND gates
forms the Pauli labels, the branch labels of the rectangles, and the cut
indicators. Stage~(2) adds three layers by
Lemma~\ref{lem:lookup-table-sampling}. In parallel, the containment flags
of \routine{BoundaryPass} take two layers after the cut indicators
(proof of Corollary~\ref{cor:combining-boolean-realization}), and the
negation~$\neg\chi_j$ is computed in parallel with the last layer of the
lookups. The selection
$Z_j=(\chi_j\land X_j)\lor(\neg\chi_j\land Y_j)$ then adds two layers. Hence $\cC_{\mathrm{lt}}$ has depth at most
$4+3+2=9$, or~$10$ if the coins are counted as a layer of gates, for all
$d$, $p$, and~$\delta$.

Stage~(1) has size $O(nds)$ and stage~(3) has size~$O(n)$. By
Lemma~\ref{lem:lookup-table-sampling}, stage~(2) consists of $N=O(n)$
circuits of size $O(2^{b_C+r_C}+|C|)$. By
Eqs.~\eqref{eq:compiled-parameters} and~\eqref{eq:compiled-seed-length},
$b_C+r_C=O(dp^{-d}\log(n/\delta)+d^2)$, so that
\begin{align}
  \mathsf{size}(\cC_{\mathrm{lt}})
  =2^{O(d^2)}\,(n/\delta)^{O(dp^{-d})}=n^{O(1)} .
  \label{eq:compiled-size}
\end{align}
By Eq.~\eqref{eq:compiled-precision}, $O(\log(n/\delta))$-bit
approximations to the gate entries and to~$p$ suffice, as assumed in
Section~\ref{sec:boolean-overview}. Given these,
Lemmas~\ref{lem:interval-tables} and~\ref{lem:lookup-table-sampling}
construct all tables and lookup circuits in time polynomial in
$2^{b_C+L_C+r_C}=n^{O(1)}$, and the other two stages are explicit. Thus
$\cC_{\mathrm{lt}}$ can be constructed in polynomial time. This completes
the proof of Corollary~\ref{cor:compiled-constant-depth}.

Without the precision assumption, the same construction with the exact
marginals~$Q_{C,a}$ in place of~$\widetilde Q_{C,a}$ shows that
$\cC_{\mathrm{lt}}$ exists for arbitrary gates, but it need not be
efficiently constructible.

\end{document}